\documentclass[a4paper,11pt]{article}

\usepackage{amsmath,amssymb,amsfonts,mathtools}
\usepackage{xcolor}
\usepackage{hyperref}
\usepackage{indentfirst}
\usepackage{diagbox}
\usepackage[left=1in,right=1in,top=1in,bottom=1in]{geometry}
\usepackage[compact]{titlesec}
\usepackage{paralist}

\usepackage{tikz,amsmath,amsthm,changepage,bm,MnSymbol,amsfonts,url}
\usetikzlibrary{positioning,arrows}
\usetikzlibrary{decorations.markings}

\mathchardef\mhyphen="2D
\usepackage{enumerate}
\usepackage{booktabs,comment}

\usepackage{cleveref}
\newcommand{\F}{\mathbb{F}}%
\newcommand{\C}{\mathbb{C}}%
\newcommand\bigzero{\makebox(0,0){\text{\large0}}}

\newcommand{\srk}{\operatorname{srk}}

\numberwithin{equation}{section}
\usepackage{amsthm}
\theoremstyle{plain}
\newtheorem{theorem}[equation]{Theorem}
\newtheorem{lemma}[equation]{Lemma}

\newtheorem{corollary}[equation]{Corollary}
\newtheorem{conjecture}[equation]{Conjecture}
\theoremstyle{definition}
\newtheorem{definition}[equation]{Definition}
\newcommand{\theproblemname}{}
\newtheorem*{theproblem}{Problem \theproblemname}
\newenvironment{problem}[1]
  {\renewcommand{\theproblemname}{#1}\begin{theproblem}}
  {\end{theproblem}}
\theoremstyle{remark}
\newtheorem*{remark}{Remark}

\newcommand{\problemname}[1]{\textnormal{\textsc{#1}}}

\newcommand{\rk}{\operatorname{rk}}
\newcommand{\id}{\operatorname{id}}
\newcommand{\trans}{\mathsf{T}}
\newcommand{\stab}{\operatorname{stab}}
\newcommand{\mult}{\operatorname{mult}}
\newcommand{\diag}{\operatorname{diag}}
\newcommand{\sgn}{\operatorname{sgn}}
\newcommand{\LatRect}{\textup{LatRect}}
\newcommand{\coldet}{\operatorname{coldet}}

\newcommand{\NP}{\mathsf{NP}}
\newcommand{\coNP}{\mathsf{coNP}}

\newcommand{\BPP}{\mathsf{BPP}}
\newcommand{\VP}{\mathsf{VP}}
\newcommand{\VDET}{\mathsf{VDET}}
\newcommand{\VNP}{\mathsf{VNP}}

\newcommand{\bbP}{\mathbb{P}}
\newcommand{\bbQ}{\mathbb{Q}}
\newcommand{\bbR}{\mathbb{R}}
\newcommand{\bbC}{\mathbb{C}}

\newcommand{\GL}{\mathsf{GL}}
\newcommand{\SL}{\mathsf{SL}}
\newcommand{\Sym}{\mathsf{Sym}}
\newcommand{\lin}[1]{\langle #1 \rangle}

\newcommand{\poly}{\operatorname{poly}}

\newcommand{\cC}{\mathcal{C}}
\newcommand{\cD}{\mathcal{D}}
\newcommand{\cP}{\mathcal{P}}
\newcommand{\cA}{\mathcal{A}}
\newcommand{\cX}{\mathcal{X}}
\newcommand{\cPX}{\mathbb{P}\mathcal{X}}
\newcommand{\cM}{\mathcal{M}}
\newcommand{\cPM}{\mathbb{P}\mathcal{M}}
\newcommand{\aS}{\mathfrak{S}}
\newcommand{\la}{\lambda}
\newcommand{\HWV}{\mathrm{HWV}}
\renewcommand{\det}{\mathrm{det}}
\newcommand{\per}{\mathrm{per}}
\newcommand{\Det}{\mathrm{Det}}

\newcommand{\HMinRk}{\problemname{HMinRank}}
\newcommand{\HMinRkU}{\problemname{HMinRank1}}
\newcommand{\HQSAT}{\problemname{HQuad}}

\author{Markus Bl\"{a}ser\thanks{Saarland University, mblaeser@cs.uni-saarland.de} 
\and Christian Ikenmeyer\thanks{University of Liverpool, christian.ikenmeyer@liverpool.ac.uk }
\and Vladimir Lysikov\thanks{University of Copenhagen, vl@math.ku.dk, 
                             supported by VILLUM FONDEN via the QMATH Centre of Excellence under Grant No. 10059} 
\and Anurag Pandey\thanks{Max Planck Institute for Informatics, apandey@mpi-inf.mpg.de}
\and Frank-Olaf Schreyer\thanks{Saarland University, schreyer@math.uni-sb.de}}
\title{Variety Membership Testing, Algebraic Natural Proofs, \\ and Geometric Complexity Theory}

\begin{document}

\maketitle
\thispagestyle{empty}

\begin{abstract}

%\markus{Abstract needs to be redone}

Variety membership testing is a central task in algebraic geometry. ``Given'' a variety $V$ and 
a point $x$ in the ambient space, we want to decide whether $x \in V$. In this paper, we are particularly
interested in the case when $V$ is given as an orbit closure and the ambient space is the set of
all tensors of order three. The border tensor rank can be phrased as such a problem. 
The first variety that we consider is the slice rank variety,
which consists of all $3$-tensors of slice rank at most $r$.
The notion of slice rank was introduced by Tao and has subsequently been used for several combinatorial problems like capsets, sunflower free sets,  tri-colored sum-free sets, and progression-free sets. 
%Complexity theoretically, nothing was known about the slice rank of tensors %prior to this work. 
We show that deciding if a given 3-tensor has slice rank at most $r$ is $\NP$-hard, that is, the membership testing problem for the slice rank variety is $\NP$-hard.
While the slice rank variety is a union of orbit closures,  we define another variety, the minrank 
variety, which can be expressed as a single orbit closure.
The minrank variety is closely related to the generalized matrix completion problem considered in Bl\"aser et al. (STOC, 2018). 
Our next result is the $\NP$-hardness of membership testing in the minrank variety, hence we establish the $\NP$-hardness of the orbit closure containment problem for tensors of order three.

Algebraic natural proofs were recently introduced by Forbes, \mbox{Shpilka} and Volk
(STOC, 2017)
and independently by Grochow, Kumar, Saks and Saraf~(CoRR, abs/1701.01717, 2017)
as an attempt to transfer Razborov and Rudich's famous barrier
result (J. Comput. Syst. Sci., 1997)
for Boolean circuit complexity to
algebraic complexity theory.
Bl\"aser et al. (STOC, 2018) also gave a version of an algebraic natural proof barrier for the matrix completion problem which relies on the hypothesis that $\coNP \subseteq \exists \BPP$. The result implied that constructing equations for the corresponding variety should be hard.  We generalize their approach to work with any family of varieties for which the membership problem is $\NP$-hard and for which we can efficiently generate a dense subset. Therefore, a similar barrier holds for the slice rank variety and the minrank variety, too.
This allows us to set up the slice rank and minrank varieties as a test-bed for geometric complexity theory (GCT), an approach initiated by Mulmuley and Sohoni (J. Comput., 2001) to attack the permanent versus determinant problem.
We determine the stabilizers of the tensors that generate the orbit closures 
of the slice rank varieties and the minrank varieties and prove that these tensors are almost characterized by their symmetries.
We prove several nontrivial equations for both the slice rank and the minrank varieties using different GCT methods. 
Many equations also work in the regime where membership testing in the slice rank or minrank varieties is $\NP$-hard. In particular,
we obtain equations by using succinctly represented large determinants, Koszul-flattenings, and
representation theoretic methods, more precisely, by using highest weight vectors and bounding multiplicities.
We view this as a promising sign that the GCT approach might indeed be successful.

\end{abstract}

\clearpage
\setcounter{page}{1}

\section{Introduction}

\subsection{Testing membership in algebraic varieties}

Testing whether a point lies in an algebraic variety is a fundamental problem in 
algebraic geometry. ``Given'' a variety $V$ and a point $x$ in the ambient space,
our task is to decide whether $x \in V$ or not. The complexity of this task
depends on how the variety $V$ is given. One natural way of representing a variety is a tuple
of circuits $C_1,\dots,C_m$ computing a set of defining polynomials $f_1,\dots,f_m$ for $V$,
that is, $V$ is the set of common zeros of $f_1,\dots,f_m$. In this case, the problem turns out
to be easy; it is deterministically polynomial-time equivalent to the arithmetic circuit identity
testing problem. For the one direction, note that $x \in V$ iff $f_1(x) z_1 + \dots + f_m(x) z_m$ 
is identically zero. Here, $z_1,\dots,z_m$ are new variables. For the other direction,
we use the fact that general arithmetic circuit identity testing can be reduced to the
case when the circuit computes a constant (that is, degree zero) polynomial \cite{DBLP:journals/siamcomp/AllenderBKM09}.
In this proof, the polynomial is transformed by a Kronecker substitution and then
evaluated at a point of the form $(B,\dots,B)$. Therefore, arithmetic circuit identity testing
even reduces to the case when $V$ is a hypersurface.
 
This shows that when $V$ is given by circuits, the membership problem is easy. A complicated variety will have
a large circuit and therefore, we have more computation time for deciding whether $x \in V$.
However, often we do not know a set of defining equations explicitly. While the computation
is possible in principle, for instance by Gr\"obner bases, it is very costly.
Therefore, we can think of other ways to represent varieties. An obvious way is
to encode the variety explicitly in the problem. One prominent example is the border tensor rank problem.
We are given a tensor $t \in K^{n \times n \times n}$ and want to know whether its border rank
is at most $r$. The tensors in $K^{n \times n \times n}$ of border rank $\le r$ form a variety.

More general is the problem when the variety is given as an \emph{orbit closure}.
Here we have a group $G$ that acts on the ambient space $S$, that is we have a mapping
$\cdot: G \times S \to S$ that satisfies the axioms $1 \cdot s = s$ and $(gh) \cdot s = g\cdot (h \cdot s)$
for all $s \in S$ and $g,h \in G$. Here $gh$ is the group operation. 
Let $\GL_n$ denote the group of all invertible $n \times n$ matrices. $\GL_n$ acts on $K^n$ by
the usual matrix-vector multiplication. $G_n := \GL_n \times \GL_n \times \GL_n$ 
acts on rank-one tensors $u \otimes v \otimes w$ by $(A,B,C) \cdot u \otimes v \otimes w = 
Au \otimes Bv \otimes Cw$ and on arbitrary tensors by linear continuation. 
The orbit of a tensor $t$ under $G_n$ is the set $G_n t := \{ g \cdot t \mid g \in G_n \}$
and its orbit closure is the closure $\overline{G_n t}$ in the Zariski topology.
It is well known that the variety of all tensors of border rank $\le r$ can be 
written with the help of an orbit closure \cite{BCS}, namely $\overline{G_r e_r}$ where $e_r$ is the so-called
unit tensor in $K^{r \times r \times r}$: A tensor $t \in K^{n \times n \times n}$
has border rank $\le r$ iff $\tilde t \in \overline{G_r e_r}$, where $\tilde t$ is an embedding of $t$
into the larger ambient space $K^{r \times r \times r}$.

Orbit closure problems have played a central role in algebraic complexity theory in the recent
years. Not only the border rank problem can be phrased as an orbit closure, but also
the famous permanent versus determinant problem. This is the starting point of the 
geometric complexity program initiated by Mulmuley and Sohoni, see Section~\ref{sec:gct}.

We can think of a tensor $t \in K^{n \times n \times m}$ as a set of 
$m$ matrices $A_1,\dots,A_m$ of size $n \times n$, stacked up on top of each other (also called slices). 
The group $\Gamma_n := \SL_n \times \SL_n$ acts
on $t$ by simultaneously multiplying each of the matrices from the left and the right.
B\"urgin and Draisma \cite{bd} showed that the \emph{noncommutative} rank of the matrix
space given by $A_1,\dots,A_m$ is maximal iff $0 \in \overline{\Gamma_n t}$.
(All such tensors $t$ are said to lie in the \emph{null cone}.) 
Garg et al.~\cite{DBLP:conf/focs/GargGOW16} show how to decide the null-cone problem in this setting in polynomial time,
hence giving a deterministic noncommutative identity testing algorithm. (Unfortunately,
we do not know whether something similar can be achieved in the commutative setting.
More unfortunately, Makam and Wigderson proved recently that the commutative case cannot be written as a null-cone problem \cite{makam2019singular}.)

While the complexity of the border rank is still unknown to our best knowledge,
the fact that the tensor rank problem is $\NP$-hard might be seen as an indication
that the border rank problem is hard, too. The noncommutative identity testing problem
however is easy. Of course, we have different group actions in these two problems.
Furthermore, in the border rank problem, the vector on the right hand side is essentially fixed,
namely to $e_r$, whereas in the latter problem, the vector on the left hand side is fixed,
namely it is zero. As our first contribution, we settle the complexity of testing whether
a tensor $t$ lies in the orbit closure of another tensor $t'$ under the group action
$\GL_k \times \GL_m \times \GL_n$. Note that when $t$ lies in the closure of $t'$, then
the whole orbit closure of $t$ is contained in the closure of $t'$. Therefore, we refer 
to this problem as the \emph{orbit closure containment problem}. We prove that the orbit closure containment problem
is $\NP$-hard for tensors under the $\GL_k \times \GL_m \times \GL_n$ action by defining a quantity
called minrank (see Sections~\ref{sec:homogminrank} and~\ref{sec:geo}). We prove that deciding whether the minrank is bounded by some given bound $b$ is an $\NP$-hard question (see Section~\ref{sec:hard}) and
furthermore, that this question can be phrased as an orbit closure containment problem. 
We also study another quantity, the so-called \emph{slice rank}. The slice rank
was introduced recently, in the proof of the capset conjecture. The tensors of slice rank
bounded by $r$ form a variety, too. Its structure seems to be more complicated, we prove
that it is the union of polynomially many orbit closures. We show that the membership
problem for the slice rank variety is $\NP$-hard, too.

If the orbit closure of $t$ is not contained in the orbit closure of $t'$, then there is a polynomial $f$ that vanishes
on the orbit closure of $t'$ but $f(t) \not= 0$. Such an $f$ is a proof that $t$ is not contained 
in the orbit closure of $t'$. 
Now let $t'$ be a tensor such that membership testing in the orbit closure of $t'$
is NP-hard, for instance, $t'$ could come from a sequence of tensors that generate the minrank varieties.
We will prove that unless the polynomial time hierarchy collapses,
not all such $f$ can have polynomial size algebraic circuits.
This can be viewed as an instance of the algebraic natural proofs framework,
introduced by Forbes et al.~as well as Grochow et al., see Section~\ref{sec:anp}.

When such an $f$ has superpolynomial circuit size, this is an indication that
proving that $t$ is not contained in the orbit closure of $t'$ might be hard.
However, when we want to separate the permanent from the determinant, we need to 
prove a statement like this. (Note however, that we currently do not know whether it is
hard to test whether a polynomial lies in the orbit closure of the determinant,
this is an algebraic variant of the so-called minimum circuit size problem.)
In the third part of this paper we investigate how methods from geometric complexity
theory might overcome this barrier by constructing equations for the slice rank variety which is a union of orbit closures and the minrank variety which in fact is an orbit closure.

\subsection{Slice rank problem}

The notion of slice rank was first used implicitly by Croot, Lev, and Pach in their application of the so-called polynomial method in their breakthrough work on progression-free sets, also known as capsets \cite{croot17}. Later Tao \cite{tao16} gave a symmetrized formulation of this method and used slice rank explicitly. In \cite{BCCGNSU17}, Blasiak et al.~used the term slice rank for the notion that Tao introduced. They used this notion to extend the results on capsets and obtained some barrier results on the group-theoretic approach to the matrix multiplication. Further Tao and Sawin \cite{ts16} explored slice rank of tensors systematically. The methods based on slice rank have been very useful in advancement of several combinatorial problems like the sunflowers free sets, the tri-colored and multi-colored sum-free sets, the capsets and the progression-free problem, and multiplicative matching in nonabelian groups (see for instance  \cite{ellenberg17, naslund17, lovasz19, sawin18}). 

We describe the notion of slice rank and then the corresponding computational problem. For this, we consider the space $V_1 \otimes V_2 \otimes V_3$. It can also be written as ${\bigotimes_{i=1}^3 V_i}$, and is generated by the decomposable (also called rank-one) tensors $v_1 \otimes v_2 \otimes v_3$, where $v_i \in V_i$. 
The usual tensor rank is the minimum number of decomposable tensors that is needed to write
a given tensor as a sum of decomposable tensors. The slice rank is defined in a similar manner,
however, the basic building blocks are not decomposable tensors but tensors that can be
decomposed into a matrix and a single vector. More formally,
consider the smaller tensor products ${\bigotimes_{1 \leq i \leq 3: i \neq j} V_i}$ and the $j$-th tensor products $\otimes_j: V_j \times \bigotimes_{1 \leq i \leq 3: i \neq j} V_i \rightarrow \bigotimes_{i=1}^3 V_i$ with its natural definition.
Now the rank one functions are the elements of the form ${v_j \otimes_j v_{\hat j}}$ for some  ${v_j \in V_j}$ and ${v_{\hat j} \in \bigotimes_{1 \leq i \leq 3: i \neq j} V_i}$. 
The slice rank (or $\srk$ for short) of a tensor $T \in {\bigotimes_{i=1}^3 V_i}$ is the smallest nonnegative integer ${r}$ such that $T$ can be expressed as a linear combination of ${r}$ rank one functions. 
%In this article $V_1 = V_2 = V_3 = \F^{n}$. 
For its comparison with other notions of rank of tensors, like subrank and multi-slice rank, see \cite[ Section 5]{christiandl18}. For its relation to the analytic rank and the partition rank, see \cite{Lovett19}. For its connection to the null cone problem of group actions, see \cite{BGOWW18, BCCGNSU17}.

The slice rank problem is the following.

\begin{problem}
1We are given $T \in \F^n \otimes \F^n \otimes \F^n$ and a number $r$, and we want to know whether $\srk(T) \leq r$.
\end{problem}

Prior to this work, nothing about the complexity status of the problem was known to our best knowledge.

\subsection{Matrix completion and minrank problems}

An instance of a matrix completion problem over some field $K$ is
an $n \times n$-matrix $A$ that is filled with elements from $K$ or
with a special symbol $*$. One can think of the $*$'s as placeholders
that can be replaced by arbitrary elements from $K$. The goal is to replace
the $*$'s in such a way that the rank of the resulting matrix
is either minimized or maximized, depending on the application. 

Matrix completion has many applications, for instance, in machine learning and network coding,
we here just refer to 
\cite{DBLP:journals/combinatorica/Peeters96, DBLP:conf/soda/HarveyKY06, DBLP:conf/colt/HardtMRW14},
which contain relevant hardness results. 
When we consider minimization, the problem is $\NP$-hard, even when the
resulting matrix has rank $3$ \cite{DBLP:journals/combinatorica/Peeters96}. When we consider maximization,
then the problem is $\NP$-hard over finite fields \cite{DBLP:conf/soda/HarveyKY06}. 
Over large enough fields,
there is a simple randomized polynomial time algorithm that simply
works by plugging in random elements from a large enough set.
The correctness of this algorithm follows from the well-known Schwartz-Zippel lemma.
%Derandomising  this matrix completion algorithm is a major open problem, however, quite
%recently a deterministic quasi-polynomial time (even quasi-NC) algorithm was given by
%Gurjar and Thierauf \cite{DBLP:conf/stoc/GurjarT17}.

We can phrase the matrix completion problem as a problem on tensors or on
tuples of matrices.
Let $E_{i,j} \in K^{n \times n}$ be the matrix that has a $1$ in position 
$(i,j)$ and zeros elsewhere. Let $A_0$ be the matrix that is obtained
from $A$ by replacing every $*$ by a $0$. For every star, we create
a matrix $E_{i,j}$ where $(i,j)$ is the position of the $*$.
Let $F_1,\dots,F_m$ be the resulting matrices.
We can view $(A_0,F_1,\dots,F_m)$ as a tensor in $K^{n \times n \times (m+1)}$.
We call $A_0,F_1,\dots,F_m$ the slices of this tensor.
Then the matrix completion problem can be phrased as follows: Find the minimum $r$
such that there are $\lambda_1,\dots,\lambda_m \in K$ fulfilling
\[
   \rk(A_0 + \lambda_1 F_1 + \dots + \lambda_m F_m) \le r.
\]
Here $\rk$ denotes the usual matrix rank.
Many variants of matrix completion have been studied in the literature.
For instance, instead of having simply $*$'s we can have variables and
each occurrence of a variable has to be replaced by the same value. This can naturally be
modeled as a tensor problem, too: Each of the $F_i$ will have a $1$ at each position
where a particular variable occurs and $0$'s elsewhere. The most general setting would
be the following: Given a tensor $t$ as a tuple of $n \times n$-matrices $(A_0,A_1,\dots,A_m)$, what is
the minimum $r$ such that there are $\lambda_1,\dots,\lambda_m$ 
with
\begin{equation} \label{eq:affine}
    rk(A_0 + \lambda_1 A_1 + \dots + \lambda_m A_m) \le r.
\end{equation}
We call this problem a generalized matrix completion problem and we call
the minimum value $r$ above the \emph{completion rank} of $t$.

In \cite{DBLP:conf/stoc/BlaserIJL18}, it is shown that given $t$ and a bound $r$, deciding whether
the completion rank of $t$ is bounded by $r$ is $\NP$-hard. Furthermore---and this is the interesting case here---even testing
whether $t$ is in the \emph{algebraic closure} of the set of all tensors of completion rank $\le r$
is $\NP$-hard. The smallest $r$ such that this is the case, is called the \emph{border completion rank}.
This makes the class of all tensors of border completion rank bounded by some number $r$ an interesting test case
for algebraic natural proofs and methods from geometric complexity theory.

When we want to address this problem with methods from geometric complexity theory,
it is unsatisfactory that in (\ref{eq:affine}), we have an affine matrix pencil, that it,
the matrix $A_0$ is always contained in the linear combination. It would be much more natural
to view this as a problem in projective space, that is, we allow arbitrary nonzero linear combinations. 
We call this measure the \emph{minrank}, since it is the smallest rank of any nonzero matrix that is contained in the linear
span of the slices of the tensor.
%in particular, when viewing the problem as a test case for geometric complexity theory method.
In the hardness proofs in \cite{DBLP:conf/stoc/BlaserIJL18},
it is crucially used that $A_0$ always has unbounded rank whereas all other $A_i$ always have constant rank.
Therefore, the hardness proofs for completion rank do not transfer to minrank.

%As the first of our main contributions, we show that even the projective question is $\NP$-hard,
%that is, we study the homogeneous version  
%\[ 
%   rk(\lambda_0 A_0 + \lambda_1 A_1 + \dots + \lambda_m A_m) \le r.
%\]
%Of course, here the $\lambda_i$ cannot be all zero. 

\subsection{Algebraic natural proofs} \label{sec:anp}

Algebraic natural proofs were introduced by Forbes, Shpilka, and Volk~\cite{DBLP:conf/stoc/ForbesSV17}
and independently by Grochow, Kumar, Saks, and Saraf~\cite{DBLP:journals/corr/Grochow0SS17} (see also \cite{ADblog,AD})
as an attempt to transfer Razborov and Rudich's famous barrier
result \cite{DBLP:journals/jcss/RazborovR97} for Boolean circuit complexity to
algebraic complexity theory.

Let $X$ be a set of indeterminates.
We fix a set of monomials $\cM \subseteq K[X]$ and we consider the linear
span $\lin \cM$ of $\cM$ in $K[X]$. Every polynomial in $\lin M$
is of the form $\sum_{m \in \cM} c_m m$. Every $f \in \lin M$ is identified
with its list of coefficients $(c_m)_{m \in \cM}$. We consider a class
$\cC \subseteq \lin M$. Think of $\cC$ as the polynomials of ``low'' complexity
in $\lin M$. An \emph{algebraic proof} or \emph{distinguisher} is a nonzero polynomial 
$D$ in $|\cM|$ variables $T_m$, $m \in M$, that vanishes on the coefficient vectors
of all polynomials in $\cC$. If for $f \in \lin M$, $D(f) \not= 0$, then
$D$ proves that $f$ is not in $\cC$, that is, $f$ has ``high'' complexity.

\begin{definition}[Algebraic Natural Proofs \cite{DBLP:conf/stoc/ForbesSV17,DBLP:journals/corr/Grochow0SS17}]
Let $X$ be a set of variables and let $M \subseteq K[X]$ be a set of monomials.
Let $\cC \subseteq \lin M$ be a set of polynomials 
and let $\cD \subseteq K[T_m: m \in M]$.  

A polynomial $D$ is an \emph{algebraic $\cD$-natural proof against $\cC$}, if 
\begin{compactenum}
\item $D \in \cD$,
\item $D$ is a nonzero polynomial, and
\item for all $f \in \cC$, $D(f) = 0$, that is, $D$ vanishes on the coefficient
vectors of all polynomials in $\cC$.
\end{compactenum}
\end{definition}

Furthermore, for $f_0 \in \lin M$, 
we call $D$ as above an algebraic $\cD$-natural proof for $f_0$ against $\cC$,
if we have $D(f_0) \not= 0$. That is, $D$ proves that $f_0$ is not in $\cC$.

A \emph{hitting set} for some class of polynomials $\cP$ in $\mu$ variables
is a set of vectors $H \subseteq K^{\mu}$ such that for all $p \in \cP$,
there is an $h \in H$ such that $p(h) \not= 0$. Forbes et al.\
as well as Grochow et al.\ go on and define $\cC$-succinct hitting sets where $\cC$ is some class of polynomials. 
Their main barrier result is that there are either algebraic $\cD$-natural proofs against $\cC$
or $\cC$-succinct hitting sets for $\cD$.

Maybe the most interesting example is when $\cC$ is the class of polynomials
in $n$ variables that have degree $\poly(n)$ and circuit size $\poly(n)$,
that is, we get the class $\VP$ when we run over all $n$. If a polynomial
vanishes on a particular set, it also vanishes on the Zariski closure of this set. 
So an algebraic proof against some class $\cC$ will vanish on polynomials $f$
that are not contained in $\cC$, but are contained in the closure $\overline{\cC}$.
Polynomials in the border $\overline{\cC} \setminus \cC$ may have higher complexity
than polynomials in $\cC$ (otherwise, they would be in $\cC$), yet they cannot
be distinguished by an algebraic proof from polynomials in $\cC$, independently
of any barrier. Therefore, to study algebraic proofs properly, one needs to look at Zariski closed  classes of polynomials. 
%It is important to remark that the complexity of polynomials in the border $\overline{\cC} \setminus \cC$ 
%may still be polynomially bounded in the complexity of $\cC$. 

The setting above is not only limited to the class $\VP$,
we can for instance also consider tensors of order three, that is, trilinear forms.
In this case, instead of circuit complexity, we study for instance the border rank
of tensors. While the complexity of the border rank is still open (note however, that testing the
rank is a hard problem 
\cite{DBLP:Hastad, shitov, DBLP:journals/corr/SchaeferS16, DBLP:conf/stoc/BlaserIJL18, DBLP:conf/approx/Swernofsky18}), 
Bl\"aser et al. \cite{DBLP:conf/stoc/BlaserIJL18} defined
a related measure, the so-called border completion rank.
They proved that border completion rank is $\NP$-hard. From this it follows that from any set of equations
defining the variety of tensors of border completion rank bounded by a certain value at least one
of the equations has superpolynomial algebraic circuit complexity, unless $\coNP \subseteq \exists\BPP$. 

Even more general, we can consider this problem for any variety $V$. In our setting, $V$ would be $\overline{\VP}$
or the orbit closure of the determinant or the variety of all tensors with border rank bounded by some $r$. 
An equation $f$ of this variety
can be considered as an algebraic proof: If $f(x) \not= 0$, then $f$ is a proof that $x$
is not in $V$. Of course, if $V$ is arbitrarily complex, then the complexity of $f$ can be of course
arbitrarily high, therefore, we are interested when $V$ can be easily described,
for instance as the closure of objects (polynomials, tensors, \dots) of low complexity.

\subsection{Geometric Complexity Theory} \label{sec:gct}
In \cite{Val:79} Valiant proved that every polynomial $f$ can be written as the determinant of a matrix whose entries are affine linear polynomials.
The required matrix size to write $f$ in this way is called the \emph{determinantal complexity} $\text{dc}(f)$.
The flagship conjecture in algebraic complexity theory is that the sequence $\text{dc}(\per_m)$ is not polynomially bounded,
where $\per_m := \sum_{\pi \in \aS_m}\prod_{i=1}^m x_{i,\pi(i)}$ is the permanent polynomial.
In terms of algebraic complexity classes, this can be succinctly phrased as $\VDET \neq \VNP$.
Mulmuley and Sohoni \cite{gct1,gct2} proposed to reinterpret Valiant's determinant versus permanent conjecture in terms of questions about
certain orbit closures and the representations in their coordinate rings.
They arrive at the potentially stronger conjecture $\overline{\VDET} \neq \overline{\VNP}$
and coined the name \emph{geometric complexity theory} (GCT) for their approach.
At the center of their attention is the study of the orbit closure $\Det_n := \overline{\GL_{n^2}\det_n}$,
as it allows us to define the \emph{border determinantal complexity} $\underline{\text{dc}}(f)$ to be the smallest $n$ such that the padded polynomial $x_{1,1}^{n-m}f$
is contained in $\overline{\GL_{n^2}\det_n}$. The conjecture $\overline{\VDET} \neq \overline{\VNP}$
is equivalent to $\underline{\text{dc}}(\per_m)$ growing superpolynomially.
The definition of $\underline{\text{dc}}$ from $\text{dc}$ is in complete analogy to going from tensor rank to the border rank of tensors, see e.g.\ \cite{Bin:80}.

To prove that a point $p$ does not lie in $\Det_n$ one searches for polynomials vanishing on $\Det_n$, that is, \emph{equations for the variety $\Det_n$}.
Those which do not vanish on $p$ are sometimes called \emph{separating polynomials} in the GCT literature, as they prove $p \notin \Det_n$.
In the language of the previous section, they would be called algebraic proofs or distinguishers.
Sometimes representation theory can be used to find equations for varieties:
Since $\Det_n$ is closed under the action of $\GL_{n^2}$, the vanishing ideal $I(\Det_n)$ decomposes into a direct sum of irreducibles in each degree.
If $\la$ is an $\GL_{n^2}$-isomorphism type and $a_\la$ is the multiplicity in the coordinate ring of the ambient space $\bbC[x_{1,1},\ldots,x_{n,n}]_n$,
then the algebraic Peter-Weyl theorem implies that an equation of type $\la$ exists if the multiplicity of $\la$ in the coordinate ring of the orbit $\GL_{n^2}\det_n$
is less than $a_\la$.
This multiplicity if known as the \emph{rectangular symmetric Kronecker coefficient}, see \cite{BLMW:11}.
This criterion is satisfied in numerous cases, see the appendix of \cite{Ike:12b}.

The GCT approach is very general and can be applied to numerous algebraic measures of complexity.
As one example, the border rank of the matrix multiplication tensor was phrased in this setup \cite{BI:11, BI:13} and explicit lower bounds for this border rank have
been found using the multiplicities in the coordinate ring of the $\GL_n \times \GL_n \times \GL_n$ orbit of the unit tensor.
These were the first lower bounds in algebraic complexity theory found using this approach.
Since our space of objects is a tensor space and the action a product of general linear groups, the multiplicities in the coordinate ring of this ambient space
are given by the infamous\footnote{``Although Kronecker coefficients are a classical subject, frustratingly little is known about them.'' \cite{Bur}} Kronecker coefficients.

\section{Our contributions}

%\markus{Orbit closure containment missing?} No

We give a detailed overview of our results and its meaning for geometric complexity theory. 
We keep this exposition at a non-expert level, the complete results can be found in the subsequent sections.

\subsection{The slice rank problem and orbit closures}

Our first contribution is that we show that the slice rank problem is $\NP$-hard
under polynomial time many-one reductions.
(see Section \ref{sec:hardslices}). For this, we use a connection of the slice rank to the size of a minimum vertex cover of a hypergraph by Tao and Sawin \cite{ts16}. They showed that for every $3$-uniform, $3$-partite hypergraph $H$, one can associate a tensor $T_{H}$, and if the edge set of the hypergraph forms an antichain, then the slice rank of the associated tensor $T_{H}$ equals the size of the minimum vertex cover of the hypergraph $H$. To our best knowledge, the complexity of the decision version of the slice rank problem
for order-three tensors has been open so far.
Prahladh Harsha, Aditya Potukuchi, and Srikanth Srinivasan kindly sent us an unpublished
manuscript, in which they prove that the order-four case is $\NP$-hard. However, this one
more tensor leg gives an additional degree of freedom, which easily allows to establish
the antichain condition. B\"urgisser et 
al.~\cite[page 27]{brgisser2017alternating}
report that Sawin has an unpublished proof that \emph{computing} the slice-rank of tensors of order three is $\NP$-hard. 
However, they also state that the decision version is open.

%\markus{Change made. Please check!} Done

We show the $\NP$-hardness of the slice rank problem for order-three tensors by showing that the $3$-uniform, $3$-partite hypergraph minimum vertex cover problem where the edge set forms an antichain is $\NP$-hard. The corresponding hypergraph minimum vertex cover problem without the antichain restriction is known to be $\NP$-hard \cite{gottlob10} by reduction from the usual 3-SAT problem. However, their reduction does not work if one wants to adapt it to the antichain restriction. We use a reduction from a restricted SAT-variant, the bounded-occurrence mixed SAT (bom-SAT) problem, in which there are 3-clauses and 2-clauses, and every variable occurs exactly thrice, once in a 3-clause and twice in 2-clauses. Because of the antichain restriction, our labelling of the gadget becomes very delicate and needs to be handled very carefully in the reduction (see Lemma \ref{lem:slicelabel}).

Next, we phrase the slice rank problem in terms of orbit closures. More specifically, we show that testing whether a tensor $T \in \F^{n \times n \times n}$ has $\srk(T) \leq r$ is equivalent to testing if the tensor $T$ is contained in a polynomially large union of orbit closures. Let $(r_1, r_2, r_3)$ be such that $r_1 + r_2 + r_3 = r$. We first embed $T$ in a larger subspace $ U' \otimes V' \otimes W' \cong \F^{s_1} \otimes \F^{s_2} \otimes \F^{s_3}$ (this is called padding), where $s_1 = r_1 + nr_2 +nr_3$, $s_2 =  nr_1 + r_2 + nr_3$ and $s_3 = nr_1 + nr_2 + r_3$, and define  
\[ 
S_{n,r_1,r_2,r_3} = \sum_{i =1}^{r_1} \sum_{j=1}^{n}e^1_i \otimes e^1_{ij}\otimes e^1_{ij} + \sum_{i =1}^{r_2} \sum_{j=1}^{n}e^2_{ij} \otimes e^2_i\otimes e^2_{ij} + \sum_{i =1}^{r_3} \sum_{j=1}^{n}e^3_{ij} \otimes e^3_{ij}\otimes e^3_i.
\]
Intuitively, in the sum above, we have $r_1$ rank-one elements of the form ${v_1 \otimes_1 v_{\hat 1}}$ with ${v_1 \in V_1}$ and ${v_{\hat 1} \in \bigotimes_{1 \leq i \leq 3: i \neq 1} V_i}$,
$r_2$ elements of the form $v_2 \otimes_2 v_{\hat 2}$, and $r_3$ elements of the form $v_3 \otimes_3 v_{\hat 3}$. 
Now $\srk(T) \leq r$ becomes equivalent to testing whether $T$ is in the orbit closure of the $S_{n,r_1,r_2,r_3}$ for some $(r_1, r_2, r_3)$ with $r_1 + r_2 + r_3 = r$. 
Thus we show that the slice rank variety $\mathcal{SV}_{\mathbb{F}^n \otimes \mathbb{F}^n \otimes \mathbb{F}^n, r}$ is the union of orbit closure of $S_{n,r_1,r_2,r_3}$ over all $(r_1, r_2, r_3)$ with $r_1 + r_2 + r_3 = r$, intersected with the ambient space $\F^{n} \otimes \F^{n} \otimes \F^{n}$,
see Section~\ref{sec:srk:orbit} for details.
Note that Tao showed that the set of all $T$ with $\srk(T) \leq r$ is closed,
so there is no need to define a notion of border slice rank (see \cite[Corollary 2]{ts16}).
This is different to the situation with determinantal complexity and border determinantal complexity or tensor rank and border rank.

Next we go on to determine the stabilizer of $S_{n, r_1, r_2, r_3}$, i.e., the subgroup of $\GL(U') \times \GL(V') \times \GL(W')$ which fixes $S_{n,r_1, r_2, r_3}$ (Theorem~\ref{slicerank:thm:stab}). We can also show that each $S_{n,r_1, r_2, r_3}$ is almost characterized by its stabilizer, i.e., it is a direct sum of three tensors that are each characterized by their respective stabilizers (Theorem \ref{thm:srk:sym}). This is an important property in the context of geometric complexity theory.
Both the permanent and the determinant are characterized by their respective stabilizers as well.

Phrasing the problem geometrically allows us to find equations for the slice rank varieties. This makes the slice rank problem an interesting ``testing ground'' for the methods of geometric complexity theory.
The situation is very similar to the permanent-determinant and border rank settings that have been studied in geometric complexity theory; 
the symmetries of the tensors determine the orbits.
This allows us to analyze the problem using representation-theoretic methods.
What makes our ``testing ground'' very appealing from a complexity theoretic point of view is the fact
that we can prove that testing containment in the slice rank varieties is $\NP$-hard, something which we
do not know for $\overline{\VP}$, the orbit closure of the determinant, or tensors of a given border rank.
This hardness allows us to reason about proof barriers.

\subsection{The minrank problem and orbit closures}

As a second test case, we study the minrank problem as a test case for the geometric complexity methodology. In contrast to the slice rank problem, the corresponding variety can
be written as a \emph{single} orbit closure. Since we also prove that containment in the minrank variety 
is $\NP$-hard, we obtain the result that the orbit closure containment problem is $\NP$-hard (Corollary \ref{cor:occhard}),
which we cannot deduce from the hardness of slice rank.

For the minrank problem, we are given a tuple of matrices $A_1, \dots, A_k$ of the same size $m \times n$ and a 
number $r$, and the problem is to decide whether 
there exist a nonzero linear combination $x_1 A_1 + \dots + x_k A_k$ with rank at most $r$.
The smallest $r$ is called the \emph{minrank} of $A_1,\dots,A_k$. Instead of thinking
of a tuple of matrices, we can also view $A_1,\dots,A_k$ as a tensor in $F^{k \times m \times n}$
with $A_1,\dots,A_k$ being its slices. We will use both views in this paper. 

In contrast to the completion rank, we allow any nontrivial linear combination of the slices
in the minrank case, whereas for completion rank, we always require $x_1 = 1$.
We can view the minrank problem as the homogeneous version of the completion rank problem.
As minrank is a homogeneous problem, from an algebraic perspective, it is more natural than the completion rank.
Instead of affine varieties, we obtain a projective variety.
The $\NP$-hardness proofs in \cite{DBLP:conf/stoc/BlaserIJL18} critically used the fact
that $x_1 = 1$, since $A_1$ was a matrix that had rank linear in the input size whereas all
other matrices had the same, constant rank. These hardness proofs do not work in the  homogeneous setting,
since all instances created in the proofs trivially have the same minrank. As one of our
main results, we prove that testing containment in these projective varieties is still $\NP$-hard.

While the minrank problem certainly is an interesting problem on its own right, we consider it here as 
a ``test-bed'' for the geometric complexity approach. To this aim, we again want to show
that we can write the minrank problem as an orbit closure problem. For a tensor $T \in F^{k \times m \times n}$ 
given as $e_1 \otimes A_1 + \dots + e_k \otimes A_k$ and a linear form $x \in (F^k)^*$, we define the contraction
$Tx$ by $Tx := x(e_1) A_1 + \dots + x(e_k) A_k$, that is, we form a linear combination of the slices.
If we take the set of all $(T, x)$ with $\rk(Tx) \le r$ and $x \not= 0$ and project on the first component,
we get all tensors of minrank at most $r$.
Since the set of all such $(T, x)$ is invariant under scaling of $T$ or $x$ by nonzero factors, it also defines a \emph{projective} variety,
and the projection on the first component is a projective variety, too, see Section~\ref{sec:geo} for more 
details.
So we are in the nice situation where the set of all tensors of minrank at most $r$ is Zariski closed (Theorem~\ref{thm:closed}).
This means that we are in the same situation as slice rank;
we do not need an additional border complexity measure, i.\,e., minrank and border minrank coincide.
We denote the corresponding variety of all tensors $T \in U \otimes V \otimes W$
of minrank at most $r$ by $\cM_{U \otimes V \otimes W, r}$ or just $\cM_r$ when the tensor space is clear from context. 

Next, we want to write the minrank varieties $\cM_{U \otimes V \otimes W, r}$ as orbit closures. 
Note that we can always embed a tensor $T \in U \otimes V \otimes W$ into a larger
ambient space $U \otimes L \otimes L$, where $V$ and $W$ are subspaces of $L$,
by filling the new entries with zeros. (This process is called \emph{padding}.)
We then show (Corollary \ref{cor:orbit}), that $\cM_{U \otimes V \otimes W, r}$ is 
the $\GL(U) \times \GL(L) \times \GL(L)$-orbit closure of the tensor
\[
  T_{k, n, r} = e_1 \otimes (\sum_{j = 1}^r e_{1j} \otimes e_{1j}) + \sum_{i = 2}^{k} e_i \otimes (\sum_{j = 1}^n e_{ij} \otimes e_{ij}),
\]
intersected with the ambient space $U \otimes V \otimes W$ (here $k = \dim U$, $n = \dim L$).
This means that we can reduce the question whether a tensor has minrank at most $r$ to the question whether it
is contained in the orbit closure of $T_{k, n, r}$.

Then we go on and determine the so-called stabilizer of $T_{k, n, r}$ (Theorem~\ref{thm:stab}),
a subgroup of $GL(U) \times \GL(L) \times GL(L)$ which fixes $T_{k, n, r}$.
This is the group of ``symmetries'' of $T_{k, n, r}$.
It can be shown that the orbit of $T_{k, n, r}$ is completely described by the stabilizer, 
that is, all tensors having the same stabilizer lie in the orbit of $T_{k, n, r}$ (Theorem~\ref{thm:sym}).
This is an important property in the context of geometric complexity theory, only slightly weaker than
the property of being characterized by the stabilizer up to scale shared by determinant and permanent.

%We think that the homogeneous minrank problem can be used as a ``testing ground'' for the methods of %geometric complexity theory.
%The situation is very similar to the permanent-determinant and border rank settings that have been studied in geometric complexity theory; 
%the symmetries of the tensors uniquely determine the orbits.
%This allows us to analyze the problem using representation-theoretic methods.
%What makes our ``testing ground'' very appealing from a complexity theoretic point of view is the fact
%that we can prove that testing containment in the minrank varieties is $\NP$-hard, something which we
%do not know for $\overline{\VP}$, the orbit closure of the determinant, or tensors of a given border rank.
%This hardness allows us to reason about proof barriers.

\subsection{Equations of slice rank varieties}

In Section~\ref{slicerank:sec:ideal} we describe many nonzero polynomials which vanish on slice rank varieties.
We use two different representation theoretic methods to find these equations, and interestingly both yield the same set of polynomials.
The first method in Section~\ref{subsec:slicerankeqdesign} uses multilinear algebra and highest weight vectors, while the second method in Section~\ref{subsec:slicerankeqmult} uses the stabilizer of the slice rank tensors and invariant theory.
Interestingly, among these equations we find that the unique $\SL_{n^2}\times\SL_{n^2}\times\SL_{n^2}$-invariant function of degree $n^3$ vanishes on tensors of slice rank $< n^2$ in $\bbC^{n^2}\otimes\bbC^{n^2}\otimes\bbC^{n^2}$.
This function is known as one of Cayley's hyperdeterminants. If the combinatorial property in \cite[Cor.~5.25(3)]{BI:17fundinv} is true, then this coincides with the \emph{fundamental invariant of the $n \times n$ matrix multiplication tensor}.

Let $k$, $m$ and $n$ denote the dimensions of $U$, $V$ and $W$ respectively.
Since the slice rank variety $\mathcal{SV}_{U \otimes V \otimes W, r}$ is invariant under the group action of $\GL(U)\times\GL(V)\times\GL(W)$,
the ideal of polynomials vanishing on it is also a representation of $\GL(U) \times \GL(V) \times \GL(W)$.
The irreducible polynomial representations of $\GL(U)$ are indexed by partitions with at most $\dim U$ many parts (a partition is a finite list of nonincreasing natural numbers).
The irreducible polynomial representations of $\GL(U) \times \GL(V) \times \GL(W)$ are indexed by triples of partitions $(\la,\mu,\nu)$ where $\la$ has at most $k$ parts, $\mu$ has at most $m$ parts, and $\nu$ has at most $n$ parts.
The multiplicity with which $(\la,\mu,\nu)$ occurs in the coordinate ring $\bbC[U \otimes V \otimes W]$ is called the \emph{Kronecker coefficient}.
It is nonzero only if $|\la|=|\mu|=|\nu|$,
in which case $D=|\la|$ is the degree of the polynomial.

Let $\ell = \lceil\sqrt{\max\{k,m,n\}}\rceil$.
In Section~\ref{slicerank:sec:ideal} we find that all partition triples that satisfy $\la_1>\ell$ or $\mu_1>\ell$ or $\nu_1>\ell$ vanish on $\mathcal{SV}_{U \otimes V \otimes W, D/\ell}$.

It is intriguing that both constructions in Section~\ref{subsec:slicerankeqdesign} and
Section~\ref{subsec:slicerankeqmult} give the same set of equations. Both approaches are rather indirect, but the multilinear algebra approach in Section~\ref{subsec:slicerankeqdesign} contains a construction principle for the functions,
whereas the approach in Section~\ref{subsec:slicerankeqmult} is purely based on invariant theory and the stabilizers of the slice rank tensors that we determine in Section~\ref{sec:srk:orbit}.
In Section~\ref{subsec:slicerankeqmult} do not write down the polynomials, but we obtain an upper bound on the multiplicities of irreducible representations in the coordinate ring of slice rank varieties.
The upper bound is given by the multiplicities in the coordinate rings of the orbits of the slice rank tensors.
When these multiplicities are all less than the Kronecker coefficients, then we get equations.

The polynomials described in Section~\ref{subsec:slicerankeqdesign} have polynomial degree, but the direct method of evaluation of these polynomials involves exponential sums. We conjecture that it is $\NP$-hard to check vanishing of these polynomials on a given tensor.
The polynomials obtained from comparing multiplicities are given by a type of representation.
It describes polynomials exactly using small size labels, but does not give a direct method of evaluation of these polynomials,
which makes it potentially able to overcome the barrier saying that in some cases algebraic proofs need to be hard.
But it also means we do not have a concrete tensor on which these polynomials do not vanish.
Indeed, the nonvanishing of the hyperdeterminant on the matrix multiplication tensor is an open question posed in \cite[Cor.~5.25(3)]{BI:17fundinv}.
It would be nice to find explicit nontrivial tensors for which we can prove a slice rank lower bounds using multiplicities.

\subsection{Equations of the minrank varieties}\label{subsec:eqforminrank}

Throughout this subsection, $k$, $m$ and $n$ denote the dimensions of $U$, $V$ and $W$ respectively.
In Section~\ref{sec:ideal} we describe several polynomials which vanish on minrank varieties.
Section~\ref{sec:eq-basic} gives a very basic example which follows directly from the definition: existence of a slice $Tx$ of rank at most $r$ means that the $(r + 1) \times (r + 1)$ minors of $Tx$ as polynomials in $x$ have a common zero.
Instead of this, a weaker condition can be checked. The degree $r + 1$ homogeneous polynomial map sending $x$ to the collection of all minors gives rise to a linear map sending $x^{\otimes (r + 1)}$ to the same minors.
The existence of the common zero in this case can be checked simply by checking the rank of this linear map.
Interestingly, in one special case this construction coincides with the construction of the hyperdeterminant~\cite{Gelfand1994}.

Another construction is an application of Koszul flattenings which are also used in the study of border rank of tensors.
The best currently known lower bounds for the border rank of the matrix multiplication tensor are based on Koszul flattenings~\cite{DBLP:journals/toc/LandsbergO15b,landsberg2017lower}
(but it is also known that these bounds cannot be significantly improved~\cite{wigdersonrank}).
A simplest example of this family of equations is the case $k = 3$, $m = n$.
Let $T = e_1 \otimes A_1 + e_2 \otimes A_2 + e_3 \otimes A_3$.
We can form a matrix
\[
\begin{bmatrix}
-A_2 & A_1 & 0 \\
-A_3 & 0 & A_1 \\
0 & -A_3 & A_2
\end{bmatrix}
\]
which is $\GL(U) \times \GL(V) \times \GL(W)$-covariant in the sense that if one tensor is obtained from another by an action of this group,
then the correponding matrices are also equivalent.
If $\rk(A_1) \leq r$, then the rank of this matrix is at most $2m + 2r$, which is less than the maximal possible $3m$ for $r < \frac{m}{2}$.
It is known~\cite{STRASSEN1983645} that for a generic tensor the rank of the above matrix is maximal,
which means that the equations for $\cM_r$ given by the $(2m + 2r + 1) \times (2m + 2r + 1)$ minors of this matrix are nontrivial.
This construction can be generalized. In Section~\ref{sec:eq-koszul} we describe a generalization that gives equations
when $m = \frac{k - p}{p + 1} n$ and $r < \frac{m}{k - p}$ where $p$ is a parameter (Theorem~\ref{thm:koszul}).

More interesting examples are given by representation-theoretic methods.
Since the minrank variety $\cM_r$ is invariant under the group action,
the ideal of polynomials vanishing on it is also a representation of $\GL(U) \times \GL(V) \times \GL(W)$.
In Section~\ref{sec:eq-rect} we describe a family of polynomials which potentially give equations for minrank varieties in cases $m, n > kr$.
They are constructed as specific highest weight vectors of certain $\GL(U) \times \GL(V) \times \GL(W)$-representations.
The nontriviality of these equations is connected to interesting combinatorial questions about Latin rectangles.
These problems arise from the evaluation of the constructed polynomial, which involve exponential sums.
We now give an example of such a degree 6 equation in the case $k = 2$, $n = m = 3$, $r = 2$.
Let $\sgn(w) \in \{-1,1\}$ denote the sign of a permutation $w$ of a set of numbers that start at 1,
e.g., $\sgn(1,3,2) = -1$ and $\sgn(1,2) = 1$.
We define $\sgn(w)=0$ for any other list of numbers $w$, e.g., $\sgn(2,3)=0$ and $\sgn(1,2,2)=0$.
Denote the components of the tensor $T \in \bbC^{2 \times 3 \times 3}$ by $T_{\alpha, \beta, \gamma}$.
Let $R:=\{(i,j)\mid 1\leq i \leq 2, 1 \leq j \leq 3\}$ be a $2 \times 3$ rectangle
and $I:=\{(\alpha, \beta, \gamma) \mid 1 \leq \alpha \leq 2, 1 \leq \beta,\gamma \leq 3\}$ the set of possible indices.
Define
\begin{equation}
\label{eq:hyperdet}\textstyle
f(T) = \sum_{\varphi : R \to I} \sigma(\varphi) \prod_{q \in R} T_{\varphi(q)},
\end{equation}
\vspace{-0.8cm}
\begin{eqnarray*}
\text{where \quad\quad} \sigma(\varphi) &=& 
\sgn\big(\varphi(1,1)_1,\varphi(2,1)_1\big)\cdot\sgn\big(\varphi(1,2)_1,\varphi(2,2)_1\big)\cdot\sgn\big(\varphi(1,3)_1,\varphi(2,3)_1\big) \\
&& \cdot\sgn\big(\varphi(1,1)_2,\varphi(1,2)_2,\varphi(1,3)_2\big)\cdot\sgn\big(\varphi(2,1)_2,\varphi(2,2)_2,\varphi(2,3)_2\big) \\
&& \cdot\sgn\big(\varphi(1,1)_3,\varphi(2,1)_3\big)\cdot\sgn\big(\varphi(1,2)_3,\varphi(2,2)_3\big)\cdot\sgn\big(\varphi(1,3)_3,\varphi(2,3)_3\big).
\end{eqnarray*}
Clearly, $f$ is a homogeneous degree 6 polynomial on $U \otimes V \otimes W$.
It will follow from the general results in Section~\ref{sec:eq-rect} that $f$ is nonzero and that $f$ vanishes on tensors of minrank $\leq 2$.
Moreover, this $f$ is a highest weight vector of weight $(3,3),(2,2,2),(3,3)$.

More indirectly, the existence of equations can be proven by showing that the ideal of $\cM_r$
contains an irreducible representations of a given weight.
We study this approach in Section~\ref{sec:eq-mult}.
There we show that nontrivial equations can be obtained purely from the description of minrank varieties
in terms of orbit closures, \emph{without} presenting the polynomial explicitly and computing its value on tensors from minrank varieties.
This is achieved by obtaining an upper bound on the multiplicities of irreducible representations in the coordinate ring of minrank varieties.
The upper bound is given by the multiplicities in the coordinate ring of the orbit of the tensor describing min-rank.
When this multiplicity is less than the multiplicity in the space of all polynomials (which is given by Kronecker coefficients), then we get an equation.
We determine the exact formula for the upper bound and verify that it indeed yields numerous equations,
for example one can calculate that one of the equations we find in this way is the one in \eqref{eq:hyperdet}.

Computational properties of the described polynomials vary.
Basic polynomials described in Section~\ref{sec:eq-basic} are compositions of determinants where inner determinants are of size $r + 1$ and outer determinants are of size $\binom{k + r}{r + 1}$.
If $r$ is constant (in particular, in the case $r = 1$), these polynomials are easy to compute.
This means that for almost all tensors of minrank greater than $1$ we can easily prove this using these basic equations.
Nevertheless, we will prove that in general this question is $\coNP$-hard and thus we are unlikely to have easy proofs for all tensors.
If, on the other hand, $r$ is linear in the size of the tensor (in this other regime we also prove hardness results) computation involves determinants of exponential degree.
Similarly, the ranks of Koszul flattenings from Section~\ref{sec:eq-koszul} are computable in polynomial time if the parameter $p$ is constant and involve determinants of exponential size if $p$ is linear in the size of the tensor.
The polynomials described in Section~\ref{sec:eq-rect} have polynomial degree, but the direct method of evaluation of these polynomials involves exponential sums. We conjecture that it is $\NP$-hard to check vanishing of these polynomials on a given tensor.
The polynomials obtained from comparing multiplicities are given by a type of representation.
It describes polynomials exactly using small size labels, but does not give a direct method of evaluation of these polynomials,
which makes it potentially able to overcome the barrier saying that in some cases algebraic proofs need to be hard.
But it also means we do not have a concrete tensor on which these polynomials do not vanish.
Of course these equations are nonzero, but we only know that it does not vanish on a tensor
of the form $T_{k,n,r}$ for large enough $r$. It would be nice to find other explicit tensors for which 
we can prove a minrank lower bound using multiplicities.

\subsection{Hardness of membership testing and algebraic natural proofs}

In Section~\ref{sec:hardslices}, we prove that testing membership in the slice rank varieties
is $\NP$-hard and
we prove in Section~\ref{sec:hard} that testing membership in the minrank varieties is $\NP$-hard.
Since the minrank varieties are orbit closures, we get as a corollary, that the orbit closure
containment problem is $\NP$-hard.
It turns out that even deciding whether the minrank is $\le 1$ is already $\NP$-hard.
We can use these hardness results to show the following lower bound for the equations of slice rank and minrank varieties:
For infinitely many $n$, there is an $m$, a tensor $T \in F^{m \times n \times n}$
and a value $r$ such that there is no 
algebraic $\poly(n)$-natural proof for the fact that the slice rank or minrank of $T$ is greater than $r$ 
unless $\coNP \subseteq \exists \BPP$. We prove this by providing a general methodology
for proving statements like this, generalizing results from \cite{DBLP:conf/stoc/BlaserIJL18}.
The two main ingredients needed to achieve this results are 
\begin{itemize}
 \item the $\NP$-hardness of the membership problem of the varieties and
 \item the ability to effectively generate a dense subset of the variety.
\end{itemize}
Here ``effectively generate'' means that we can provide a vector of polynomials, each computed
by a polynomial sized circuit, such that the image of the vector (interpreted as a polynomial map)
lies dense in the variety. Since the minrank varieties are orbit closures, it follows easily that the minrank
varieties satisfy the second property. For the slice rank varieties, this works, too, since
they are a polynomial union of orbit closures.

Now we can go as follows: If we assume that almost all minrank varieties are described by a set of equations
of polynomial size, we can decide non-membership as follows: Given a point $x$, we guess a polynomially sized circuit
for an equation $f$. Since we can effectively generate a dense set, we can check whether $f$ vanishes
on this set and therefore on the whole variety using polynomial identity testing. This can be done 
by an $\exists\BPP$-machine. Then we simply test whether $f(x) \not= 0$ using polynomial identity testing again
and can therefore decide non-membership. Since membership testing is $\NP$-hard, non-membership testing is
$\coNP$-hard and the result follows.

We can interpret this as a barrier result: We proved that we can get equations for the slice rank or minrank
varieties by different methods from geometric complexity. The result above means that in a full set of
equations, that is, for a set of equations that describe the variety completely, not all of them
will have algebraic circuits of polynomial size (unless the polynomial time hierarchy collapses).
Therefore, if we want to prove lower bounds with the GCT approach, we have to argue why 
in our lower bound proof, we do not (implicitly) evaluate the circuit to prove that a point is not
contained in the variety, but we do something more clever.

\subsection{Does GCT avoid natural proofs barriers?}

By the results mentioned in the previous section, if the $(V_n)$-membership problem
is $\NP$-hard for some family $(V_n)$ of varieties, then not all equations of $V_n$ can have polynomial size circuits.
However, in this paper we have constructed various equations for the slice rank and minrank varieties using the GCT methodology,
even in the regime where the membership problem is $\NP$-hard. While we do not know
whether all of the equations have polynomial size circuit---we rather suspect not, 
since they are described by exponential sized determinants or exponential sums---they have polynomial size descriptions
in other models, for instance, they are given by:
\begin{itemize}
\item succinctly represented exponential size determinants,
\item succinctly represented exponential sums, or
\item succinct representation-theoretic objects.
\end{itemize}

But by the end of the day, given a variety $V$ and a point $x$,
the GCT approach produces the description of an equation $f$ of the variety $V$
under consideration, such that $f(x) \not= 0$. The description of this equation
can be very short, as mentioned above.
However, we do not only have to give a description of the equation, but we also 
have to prove that $f(x) \not = 0$ and that $f$ vanishes on the corresponding variety.

More specifically, we have a sequence $(V_n)$ of varieties and a sequence of points $(x_n)$
and we want to prove that $x_n \notin V_n$. We do this by constructing a 
sequence of polynomials $(f_n)$ such that $f_n$ vanishes on $V_n$ and $f_n(x_n) \not= 0$.
What the term ``constructing'' means, that is, how do we represent the polynomial $f_n$
and how do we prove that it vanishes on $V_n$ and $f_n(x_n) \not = 0$, might depend
on our lower bound method. In our example, we can think of $(V_n)$ as being a sequence of minrank
varieties generated by $T_{k(n),n,r(n)}$ where $k(n)$ and $r(n)$ are chosen in
such a way that the membership problem in $(V_n)$ is $\NP$-hard, see Section~\ref{sec:hard}
for various possible choices of parameters. Slice rank is a little bit more complicated,
but essentially the same reasoning works.

From Theorem~\ref{thm:hardproof} it follows that when the $(V_n)$-membership problem 
is $\NP$-hard and $(V_n)$ fulfills the further, very natural prerequisites of the theorem, 
then we cannot expect that all equations of $V_n$ have polynomial circuit size. So the question
is whether we can prove that $f_n$ vanishes on $V_n$ and that $f_n(x_n) \not= 0$ despite
this natural proof barrier. 

One can think of various ways how to circumvent this barrier. Think of $(V_n)$ being a sequence of slice rank or minrank varieties
such that membership testing is hard:

\begin{itemize}
\item First of all, we might be lucky and the equations $f_n$ that we picked for $x_n$ have polynomial size circuits.
We cannot rule out this possibility, however, it seems very unlikely to us that this might actually happen.

\item For all our equations that we constructed in the paper, we were able to show that they vanish
on the corresponding variety. This means that if one of these equations has superpolynomial circuit size---what we consider
to be very likely---all our proofs in Sections~\ref{slicerank:sec:ideal} and~\ref{sec:ideal} cannot rely on evaluating this underlying circuit on a generic point of the variety,
even not implicitly. We think that this is a good sign. 

(A small side remark: One could say that to prove the vanishing of an equation, we do not need to evaluate the circuit,
it suffices to evaluate an algebraic decision tree.
However, it is well known that for every algebraic 
decision tree that decides membership in a variety, there is an algebraic circuit of roughly the same size
that \emph{computes} a multiple of this equation (by following the so-called generic path and then applying
Hilbert's Nullstellensatz) see e.g.~\cite{BCS}.
Then we can use Kaltofen factorization \cite{Kaltofen1987} and
the recent variants \cite{DBLP:conf/stoc/Dutta0S18}, which also work to some extent for exponential degree
equations, to get a circuit computing the equation. So algebraic decision trees will most likely not overcome
the natural proofs barrier.)

\item The second item that we have to prove is
$f_n(x_n) \not= 0$. In our setting, we easily can overcome this problem. We can choose
$S_{n,r_1(n),r_2(n),r_3(n)}$ such that $r_1(n) + r_2(n) + r_3(n)$ is large enough or $T_{k(n),n,r'(n)}$ for some large enough $r'(n) > r(n)$ as our point, respectively. Then, since we have an ascending chain of orbit closures,
$S_{n,r_1(n),r_2(n),r_3(n)}$ or
$T_{k(n),n,r'(n)}$ does not lie in $V_n$ by design and hence the equation $f_n$ will not vanish on 
it, since it is nontrivial. Again, we believe that this a promising sign, too. 
Of course, this will not be as easy
in the permanent versus determinant setting. However, there is  hope that GCT can prove $f(\per_n)\neq 0$ using symmetry properties of 
$\per_n$. Unlike a generic point, the permanent is characterized by its symmetries, so such a proof would be special to the permanent.
Note that for our points $S_{n,r_1(n),r_2(n),r_3(n)}$ and $T_{k(n),n,r'(n)}$,
the situation is similar. They
are almost characterized by their symmetries, for instance,
$S_{n,r_1(n),r_2(n),r_3(n)}$ is the direct sum of three tensors that
are characterized by its symmetries.

\item Finally, in the case of the permanent versus the determinant problem, it is possible that 
deciding whether a family of polynomials is in $\overline{\VP}$ is an easy problem.
This is the algebraic analogue of the minimum circuit size problem, the complexity of which
is widely open. The same could be the case for the border tensor rank problem, although here, this is rather
unlikely, since we know that the tensor rank problem is $\NP$-hard.

%\item some point in the context of slice rank.
\end{itemize}

\section{Geometric description of slice rank varieties}

For the necessary mathematical background,
the reader is referred \cite{Shafarevich1994,landsberg1,landsberg2,biskript}.

In this section, we describe the geometric description of the slice rank varieties for 3-tensors.
Let us say we are given a 3-tensor $T \in U \otimes V \otimes W$, and we are interested in finding out if it has slice rank at most $r$, i.e., if $\srk(T) \leq r$.

In what follows, we phrase this problem geometrically and formulate it as  variety membership testing problem. More explicitly, we write it as membership testing of $T$ in a union of orbit closures of certain tensors.

\begin{lemma}(\cite[Corollary 2]{ts16})\label{lem:svclosed}
Let $U, V, W$ be vector spaces over an algebraically closed field $\mathbb{F}$. The set of all tensors $T \in U \otimes V \otimes W$ with slice rank at most $r$ is a Zariski closed set.
\end{lemma}

In fact, they even showed that the set of all tensors $T \in U \otimes V \otimes W$ with slice rank at most $r$ decomposed as $(r_1,r_2,r_3)$  for a fixed tuple $(r_1, r_2,r_3)$ with $r_1 + r_2 + r_3 = r$ is also Zariski closed.

\begin{definition}
We call the the affine variety
\[ \mathcal{SV}_{U \otimes V \otimes W, r} = 
\{ T \in U \otimes V \otimes W | \srk(T) \leq r \} \]
the \emph{affine slice rank variety} or simply the \emph{slice rank variety}. 
\end{definition}

When clear from the context, we drop the index $U \otimes V \otimes W$.

\begin{lemma}\label{lem:subspace}
Let $U, V$, and $W$ be subspaces of vector spaces $U', V'$, and $W'$, respectively. Then 
\[\mathcal{SV}_{U \otimes V \otimes W, r} = \mathcal{SV}_{U' \otimes V' \otimes W', r} \cap (U \otimes V \otimes W). \]
\end{lemma}

\begin{proof}
A tensor lies in $\mathcal{SV}_{U \otimes V \otimes W,r}$ iff it is an element of the space $U \otimes V \otimes W$ and has slice rank at most $r$, i.e. lies in $\mathcal{SV}_{U' \otimes V' \otimes W', r}$.
\end{proof}

\begin{lemma}\label{lem:invariant}
The slice rank variety $\mathcal{SV}_{U \otimes V \otimes W ,r}$ is invariant under the standard action of $GL(U)\times GL(V) \times GL(W)$ on $U \otimes V \otimes W$.
\end{lemma}
\begin{proof}
If $\srk(T) \leq r$, we have $T = \displaystyle \sum_{i=1}^{r_1} u_{i,1} \otimes_1 T_{i,1} + \sum_{i=1}^{r_2} u_{i,2}  \otimes_2 T_{i,2} + \sum_{i=1}^{r_3} u_{i,3}  \otimes_3 T_{i,3}$ for some $(r_1,r_2,r_3)$ such that $r_1 + r_2 + r_3 = r$, where $u_{i,1} \in U$, $u_{i,2} \in V$, $w_{i,3}\in W$, and $T_{i,1} \in V \otimes W$, $T_{i,2} \in U \otimes W$, $T_{i,3} \in U \otimes V$. Clearly when $A \otimes B \otimes C \in \GL(U) \times \GL(V) \times \GL(W)$ acts on $T$, the slice rank remains at most $r$.
\end{proof}

\subsection{Slice rank varieties and orbit closures}
\label{sec:srk:orbit}

%For simplicity of exposition, we will consider the case when $U = V = W = \mathbb{F}^n$ i.e. the tensor $T \in U \otimes V \otimes W = \mathbb{F}^{n} \otimes \mathbb{F}^{n} \otimes \mathbb{F}^{n}$.

For every tuple $(r_1, r_2, r_3)$ of non-negative integers such that $r_1 + r_2 + r_3 =r$, we consider the vector spaces $U'_{(r_1, r_2, r_3)} = \mathbb{F}^{r_1} \oplus (\mathbb{F}^n)^{\oplus(r_2)}
 \oplus (\mathbb{F}^n)^{\oplus(r_3)}$, $V'_{(r_1, r_2, r_3)} = (\mathbb{F}^n)^{\oplus(r_1)}
  \oplus \mathbb{F}^{r_2} \oplus  (\mathbb{F}^n)^{\oplus(r_3)}$, and $W'_{(r_1, r_2,r_3)} =(\mathbb{F}^n)^{\oplus(r_1)} \oplus  (\mathbb{F}^n)^{\oplus(r_2)} \oplus \mathbb{F}^{r_3}$. We will drop the index $(r_1,r_2,r_3)$ in the following.
  
$U'$ has dimension $s_1(r_1,r_2,r_3) = r_1+nr_2 + nr_3$, and is 
decomposed into $1 + r_2 + r_3$ summands, where one summand is of dimension $r_1$, while the other summands are of dimensions $n$ each. Similarly, $V'$ and $W'$ have dimensions $s_2(r_1,r_2,r_3) = nr_1+ r_2+ nr_3$ and $s_3(r_1,r_2,r_3) = nr_1 + nr_2 + r_3$, respectively, and are 
 decomposed analogously as $U'$, into $r_1 + 1 + r_3$ summands and $r_1 + r_2 + 1$ summands respectively. We will denote $s_1(r_1,r_2,r_3)$, $s_2(r_1,r_2,r_3)$ and $s_3(r_1,r_2,r_3)$ simply by $s_1$, $s_2$, and $s_3$, respectively.
Thus $U' \otimes V' \otimes W' \cong \mathbb{F}^{s_1} \otimes \mathbb{F}^{s_2} \otimes \mathbb{F}^{s_3}$.

Let us give names to the components:
Let $L^1$ be $(\mathbb{F}^n)^{\oplus(r_1)}$ of dimension $nr_1$, $L^2$ be $(\mathbb{F}^n)^{\oplus(r_2)}$, and $L^3$ be $(\mathbb{F}^n)^{\oplus(r_3)}$, respectively, and we have vector spaces $\tilde{U} = \mathbb{F}^{r_1}$, $\tilde{V} = \mathbb{F}^{r_2}$ and $\tilde{W} =  \mathbb{F}^{r_3}$ 
respectively. Let $L^k_i $ be the $i$-th summand of $L^k, k\in \{1,2,3 \}$ with standard basis $e^k_{ij}$, $j \in [n]$, and let $e^1_i, e^2_i$ and 
$e^3_i$ be the standard basis of $\tilde{U}$, $\tilde{V}$ and $\tilde{W}$.
We have $U'= \tilde U \oplus L^2_1 \oplus \dots \oplus L^2_{r_2} \oplus L^3_1 \oplus \dots \oplus L^3_{r_3}$ and similar decomposition for $V'$ and $W'$.
   
\begin{definition}  \label{def:srk:unittensor}
   For $(r_1,r_2,r_3)$,
   we define the  \emph{unit slice rank tensor} $S_{n,r_1,r_2,r_3} \in (\tilde{U} \otimes L^1 \otimes L^1) \oplus (L^2 \otimes \tilde{V} \otimes L^2) \oplus (L^3 \otimes L^3 \otimes \tilde{W}) \subseteq U' \otimes V' \otimes W'$ as
\[ S_{n,r_1,r_2,r_3} = \sum_{i =1}^{r_1} \sum_{j=1}^{n}e^1_i \otimes e^1_{ij}\otimes e^1_{ij} + \sum_{i =1}^{r_2} \sum_{j=1}^{n}e^2_{ij} \otimes e^2_i\otimes e^2_{ij} + \sum_{i =1}^{r_3} \sum_{j=1}^{n}e^3_{ij} \otimes e^3_{ij}\otimes e^3_i.
\]
\end{definition}

Along $\tilde{U}$ we have $r_1$ slices where each slice contains an $n \times n$ identity matrix each in disjoint blocks. Then along $\tilde{V}$, we have $r_2$ slices with $n \times n$ identity matrices in disjoint blocks. Finally, we have $r_3$ slices with $n \times n$ identity matrices in disjoint blocks along $\tilde{W}$. Thus $S_{n,r_1,r_2,r_3}$ can be decomposed into three summands $S_{n.r_1} \in \tilde{U} \otimes L^1 \otimes L^1, S_{n,r_2} \in L^2 \otimes \tilde{V} \otimes L^2$ and $ S_{n,r_3}\in L^3 \otimes L^3 \otimes \tilde{W}$ such that $S_{n,r_1,r_2,r_3} = S_{n,r_1} \oplus S_{n,r_2} \oplus S_{n,r_3}$.

The group $\GL_{s_1} \times \GL_{s_2} \times \GL_{s_3}$ acts on $U' \otimes V' \otimes W'$ in a natural way. The slice rank variety can be defined as the union of orbit closures of $S_{n,r_1,r_2,r_3}$ under the action of $\GL_{s_1} \times \GL_{s_2} \times \GL_{s_3}$, where the union is taken over $(r_1,r_2,r_3)$ such that $r_1+r_2+r_3 = r$.

\begin{lemma} \label{lem:srk:orbit}
Let $U$, $V$, and $W$ be $n$-dimensional subspaces of $U'$, $V'$, and $W'$, respectively. Then we have
\[ \mathcal{SV}_{U \otimes V \otimes W, r} = \bigcup_{\substack{r_1, r_2, r_3 \\ r_1 + r_2 + r_3 = r}} \overline{(\GL_{s_1} \times \GL_{s_2} \times \GL_{s_3})S_{n, r_1, r_2,r_3}} \cap (U \otimes V \otimes W).\]
\end{lemma}

Note that each of the orbit closures is taken in a different ambient space,
since each $S_{n,r_1,r_2,r_3}$ lives in a different ambient space. 
But since we intersect each closure with $U \otimes V \otimes W$,
this is fine. We could also embed all $S_{n,r_1,r_2,r_3}$ into
a larger ambient space, however, this is disadvantageous when we want
to determine the stabilizers later.

\begin{proof}
First of all note that for every such $(r_1,r_2,r_3)$, we have that $S_{n,r_1,r_2,r_3} \in \mathcal{SV}_{U' \otimes V' \otimes W',r}$, simply by the construction of $S_{n,r_1,r_2,r_3}$, where $U' \cong \F^{s_1}, V \cong \F^{s_2}, W' \cong \F^{s_3}$. Now since by Lemma \ref{lem:invariant}, $\mathcal{SV}_{U' \otimes V' \otimes W',r}$ is invariant under the action of $\GL(U')\times \GL(V') \times \GL(W')$, we have that the entire orbit $(\GL_{s_1} \times \GL_{s_2} \times \GL_{s_3})S_{n, r_1, r_2,r_3}$ lies in it. Also, from Lemma \ref{lem:svclosed} (see \cite[Corollary 2]{ts16}), it follows that $(\GL_{s_1} \times \GL_{s_2} \times \GL_{s_3})S_{n, r_1, r_2,r_3}$ is contained in a Zariski closed subset of $\mathcal{SV}_{U' \otimes V' \otimes W'}$ and hence the orbit closure $\overline{(\GL_{s_1} \times \GL_{s_2} \times \GL_{s_3})S_{n, r_1, r_2,r_3}}$  also lies in $\mathcal{SV}_{U' \otimes V' \otimes W'}$. 
%Since all these inclusions hold for all $(r_1,r_2,r_3)$, the union of the %corresponding orbit closures also lies in $\mathcal{SV}_{U' \otimes V' \otimes %W'}$. 
Now we apply Lemma \ref{lem:subspace} to get the desired inclusion.

For the other direction, let us assume $T \in \mathcal{SV}_{U \otimes V \otimes W, r}$. Since $\srk(T) \leq r$, we have that we have $T = \displaystyle \sum_{i=1}^{r_1} u_{i,1} \otimes_1 T_{i,1} + \sum_{i=1}^{r_2} u_{i,2}  \otimes_2 T_{i,2} + \sum_{i=1}^{r_3} u_{i,3}  \otimes_3 T_{i,3}$, for some $(r_1,r_2,r_3)$ such that $r_1 + r_2 + r_3 = r$, where $u_{i,1} \in U$, $u_{i,2} \in V$, and $w_{i,3}\in W$ and $T_{i,1} \in V \otimes W$, $T_{i,2} \in U \otimes W$, and  $T_{i,3} \in U \otimes V$. Since $\forall i \in [r_1]$, $\rk(T_{i,1}) \leq n$, we can write $T_{i,1}$ as $(Q_{i,1} \otimes R_{i,1})(\sum_{j=1}^n e^1_{i,j} \otimes e^1_{i,j})$ for linear maps $Q_{i,1}: L^{1}_i \rightarrow V$ and $R_{i,1} : L^{1}_i \rightarrow W$. 
Analogously, $T_{i,2} = (P_{i,2} \otimes R_{i,2})(\sum_{j=1}^n e^2_{i,j} \otimes e^2_{i,j})$ for linear maps $P_{i,2}: L^{2}_i \rightarrow U$ and $R_{i,2} : L^{2}_i \rightarrow W$, and $T_{i,3} = (P_{i,3} \otimes Q_{i,3})(\sum_{j=1}^n e^3_{i,j} \otimes e^3_{i,j})$ for linear maps $P_{i,3}: L^{3}_i \rightarrow U$ and $Q_{i,3} : L^{3}_i \rightarrow V$.

Let $Q_1: L^1 \rightarrow V$ and $R_1: L^1 \rightarrow W$ be linear maps which are equal to $Q_{i,1}$ and $R_{i,1}$, respectively, when restricted to the $i$-th slice $L^1_i$. Similarly we have maps $P_2: L^2 \rightarrow U$ and $R_2: L^2 \rightarrow W$ whose restrictions to $i$-th slices are $P_{i,2}$ and $R_{i,2}$, respectively, and $P_3: L^3 \rightarrow U$ and $Q_3: L^3 \rightarrow V$ have their restrictions as $P_{i,3}$ and $Q_{i,3}$.

Finally, we also have linear maps $P_1: \tilde{U} \rightarrow U$ sending $e^1_i$ to $u_{i,1}$ ,  $Q_2: \tilde{V} \rightarrow V$ sending $e^2_i$ to $u_{i,2}$ and $R_3: \tilde{W} \rightarrow W$ sending $e^3_i$ to $u_{i,3}$. 

Thus $T = ((P_1 \otimes Q_1 \otimes R_1) \oplus (P_2 \otimes Q_2 \otimes R_2) \oplus  (P_3 \otimes Q_3 \otimes R_3)) S_{n,r_1,r_2,r_3}$ for some $(r_2,r_2,r_3)$. 
The closure of $\GL_{s_1}, \GL_{s_2}$ and $\GL_{s_3}$ contains all linear endomorphisms of $U', V'$ and $W'$, respectively, and thus contains $(P_1 \otimes Q_1 \otimes R_1) \oplus (P_2 \otimes Q_2 \otimes R_2) \oplus  (P_3 \otimes Q_3 \otimes R_3)$. Therefore, $T$ lies in the closure  $\overline{(\GL_{s_1} \times \GL_{s_2} \times \GL_{s_3})S_{n, r_1, r_2,r_3}}$ for some $(r_1, r_2, r_3)$ with $r_1 + r_2 + r_3 = r$.
\end{proof}

\begin{corollary}
\[ \mathcal{SV}_{n, r} := \mathcal{S}_{\mathbb{F}^n \otimes \mathbb{F}^n \otimes \mathbb{F}^n, r} = \bigcup_{\substack{r_1, r_2, r_3 \\ r_1 + r_2 + r_3 = r}} \overline{(\GL_{s_1} \times \GL_{s_2} \times \GL_{s_3})\mathcal{S}_{n, r_1, r_2,r_3}} \cap (\mathbb{F}^n \otimes \mathbb{F}^n \otimes \mathbb{F}^n)\]
\end{corollary}
\begin{proof}
We identify $\mathbb{F}^n$ as a subspace of $U', V'$ as well as $W'$ by embedding an element $x \in \mathbb{F}^n$ in the bigger spaces as $(x,0, \ldots, 0)$ where we have zeros in all the other coordinates except the first $n$ coordinates. Now apply Lemma \ref{lem:subspace}.
\end{proof}

% Markus: I think, the proof above is not needed.

We now describe the stabilizers of the unit slice rank tensors.

\begin{lemma}\label{lem:stabidentity}
  The stabilizer of $\sum_{i = 1}^k e_i \otimes e_i \in \F^k \otimes \F^k$ in $\GL_k \times \GL_k$ consists of elements of the form $(A, A^{-\trans})$.
\end{lemma}
\begin{proof}
For the left action of $\GL_k \times \GL_k$ on $\F^k \otimes \F^k$ consider the corresponding left-right action:
$AXB := (A,B^{\trans})X$ for $A,B\in \GL_k$ and $X \in \F^k\otimes \F^k$.
If we interpret $F^k\otimes F^k$ as the space of $k \times k$ matrices, then
$\sum_{i = 1}^k e_i \otimes e_i$ is the identity matrix $I$ and we observe that $AXB$ is the usual product of matrices.
Clearly $AIB=I$ iff $A=B^{-1}$.
\end{proof}

\begin{theorem}\label{slicerank:thm:stab}
  For $n \ge 2$, the stabilizer of $S_{n, r_1,r_2,r_3}$ in $\GL_{s_1} \times \GL_{s_2} \times \GL_{s_3}$ is isomorphic to $\bigoplus\limits_{i=1}^{3} ((\GL_n \times \GL_1)^{r_i} \rtimes \mathfrak{S}_{r_i})$. The element 
    $(Z_{i1}, z_{i1}, \dots, Z_{ir_i}, z_{ir_i}) \in  (\GL_n \times \GL_1)^{r_i}$ for $i = 1,2,3$
  is embedded into $\GL_{r_1} \times \GL_{nr_1} \times \GL_{nr_1}$,  $\GL_{nr_2} \times \GL_{r_2} \times \GL_{nr_2}$ and $\GL_{nr_3} \times \GL_{nr_3} \times \GL_{r_3}$ respectively, via 
  \begin{align*}
    (\diag(z_{11}, \dots, z_{1r_1}), \diag(Z_{11}, \dots, Z_{1r_1}), \diag((z_{11} Z_{1r_1})^{-\trans}, \dots, (z_{1r_1} Z_{1r_1})^{-\trans})), & \\
    (\diag(Z_{21}, \dots, Z_{2r_2}), \diag(z_{21}, \dots, z_{2r_2}), \diag((z_{21} Z_{2r_2})^{-\trans}, \dots, (z_{2r_2} Z_{2r_2})^{-\trans})), & \text{ and} \\
      ( \diag(Z_{31}, \dots, Z_{3r_3}), \diag((z_{31} Z_{3r_3})^{-\trans}, \dots, (z_{3r_3} Z_{3r_3})^{-\trans}), \diag(z_{31}, \dots, z_{3r_3})
      ),
\end{align*}  
 respectively. The $\mathfrak{S}_{r_i}$ factor permutes the $r_i$ coordinates of $\tilde{U}, \tilde{V}$ and $\tilde{W}$,  and the $r_i$ summands of  $L^i \times L^i$ simultaneously.
\end{theorem}
\begin{proof}
  Let $S:= S_{n,r_1,r_2,r_3} = S_{n,r_1} \oplus S_{n,r_2} \oplus S_{n,r_3}$, and $(A, B, C) \in \stab S$, that is, $(A \otimes B \otimes C) S = S$.
 It will be useful to visualize $A$, $B$ and $C$ as
 \[
A = 
\left( \begin{array}{c|c|c}
   A_{11} & A_{12} & A_{13} \\
   \midrule
   A_{21} & A_{22} & A_{23}  \\
   \midrule
   A_{31} & A_{32}  & A_{33} 
\end{array}\right), \enspace
B = 
\left( \begin{array}{c|c|c}
   B_{11} & B_{12} & B_{13} \\
   \midrule
   B_{21} & B_{22} & B_{23} \\
   \midrule
   B_{31} & B_{32} & B_{33},
\end{array}\right), \enspace
C = 
\left( \begin{array}{c|c|c}
   C_{11} & C_{12} & C_{13} \\
   \midrule
   C_{21} & C_{22} & C_{23} \\
   \midrule
   C_{31} & C_{32} & C_{33}
\end{array}\right). 
\]
Above,
$A_{11}$ is an $r_1 \times r_1$ matrix, $A_{22}$ is an $nr_2 \times nr_2$ matrix, and $A_{33}$ is an $nr_3 \times nr_3$ matrix, respectively.
 $B_{11}$ is an $nr_1 \times nr_1$ matrix, $B_{22}$ is an $r_2 \times r_2$ matrix, and $B_{33}$ is an $nr_3 \times nr_3$ matrix, respectively, and  $C_{11}$ is an $nr_1 \times nr_1$ matrix , $C_{22}$ is an $nr_2 \times nr_2$ matrix, and $C_{33}$ is an $r_3 \times r_3$ matrix, respectively.
  
 Let $S = \sum_{i=1}^{s_1} e_i \otimes_1 S_{i}$, where $S_i$ is the $i$-th slice of $S$.
 Then,
  \[(A \otimes B \otimes C) S_{n,r_1,r_2,r_3} = \sum_{i = 1}^{s_1} (Ae_i) \otimes (B \otimes C)S_i = \sum_{i = 1}^{s_1} e_i \otimes (B \otimes C)(\sum_{j = 1}^{s_1} {a_{ij}} S_j).\]

  First of all we divide the set of slices into groups. These include:
  \begin{itemize}
      \item $r_1$ groups of size $1$ each, $\{1\}, \ldots, \{n\}$,
      \item $r_2$ groups of size $n$ each, $\{r_1 +1, \ldots, r_1 + n \}, \ldots, \{r_1 + (r_2 -1)n +1, \ldots, r_1 + r_2n\}$,
      \item $r_3$ groups of size $n$ each, $\{r_1 + r_2n + 1, \ldots, r_1 + r_2n + n \}, \ldots, \{r_1 + r_2n + (r_3 -1)n+ 1, \ldots, r_1 + r_2n + r_3n\}$.
  \end{itemize}

  We first consider the first $r_1$ groups of slices, i.e., slices $S_i$ for $i \in \{ 1, \ldots,r_1\}$ to deduce about the first $r_1$ rows of $A$. 
  
  Recall that for $i \in \{1, \ldots, r_1\}$, $\rk(S_i) = n$ (by definition). Thus we have that $\rk(\sum_{j = 1}^{s_1} {a_{ij}} S_j)$ and consequently the rank of the $i$-th slice of $(A \otimes B \otimes C) S$ will be at least $qn$, where $q$ is the number of nonzero entries among $a_{i1}, \dots, a_{ir_1}$.
  Therefore, there will be at most one $j \in \{1,\ldots, r_1 \}$ such that $a_{ij}$ is nonzero. Now consider the case when for some $i' \in \{ 1,\ldots, r_1 \}$, $a_{i'j} \neq 0$ for some $j \in \{r_1 +1,\ldots, r_1 + nr_2, \ldots, r_1 + nr_2 + nr_3\}$. 
  First of all, it implies that $a_{i'j} = 0$, for all $j \in \{1, \ldots, r_1\}$, otherwise  $\rk(\sum_{j = 1}^{s_1} {a_{i'j}} S_j) \geq n+1$. Now since $A$ induces a bijection among the slices, every slice is involved in the linear combination of at least one of the slices. 
  And since two slices of rank $n$ cannot be involved in the linear combination of first $r_1$ slices, the above forces that at least one of the rank $n$ slices is involved in the linear combination of a slice $S_i$ for $i \in \{r_1 +1,\ldots, r_1 + nr_2, \ldots, r_1 + nr_2 + nr_3\}$, i.e., $a_{ij} \neq 0$ for some $i \in \{r_1 +1,\ldots, r_1 + nr_2, \ldots, r_1 + nr_2 + nr_3\}$,  $j\in  \{1, \ldots, r_1\}$. 
  But this implies that $\rk(\sum_{j = 1}^{s_1} {a_{ij}} S_j) \geq n$ and not $1$, which cannot be the case if $A \otimes B \otimes C \in \stab S$.
  Thus $a_{ij} = 0$ for all $j \in \{r_1 +1,\ldots, r_1 + nr_2, \ldots, r_1 + nr_2 + nr_3\}$. Thus $A_{11}$ is a product of a diagonal matrix and a permutation matrix, and $A_{12} = A_{13} = 0$.
  Symmetrical argument implies that $B_{21} = B_{23} = C_{31} = C_{32} = 0$, and both $B_{22}$ and $C_{33}$ are products of a diagonal matrix and a permutation matrix.

 Now consider the first group from the second case i.e. $i \in \{r_1 +1,\ldots, r_1 + n\}$:
 Here, first of all recall that $\rk(S_i) = 1$ for all $i$. Thus $\rk(\sum_{j=1}^{s_1}a_{ij}S_j)$ has to be $1$. This immediately implies that $a_{ij} = 0$ for all $j \in \{1, \ldots, r_1 \}$, otherwise the resulting rank will be at least $n$. We further argue that $a_{ij} = 0$ for all $j \in \{ r_1 + r_2n +1, \ldots, r_2n+n,\ldots,r_1 + r_2n+r_3n \}$. 
  Assume the contrary. Then $\sum_{j=1}^{s_1}a_{ij}S_j$ will have something in the bottom right $r_3 \times nr_3$ block. In order for $(A, B, C)$ to be in $\stab S$, $(B \otimes C)(\sum_{j=1}^{s_1}a_{ij}S_j)$ should bring it back to its original place, i.e., in the central $nr_2 \times r_2$ block. 
  However $C_{33}$ being a product of a diagonal matrix and a permutation matrix, $C$ will only permute the the last $r_3$ rows  of $\sum_{j=1}^{s_1}a_{ij}S_j$ within themselves and hence $B \otimes C$ will not bring $\sum_{j=1}^{s_1}a_{ij}S_j$ to the central block as needed. Thus from the above discussion, we have $A_{21} = A_{23} = A_{31} = A_{32} = B_{12} = B_{13} = B_{31} = B_{32} = C_{12} = C_{13} = C_{31} = C_{32} = 0$. 
  Finally, $A_{22}$ will be a product of a block diagonal matrix and a block permutation matrix. For this, notice that for a fixed $i \in \{r_1 +1,\ldots, r_1 + r_2n\}$, $a_{ij}$ cannot be non-zero for $j$'s belonging to more than one group, otherwise the rank of the resulting slice exceeds $1$. 
  
  Thus $A$ will be a product of a diagonal matrix with a permutation matrix in the top left block. In the central block, it will be a product of a block diagonal matrix with a block permutation matrix. Similarly, for the bottom right block, too, it will be a product of a block diagonal matrix with a block permutation matrix.  Thus the picture becomes 
  
   \[
A = 
\left( \begin{array}{c|c|c}
   A_{11} & \bigzero & \bigzero \\
   \midrule
   \bigzero & A_{22} & \bigzero  \\
   \midrule
   \bigzero & \bigzero  & A_{33} 
\end{array}\right), \enspace
B = 
\left( \begin{array}{c|c|c}
   B_{11} & \bigzero & \bigzero \\
   \midrule
   \bigzero & B_{22} & \bigzero \\
   \midrule
   \bigzero & \bigzero & B_{33}
\end{array}\right), \enspace
C = 
\left( \begin{array}{c|c|c}
   C_{11} & \bigzero & \bigzero\\
   \midrule
   \bigzero & C_{22} & \bigzero \\
   \midrule
   \bigzero & \bigzero & C_{33}
\end{array}\right), 
\] 
where $A_{11}, B_{22}$ and $C_{33}$ are products of a diagonal matrix and a permutation matrix, and $A_{22}, A_{33}, B_{11}, B_{33}, C_{11}$ and $C_{22}$ are all products of a block diagonal matrix and a block permutation matrix. Thus we can decompose $(A,B,C) \in \stab S$ as $((A_{11},B_{11},C_{11}) \oplus (A_{22},B_{22},C_{22}) \oplus (A_{33},B_{33},C_{33}))$ where $(A_{11},B_{11},C_{11})$ acts on $\tilde{U} \otimes L^1 \otimes L^1$,  $(A_{22},B_{22},C_{22})$ acts on $ L^2 \otimes \tilde{V} \otimes L^2$ and $(A_{33},B_{33},C_{33})$ acts on  $L^3 \otimes L^3 \otimes \tilde{W}$, respectively. Hence for $(A,B,C)$ to be in $\stab S$, with $S = S_{n,r_1} \oplus S_{n,r_2} \oplus S_{n,r_3}$, $(A_{11},B_{11},C_{11})$ must preserve $S_{n,r_1}$, $(A_{22},B_{22},C_{22})$ must preserve $S_{n,r_2}$, and $(A_{33},B_{33},C_{33})$ must preserve $S_{n,r_3}$, i.e., $\stab S = \stab S_{n,r_1} \oplus \stab S_{n,r_2}\oplus \stab S_{n,r_3}$, where $\stab S_{n,r_1} \subseteq \GL_{r_1} \times \GL_{nr_1} \times \GL_{nr_1}$, $\stab S_{n,r_2} \subseteq \GL_{nr_2} \times \GL_{r_2} \times \GL_{nr_2}$ and $\stab S_{n,r_3} \subseteq \GL_{nr_3} \times \GL_{nr_3} \times \GL_{r_3}$.

 We consider $\stab S_{n,r_1}$ now. Let $P_{\sigma_1}$ be an element of $\GL_{r_1} \times \GL_{nr_1} \times \GL_{nr_1}$ which permutes the $r_1$ coordinates of $\tilde{U} \subseteq U'$ and the $r_1$ summands of $L^1 \times L^1 \subseteq V'\times W'$ according to the permutation $\sigma_1$. It is easy to see that $P_{\sigma_1} \in \stab S_{n,r_1}$. Hence, $(A_{11},B_{11},C_{11})P^{-1}_{\sigma_1} = (\tilde{A}_{11}, \tilde{B}_{11}, \tilde{C}_{11}) \in \stab S_{n,r_1}$. 
Using this and the previous discussion, we have that $\tilde{A}_{11}$ is a diagonal matrix $\diag(\tilde{a}^{1}_{11}, \ldots, \tilde{a}^{r_1}_{11})$. Let 
   $\tilde{A}_{11}'$ be the linear map which scales elements of $L^1_{i}$ by $\tilde{a}^{i}_{11}$ for each $i \in [r_1]$. 
   Clearly
  $(\tilde{A}_{11}^{-1}, \id, \tilde{A}_{11}')$ also preserves $S_{n,r_1}$.
  Therefore, $(\tilde{A}_{11}, \tilde{B}_{11}, \tilde{C}_{11}) \cdot (\tilde{A}_{11}^{-1}, \id, \tilde{A}_{11}') = (\id, \tilde{B}_{11}, \hat{C}_{11})$ is in $\stab S_{n,r_1}$ 

  Now, since the first component of $(\id, \tilde{B}_{11}, \hat{C}_{11})$ is the identity, it preserves $S_{n,r_1}$ if and only if $\tilde{B}_{11} \otimes \hat{C}_{11}$ preserves each slice of $S_{n,r_1}$.
  If it preserves each slice, it also preserves its sum $ \sum_{i =1}^{r_1} \sum_{j=1}^{n} e^1_{ij}\otimes e^1_{ij}$, the full rank diagonal matrix of size $nr_1 \times nr_1$.
  Therefore, by Lemma \ref{lem:stabidentity}, $\hat{C}_{11} = \tilde{B}_{11}^{-\trans}$. Thus $(\id,\tilde{B}_{11}, \tilde{B}_{11}^{-\trans} ) \in \stab S_{n,r_1}$.
  
  Thus, we decomposed an element $A_{11} \otimes B_{11} \otimes C_{11} \in \stab S_{1}$ into a product of three special elements $(\id, \diag(B_{11}^{1}, \dots, B_{11}^{r_1}), \diag(B_{11}^{1}, \dots, B_{11}^{r_1})^{-\trans})$, $(\diag(\tilde{a}^{1}_{11}, \ldots, \tilde{a}^{r_1}_{11}), \id, \diag(\tilde{a}^{1}_{11} \id, \ldots, \tilde{a}^{r_1}_{11} \id)^{-1})$, and $P_{\sigma_1}$ for some permutation $\sigma_1 \in \mathfrak{S}_{r_1}$.
  These three types of elements correspond to three subgroups of $\stab S_{n,r_1}$.
  The subgroups intersect only in the identity; elements of the first two types commute, and the conjugation with $P_{\sigma_1}$ permutes $\tilde{a}^i_{11}$ and $B_{11}^i$ according to $\sigma_1$, so the product of the first subgroups is direct, and the last product is semidirect. Symmetrical arguments for $\stab S_{n,r_2}$ and $\stab S_{n,r_3}$ finishes the proof.
\end{proof}

Now we show that the unit slice rank tensor $S_{n,r_1,r_2,r_3}$ is almost characterized by its stabilizer. More precisely, it is a direct sum of three tensors $S_{n,r_1}, S_{n,r_2}$ and $S_{n,r_3}$ that are each characterized by their respective stabilizers.

\begin{theorem} \label{thm:srk:sym}
  Suppose $T$ is a tensor in $U' \otimes V' \otimes W'  = (\tilde{U} \oplus L^2 \oplus L^3) \otimes (L^1 \oplus \tilde{V} \oplus L^3) \otimes (L^1 \oplus L^2 \oplus \tilde{W}) $. If $\stab T = \stab S_{n, r_1, r_2, r_3}$, then 
 $$ T =  
  ((\diag(\alpha, \dots, \alpha), \id, \id), ( \id, \diag(\beta, \dots, \beta), \id),   ( \id, \id, \diag(\gamma, \dots, \gamma)) S_{n, r_1,r_2,r_3},$$ 
  for some $ \alpha, \beta, \gamma \neq 0$, i.e., $T =  \alpha \cdot S_{n,r_1} \oplus \beta \cdot S_{n,r_2} \oplus \gamma \cdot S_{n,r_3}$.
\end{theorem}

\begin{proof}
  Suppose $T$ is stabilised by $\stab S_{n, r_1,r_2,r_3}$.
  We first establish that $T$ also has the block structure like $S_{n, r_1,r_2,r_3}$, i.e., even though the ambient space of $T$ is $(\tilde{U} \oplus L^2 \oplus L^3) \otimes (L^1 \oplus \tilde{V} \oplus L^3) \otimes (L^1 \oplus L^2 \oplus \tilde{W})$, it sits completely inside one of the smaller subspaces $(\tilde{U} \otimes L^1 \otimes L^1) \oplus (L^2 \otimes \tilde{V} \otimes L^2) \oplus (L^3 \otimes L^3 \otimes \tilde{W})$ which also contains $S_{n,r_1,r_2,r_3}$.
  
  For this, we take the element $(A,B,C)  \in \stab S_{n,r_1,r_2,r_3} \subseteq \GL_{s_1} \times \GL_{s_2} \times \GL_{s_3}$ where we have $A = (\diag(\alpha_1, \ldots, \alpha_1), \diag(\alpha_2 \cdot \id, \ldots, \alpha_2 \cdot \id), \diag(\alpha_3 \cdot \id, \ldots, \alpha_3 \cdot \id))$, whereas $B = ( \diag(\beta_1 \cdot \id, \ldots, \beta_1 \cdot \id), \diag(\beta_2, \ldots, \beta_2), \diag(\beta_3 \cdot \id, \ldots, \beta_3 \cdot \id))$ and $C = (  \diag(\gamma_1 \cdot \id, \ldots, \gamma_1 \cdot \id), \diag(\gamma_2 \cdot \id, \ldots, \gamma_2 \cdot \id), \diag(\gamma_3, \ldots, \gamma_3))$ such that $\alpha_1  \beta_1  \gamma_1 =  \alpha_2  \beta_2  \gamma_2 = \alpha_3  \beta_3  \gamma_3 = 1$. Now, $U' \otimes V' \otimes W'$ is a direct sum of $27$ subspaces, and $T$ can be decomposed into corresponding $27$ blocks. On the action of $(A,B,C)$ on $T$, only three of the blocks remain fixed, i.e., the ones corresponding to the subspaces $(\tilde{U} \otimes L^1 \otimes L^1)$,  $(L^2 \otimes \tilde{V} \otimes L^2)$ and  $(L^3 \otimes L^3 \otimes \tilde{W})$ because the entries in these blocks get scaled by $\alpha_1 \beta_1 \gamma_1, \alpha_2 \beta_2 \gamma_2$ and $\alpha_3 \beta_3 \gamma_3$ respectively, which are all equal to unity. The blocks corresponding to the other subspaces will be scaled by non-unity and hence will not remain fixed. Hence for $(A,B,C)$ to be in the stabilizer of $T$, only the blocks corresponding to these three subspaces will be non-zero, which is also the case for $S_{n, r_1,r_2,r_3}$.
  
  Recall from Definition \ref{def:srk:unittensor} that $S_{n,r_1,r_2,r_3}$ can be decomposed as $S_{n,r_1} \oplus S_{n,r_2} \oplus S_{n,r_3}$. Thus, we decompose $T$ into blocks as $T = T_{n,r_1} \oplus T_{n,r_2} \oplus T_{n,r_3}$. We focus on $T_{n,r_1} =: T'$, where $T' \in (\tilde{U} \otimes L^{1} \otimes L^{1})$.
  
   Let $T'_1, \dots, T'_{r_1}$ be the slices of $T'$. Decompose them into blocks $T'_i = (T'_{ijk})$ according to the decomposition of $L^{1}$ into $L^{1}_i$.
  
  Let $A_i(\lambda): \tilde{U} \to \tilde{U}$ be the map which scales the $i$-th coordinate by $\lambda$, leaving other in place, and $B_i(\lambda): L^{1} \to L^{1}$ be the map which scales $L^{1}_i$ by $\lambda$ and acts like identity on the other summands.
  Applying to $T'$ the transformation $(A_i(\lambda^{-2}), B_i(\lambda), B_i(\lambda)) \in \stab S_{n,r_1}$,
  we see that all blocks of $T'_i$ except $T'_{iii}$ are zero, as they are multiplied by a coefficient $\lambda^{-2}$ or $\lambda^{-1}$ in this transformation.
  
  Applying $(\id, \diag(Z_{11}, \dots, Z_{1r_1}), \diag(Z_{11}, \dots, Z_{1r_1})^{-\trans})$ with arbitrary $Z_{1i}$ to $T'$, we obtain that each $T'_{iii}$ has the form $a_{1i} \sum_{j = 1}^{\dim L^{1}_i} e^{1}_{ij} \otimes e^{1}_{ij}$.
  Applying permutations on the blocks of $T'_{iii}$, we see that $a_{11} = \cdots = a_{1r_1} =: \alpha$.

  Therefore
  \[
    T' = \alpha \sum_{i = 1}^{r_1} \sum_{j = 1}^{n} e_i \otimes_1 e^{1}_{ij} \otimes e^{1}_{ij} = \alpha \cdot S_{n,r_1}.
  \]
  If $\alpha \neq 0$, then $T' = T_{n,r_1} = (\diag(\alpha, \dots, \alpha), \id, \id) S_{n, r_1}$.
  Applying the symmetrical arguments on $T_{n,r_2}$ and $T_{n,r_3}$, we get that
   $$T =  
  (\diag(\alpha, \dots, \alpha), \id, \id) S_{n,r_1} \oplus ( \id, \diag(\beta, \dots, \beta), \id)S_{n,r_2} \oplus ( \id, \id, \diag(\gamma, \dots, \gamma)) S_{n,r_3},$$ for some $ \alpha, \beta, \gamma \neq 0$,
  or simply
  $T =  \alpha \cdot S_{n,r_1} \oplus \beta \cdot S_{n,r_2} \oplus \gamma \cdot S_{n,r_3}$.
\end{proof}

\section{Ideals of slice rank varieties}\label{slicerank:sec:ideal}
In this section we will consider tensors over $\bbC$ (or over algebraically closed field of characteristic 0).
For an affine variety $\cA$, we denote its ideal by $I(\cA)$ and its coordinate ring by $\bbC[\cA]$.
Since slice rank varieties $\mathcal{SV}_{U\otimes V \otimes W,r}$ are scale-invariant, their ideals and coordinate rings inherit grading from $\bbC[U \otimes V \otimes W]$.
Since they are invariant under the action of $\GL(U) \times \GL(V) \times \GL(W)$, this group acts on $I(\mathcal{SV}_{U\otimes V \otimes W,r})$ and $\bbC[\mathcal{SV}_{U\otimes V \otimes W,r}]$.
We study representation-theoretic properties of specific equations in $I(\cM_r)$ and multiplicities of irreducible representations in $I(\mathcal{SV}_{U\otimes V \otimes W,r})$ and $\bbC[\mathcal{SV}_{U\otimes V \otimes W,r}]$.
Irreducible representations in $\bbC[U \otimes V \otimes W]_d$ are indexed by triples $(\lambda,\mu,\nu)$ of partitions of $D$,
and the multiplicity of the type $(\lambda,\mu,\nu)$ in $\bbC[U \otimes V \otimes W]_D$ is called the \emph{Kronecker coefficient}.
Finding a combinatorial interpretation for the Kronecker coefficient is Problem~10 in \cite{sta:00}.

Here we describe some polynomials in the ideals of slice rank varieties and study their representation-theoretic properties.

Throughout the section we assume $\dim U = k$, $\dim V = m$, $\dim W = n$.

Interestingly, the equations that we find in Section~\ref{subsec:slicerankeqdesign} and Section~\ref{subsec:slicerankeqmult} are exactly the same, even though the method of finding them is very different.
Also note that the equations are given indirectly just by their representation isomorphism type.

\begin{theorem}\label{thm:equations}
Let $\la \vdash_{n^2}D$, $\mu \vdash_{n^2}D$, $\nu \vdash_{n^2}D$ be partitions of $D$ with at most $n^2$ rows.
If $\la_1 \leq n$ and $\mu_1 \leq n$ and $\nu_1 \leq n$,
then all $\GL_{n^2}^3$-modules of type $(\la,\mu,\nu)$ are in the vanishing ideal of $\mathcal{SV}_{\bbC^{n^2} \otimes \bbC^{n^2} \otimes \bbC^{n^2},r}$ for all $r<D/n$.
\end{theorem}
The whole $\GL_{n^2}^3$-module $\{\la,\mu,\nu\}$ being in the vanishing ideal means that if not all of the functions in these modules vanish at $T$, then $(\la,\mu,\nu)$ is an occurrence obstruction.

The abundance of equations that we get from Theorem~\ref{thm:equations} is quite remarkable.
Note that using Schur-Weyl duality we have that the multiplicity of
$\{\la,\mu,\nu\}$ in $\Sym^D \otimes^3 \bbC^{n^2}$ is the Kronecker coefficient $k(\la,\mu,\nu)$.
The dimension of the space of degree $D$ equations that we obtain for slice rank $r$ is
$\sum_{D=nr+1}^{n^3}\sum_{\la,\mu,\nu\vdash_{n^2}D} k(\la,\mu,\nu)\dim\{\la\}\dim\{\mu\}\dim\{\nu\}$.
The dimensions of $\{\la\}$, $\{\mu\}$, $\{\nu\}$ are given by the hook content formula.

Theorem~\ref{thm:equations} gives equations for slice rank in the full range, up to
the most extreme case $\la=\mu=\nu=n^2 \times n \vdash n^3$, which gives equations for slice rank $<n^3/n = n^2$.
This equation is known as Cayley's hyperdeterminant.
Its evaluation at the matrix multiplication tensor is explained combinatorially in \cite[Prop.~5.24]{BI:17fundinv}.

Many more equations for slice rank $<n^2$ than just this hyperdeterminant are readily constructed from Theorem~\ref{thm:equations}.
It is an open question whether or not the equations from Theorem~\ref{thm:equations} cut out $\mathcal{SV}_{\bbC^{n^2} \otimes \bbC^{n^2} \otimes \bbC^{n^2},r}$
set-theoretically.

\subsection{Equations from designs}\label{subsec:slicerankeqdesign}

We study slice rank in $\bbC^{N} \otimes \bbC^{N} \otimes \bbC^{N}$.
It will be natural to have $N=n^2$.
We start by establishing a construction principle for equations.
\subsubsection*{Construction of highest weight vectors}
Given a representation $X$ of $\GL(U)\times\GL(V)\times\GL(W)$,
a \emph{highest weight vector} $t \in X$ is a vector that
\begin{enumerate}
\item is invariant under the action of triples of upper triangular matrices with 1s on the main diagonal
\item satisfies that $x$ is rescaled under the action of triples of diagonal matrices as follows for some triple of partitions $(\la,\mu,\nu)$:
\end{enumerate}
\begin{eqnarray*}
&& (\diag(\alpha_1,\ldots,\alpha_{k}),\diag(\beta_1,\ldots,\beta_{m}),\diag(\gamma_1,\ldots,\gamma_{n}))x \\
&=&
\alpha_1^{\la_1}\cdots\alpha_{k}^{\la_{k}}
\beta_1^{\mu_1}\cdots\beta_{m}^{\mu_{m}}
\gamma_1^{\nu_1}\cdots\gamma_{n}^{\nu_{n}}x
\end{eqnarray*}
The triple $(\la,\mu,\nu)$ is called the \emph{type} of the highest weight vector.
Each irreducible $\GL(U)\times\GL(V)\times\GL(W)$-representation $X$ of type $(\la,\mu,\nu)$
has exactly one highest weight vector (up to scale),
and the type of $X$ coincides with the type of its highest weight vector.
Moreover, $X$ equals the linear span of the $\GL(U)\times\GL(V)\times\GL(W)$-orbit of its highest weight vector.

To construct an irreducible representation of nontrivial equations for a variety, we construct the corresponding highest weight vector and prove that it vanishes on the variety.
The tensor product of tensor powers $U^{\otimes D} \otimes V^{\otimes D} \otimes W^{\otimes D}$ is
known to decompose into irreducibles of the group $\GL(U)\times\GL(V)\times\GL(W)\times\aS_D\times\aS_D\times\aS_D$ via Schur-Weyl duality:
\begin{equation*}
U^{\otimes D} \otimes V^{\otimes D} \otimes W^{\otimes D} = \bigoplus_{\la,\mu,\nu} \{\la\} \otimes \{\mu\} \otimes \{\nu\} \otimes [\la] \otimes [\mu] \otimes [\nu],
\end{equation*}
where the sum runs over all partition triples $\la,\mu,\nu$ such that $\la$, $\mu$, $\nu$ have $D$ boxes and $\la$ has at most $\dim U$ many rows, $\mu$ has at most $\dim V$ many rows, and $\nu$ has at most $\dim W$ many rows.
The $D$-th tensor power of $U \otimes V \otimes W$ is isomorphic to $U^{\otimes D} \otimes V^{\otimes D} \otimes W^{\otimes D}$.
Let $\varrho:U^{\otimes D} \otimes V^{\otimes D} \otimes W^{\otimes D} \to (U \otimes V \otimes W)^{\otimes D}$
denote this canonical isomorphism.
Moreover, if we embed $\aS_D \hookrightarrow \aS_D\times\aS_D\times\aS_D$, $\pi \mapsto(\pi,\pi,\pi)$,
then the space of homogeneous degree $D$ polynomials on $U \otimes V \otimes W$ can be identified with the $\aS_D$-invariant linear subspace of
$(U \otimes V \otimes W)^{\otimes D}$ in a very natural way via polarization and restitution:
If $F$ is a homogeneous degree $D$ polynomial on $U^* \otimes V^* \otimes W^*$ and $f$ its corresponding tensor in $((U \otimes V \otimes W)^{\otimes D})^{\aS_D}$,
then the evaluation of 
$f$ at a point $t \in U^* \otimes V^* \otimes W^*$ equals
the tensor contraction
\begin{equation}\label{slicerank:eq:evaluationcontraction}
F(t) = \langle f, (t^{\otimes D})\rangle.
\end{equation}
We fix a basis of $U$, $V$, and $W$ and denote each basis with $e_1, e_2, \ldots$ when there is no possibility of confusion.
For a partition $\la$, let $\la^t$ denote its transpose.
Given a triple $(\la,\mu,\nu)$ of partitions of $D$, a highest weight vector $h$ of type $(\la,\mu,\nu)$ in
$U^{\otimes D} \otimes V^{\otimes D} \otimes W^{\otimes D}$ can be constructed via $h := h_\la \otimes h_\mu \otimes h_\nu$,
where
\[
h_\la := e_1 \wedge e_2 \wedge \cdots \wedge e_{\la^t_1} \otimes e_1 \wedge e_2 \wedge \cdots \wedge e_{\la^t_2} \otimes \cdots \otimes
e_1 \wedge e_2 \wedge \cdots \wedge e_{\la^t_{\la_1}}.
\]
Let $\pi^{(1)}\in \aS_D$, $\pi^{(2)}\in \aS_D$, $\pi^{(3)}\in \aS_D$.
Clearly $(\pi^{(1)} h_\la) \otimes (\pi^{(2)} h_\mu) \otimes (\pi^{(3)} h_\nu)$ is also a highest weight vector of type $(\la,\mu,\nu)$.
The projection of $\varrho((\pi^{(1)} h_\la) \otimes (\pi^{(2)} h_\mu) \otimes (\pi^{(3)} h_\nu))$
onto $((U\otimes V \otimes W)^{\otimes D})^{\aS_D}$ corresponds to a highest weight vector
of type $(\la,\mu,\nu)$
in $\bbC[U\otimes V\otimes W]_D$ via eq.~\eqref{slicerank:eq:evaluationcontraction}.
The vector space of highest weight vectors of weight $(\la,\mu,\nu$ in $\bbC[U\otimes V\otimes W]_D$ is spanned by the vectors that are constructed in this fashion.

\subsubsection*{Evaluation via products of determinants}
Let $T \in U^* \otimes V^* \otimes W^*$.
Considering eq.~\eqref{slicerank:eq:evaluationcontraction}, we aim to understand the tensor contraction
\begin{equation}\label{slicerank:eq:tensorcontractionfI}
\langle f, T^{\otimes D}\rangle.
\end{equation}
For $T = \sum_{j=1}^r a_j \otimes b_j \otimes c_j$ we expand
\[
T^{\otimes D} = \sum_{J:[D]\to[r]} a_{J(1)} \otimes b_{J(1)} \otimes c_{J(1)} \otimes \cdots \otimes a_{J(D)} \otimes b_{J(D)} \otimes c_{J(D)}.
\]
For a list of vectors $v_1,\ldots,v_k$ of large enough dimension we define $\det(v_1,\ldots,v_k)$ to be the determinant of the $k \times k$ matrix whose columns are given by the top $k$ entries of each $v_i$.
Note that for $\mu=\la^t$ we have
\begin{equation}
\langle h_{\la}, v_1 \otimes \cdots \otimes v_{|\la|}\rangle = \det(v_1,\ldots,v_{\mu_1}) \det(v_{|\mu_1|+1},\ldots,v_{\mu_1+\mu_2}) \cdots \det(v_{|\mu|-\mu_{\la_1}+1},\ldots,v_{|\mu|}).
\end{equation}

\cite{BI:13} established that there is a basis of highest weight vectors in $\C[U\otimes V\otimes W]_D$ for which the contraction in \eqref{slicerank:eq:tensorcontractionfI} has a combinatorial description as follows:
Given $\la,\mu,\nu$ we consider a hypergraph on $D$ vertices with 3 types of hyperedges (we call these ``layers'' of hyperedges) such that every layer of hyperedges is a set partition of the vertices. Moreover, to every column in $\la$ we attach a hyperedge of layer 1 such that the number of vertices in the hyperedge equals the length of the column. We do the same for layer 2 and $\mu$ and layer 3 and $\nu$.
We end up with a hypergraph in which every vertex lies in exactly one hyperedge of layer 1, one hyperedge of layer 2, and one hyperedge of layer 3. We require that no two vertices lie in the same three hyperedges.
As described in \cite{BI:13},
for a hypergraph $H$
we get a highest weight vector $f_H$ (or $f_H=0$) such that for any tensor $T = \sum_{i=1}^R a_i \otimes b_i \otimes c_i$
the evaluation \eqref{slicerank:eq:tensorcontractionfI} can be written as follows:
\begin{multline}
\label{eq:sumofproductofdets}
\langle f_H, T^{\otimes D}\rangle = \sum_{J:[D]\to[R]} \prod_{\text{layer 1 hyperedge }e} \det(a_{J(e_1),\ldots,a_J(e_{|e|})})
\prod_{\text{layer 2 hyperedge }e} \det(b_{J(e_1),\ldots,b_J(e_{|e|})})\\
\prod_{\text{layer 3 hyperedge }e} \det(c_{J(e_1),\ldots,c_J(e_{|e|})})
,
\end{multline}
where we fixed an order on each hyperedge, $e_i$ is the $i$th vertex of $e$, and the determinant of an list of $m$ vectors of dimension $N$ is the determinant of the square matrix in which the columns are given by the top $m$ entries of the vectors.

\subsubsection*{The equations vanish on low slice rank}
\begin{proof}[Proof of Theorem~\ref{thm:equations}]
Let $n = \lceil \sqrt{N} \rceil$.
Since $\la_1 \leq n$, $\mu_1 \leq n$, $\nu_1 \leq n$,
the hypergraph can be reinterpreted as a cardinality $D$ subset of $[n]^3$,
where the slices in $x$-, $y$-, and $z$-direction of $[n]^3$ form the three layers of hyperedges: points share a layer $k$ hyperedge iff they share the $k$-th coordinate, see \cite{BI:13}.
We evaluate $f_H$ at a tensor $T = \sum_{i=1}^{r_1} \sum_{j=1}^n u_{i} \otimes v_{i,j} \otimes w_{i,j} + \sum_{i=1}^{r_2} \sum_{j=1}^n u_{i,j} \otimes v_{i} \otimes w_{i,j} + \sum_{i=1}^{r_3} \sum_{j=1}^n u_{i,j} \otimes v_{i,j} \otimes w_{i}$ of slice rank at most $r_1+r_2+r_3=r$.
In total we have $rn$ many triads.
Each triad has a \emph{parent vector}, which is either $u_i$, $v_i$, or $w_i$, depending of the triad's layer.

A map $J:[D]\to[r n]$ corresponds to a placement of triads on the $D$ vertices.
If two triads with the same parent vector share the hyperedge of their layer, then the determinant corresponding to this hyperedge vanishes (because determinants of matrices with a repeating column are zero), see \eqref{eq:sumofproductofdets}.
Hence we do not have to consider these summands $J$ in \eqref{eq:sumofproductofdets}.
Therefore in the remaining placements $J$, for each parent vector there can only be at most $n$ many placed triads with this parent vector.
Thus the evaluation at $T$ is zero if the number of parent vectors is less than $D/n$.
If the slice rank of $T$ at most $r$, then we can write $T$ using only $r$ many parent vectors.
Therefore $f_H$ vanishes on all points of slice rank less than $D/n$.
Since $H$ was arbitrary and since the vector space of all highest weight vectors of type $(\la,\mu,\nu)$ is generated by the $f_H$, Theorem~\ref{thm:equations} follows.
\end{proof}

\subsection{Equations from multiplicities}\label{subsec:slicerankeqmult}
In this section we search for equations for the slice rank variety by using the symmetry group to study the representation theoretic multiplicities in the orbits of $S_{N,r_1,r_2,r_3}$.
It turns out that we obtain precisely the same equations as in \ref{subsec:slicerankeqdesign}, but without the explicit construction of highest weight functions.

We consider the space $\bbC^{s_1}\otimes\bbC^{s_2}\otimes\bbC^{s_3}$ and recall $\bbC^{s_1} = \bbC^{r_1} \oplus (\bbC^{N})^{r_2} \oplus (\bbC^{N})^{r_3}$, $\bbC^{s_2} = (\bbC^{N})^{r_1} \oplus \bbC^{r_2} \oplus (\bbC^{N})^{r_3}$, $\bbC^{s_3} = (\bbC^{N})^{r_1} \oplus (\bbC^{N})^{r_2} \oplus \bbC^{r_3}$.
Let $G := \GL_{s_1} \times \GL_{s_2} \times \GL_{s_3}$.
Let $H$ denote the continuous part of the stabilizer of $S_{N,r_1,r_2,r_3}$ (i.e., ignoring the symmetric groups).
Let $(\GL_1 \times \GL_1 \times \GL_1)/\bbC^*$ := $\{(\alpha,\beta,\gamma)\mid \alpha\beta\gamma=1\}$.
Then $H$ is generated by $H_{\textsf{large}}$ and $H_{\textsf{small}}$,
where $H_{\textsf{large}} = \GL_n^{r_1} \times \GL_n^{r_2} \times \GL_n^{r_3}$ embedded into $G$ via
\[
(g_1^1,\ldots,g_1^{r_1};g_2^1,\ldots,g_2^{r_2};g_3^1,\ldots,g_3^{r_3}) \mapsto
\]
\begin{eqnarray*}
\big(
\diag(\textrm{Id}_{r_1},g_2^1,\ldots,g_2^{r_2},(g_3^1)^{-T},\ldots,(g_3^{r_3})^{-T})
\big);
\big(
\diag((g_1^1)^{-T},\ldots,(g_1^{r_1})^{-T},\textrm{Id}_{r_2},g_3^1,\ldots,g_3^{r_3}
\big);\\
\big(
\diag(g_1^1,\ldots,g_1^{r_1},(g_2^1)^{-T},\ldots,(g_2^{r_2})^{-T},\textrm{Id}_{r_3}
\big),
\end{eqnarray*}
and $H_{\textsf{small}}$ is $\big((\GL_1 \times \GL_1 \times \GL_1)/\bbC^*\big)^{r_1+r_2+r_3}$. The first factor $(\GL_1 \times \GL_1 \times \GL_1)/\bbC^*$ is embedded in $G$ via $(\alpha,\beta,\gamma) \mapsto$
\[
\Big(\diag(\alpha,1,\ldots,1,\textrm{Id}_{r_2},\textrm{Id}_{r_3});\diag(\underbrace{\beta,\beta,\ldots,\beta}_{n \text{ times}},1,\ldots,1,\textrm{Id}_{r_2},\textrm{Id}_{r_3});\diag(\underbrace{\gamma,\gamma,\ldots,\gamma}_{n \text{ times}},1,\ldots,1,\textrm{Id}_{r_2},\textrm{Id}_{r_3})\Big)
,
\]
while the other factors are embedded analogously.

$(\{\la\}_{s_1} \otimes \{\mu\}_{s_2} \otimes \{\nu\}_{s_3})^H = (\{\la\}_{s_1} \otimes \{\mu\}_{s_2} \otimes \{\nu\}_{s_3})^{H_{\textsf{small}}} \cap (\{\la\}_{s_1} \otimes \{\mu\}_{s_2} \otimes \{\nu\}_{s_3})^{H_{\textsf{large}}}$

We will treat $H_{\textsf{large}}$ and $H_{\textsf{small}}$ independently.

First, we decompose $(\{\la\}_{s_1} \otimes \{\mu\}_{s_2} \otimes \{\nu\}_{s_3})$ into irreducibles of $\GL_1^{r_1}\times(\GL_n)^{r_1}\times\GL_1^{r_2}\times(\GL_n)^{r_2}\times\GL_1^{r_3}\times(\GL_n)^{r_3}$ as follows.

We use the multi-Littlewood-Richardson rule (note that while $\la_1$ is the length of the first row of $\la$, we have that $\la_1''$ is a partition. We will not need to refer to its row lengths):
\begin{multline*}
\{\la\}_{s_1} = \bigoplus_{\substack{\ell_1,\ldots,\ell_{r_1} \\ \la_1'',\ldots,\la_{r_2}''\vdash_n\\\la_1''',\ldots,\la_{r_3}''' \vdash_n}} c_{(\ell_1),\ldots,(\ell_{r_1}),\la_1'',\ldots,\la_{r_2}'',\la_1''',\ldots,\la_{r_3}'''}^\la\\
\{(\ell_1)\} \otimes \cdots \otimes \{(\ell_{r_1})\} \otimes \{\la_1''\} \otimes \cdots \otimes \{\la_{r_2}''\} \otimes \{\la_1'''\} \otimes \cdots \otimes \{\la_{r_3}'''\}
\end{multline*}
\begin{multline*}
\{\mu\}_{s_2} = \bigoplus_{\substack{\mu_1',\ldots,\mu_{r_1}'\vdash_n\\
m_1,\ldots,m_{r_2} \\ \mu_1''',\ldots,\mu_{r_3}''' \vdash_n}} c_{\mu_1',\ldots,\mu_{r_1}',(m_1),\ldots,(m_{r_2}),\mu_1''',\ldots,\mu_{r_3}'''}^\mu \\
\{\mu_1'\} \otimes \cdots \otimes \{\mu_{r_1}'\} \otimes \{(m_1)\} \otimes \cdots \otimes \{(m_{r_2})\} \otimes \{\mu_1'''\} \otimes \cdots \otimes \{\mu_{r_3}'''\}
\end{multline*}
\begin{multline*}
\{\nu\}_{s_3} = \bigoplus_{\substack{\nu_1',\ldots,\nu_{r_1}'\vdash_n\\
\nu_1'',\ldots,\nu_{r_2}'' \vdash_n\\ n_1,\ldots,n_{r_3} }} c_{\nu_1',\ldots,\nu_{r_1}',\nu_1'',\ldots,\nu_{r_2}'',(n_1),\ldots,(n_{r_3})}^\nu \\
\{\nu_1'\} \otimes \cdots \otimes \{\nu_{r_1}'\} \otimes \{\nu_1''\} \otimes \cdots \otimes \{\nu_{r_2}''\} \otimes \{(n_1)\} \otimes \cdots \otimes \{(n_{r_3})\}
\end{multline*}
Each summand is a representation of $H_{\textsf{large}}$ and $H_{\textsf{small}}$.

The summands that are invariant under $H_{\textsf{large}}$ are the ones for which $\mu_i'=\nu_i'$ and $\nu_i''=\la_i''$ and $\la_i'''=\mu_i'''$. This follows from the fact that $\dim(\{\la\}\otimes\{\mu^*\})^{\GL_n}=1$ if $\la=\mu$ and 0 otherwise.

The summands that are invariant under $H_{\textsf{small}}$ are the ones
for which $|\mu_i'|=|\nu_i'|=\ell_i$ and $|\nu_i''|=|\la_i'|=m_i$ and $|\la_i'''|=|\mu_i'''|=n_i$.

Hence the dimension of the $H$-invariant space in $\{\la\}_{s_1} \otimes \{\mu\}_{s_2} \otimes \{\nu\}_{s_3}$ is
\[
\sum_{\substack{\la_1',\ldots,\la_{r_1}'\\\mu_1'',\ldots,\mu_{r_2}''\\\nu_1''',\ldots,\nu_{r_3}'''}}
c_{(|\la_1'|),\ldots,(|\la_{r_1}'|),\mu_1'',\ldots,\mu_{r_2}'',\nu_1''',\ldots,\nu_{r_3}'''}^\la
\cdot
c_{\la_1',\ldots,\la_{r_1}',(|\mu_1''|),\ldots,(|\mu_{r_2}''|),\nu_1''',\ldots,\nu_{r_3}'''}^\mu
\cdot
c_{\la_1',\ldots,\la_{r_1}',\mu_1'',\ldots,\mu_{r_2}'',(|\nu_1'''|),\ldots,(|\nu_{r_3}'''|)}^\nu
\]

We are interested in the case in which all summands vanish.
Note that a Littlewood-Richardson coefficient is zero if at least one of its lower partition parameters is not a subpartition of the upper partition parameter. In particular, for nonzeroness we require
$|\la_1'|,\ldots,|\la_{r_1}'| \leq \la_1$ and
$|\mu_1''|,\ldots,|\mu_{r_2}''| \leq \mu_1$ and
$|\nu_1'''|,\ldots,|\nu_{r_3}'''| \leq \nu_1$.
Let $\la_1,\mu_1,\nu_1 \leq k$.

Let $|\la|=|\mu|=|\nu|=D$ be the degree. (Clearly we have $D \leq kn$, because otherwise there are no such $\la$, $\mu$, $\nu$.)
Note that
\[
|\la_1'|+\cdots+|\la_{r_1}'|=|\mu_1''|+\cdots+|\mu_{r_2}''|=|\nu_1'''|+\ldots+|\nu_{r_3}'''| = D.
\]
Hence for nonzeroness we require $D = |\la_1'|+\cdots+|\la_{r_1}'| \leq r_1 k$. Analogously $D \leq r_2 k$ and $d \leq r_3 k$.

Hence if $\la_1,\mu_1,\nu_1 \leq k$, then the irreducible $G$-representation $\{\la,\mu,\nu\}$ does not occur in 
$\bbC[G S_{N,r_1,r_2,r_3}]$ (and hence not in $\bbC[\overline{G S_{N,r_1,r_2,r_3}}]$) if $r_1>D/k$ or $r_2>D/k$ or $r_3>D/k$, in particular if $r_1+r_2+r_3>D/k$.

For the sake of comparing these equations to the equations found in Section~\ref{subsec:slicerankeqdesign}, let $N=n^2$ and $k=n$.
Then we get a degree $D$ equation vanishing on $\overline{G S_{N,r_1,r_2,r_3}}$ if $r_1 > D/n$ or $r_2 > D/n$ or $r_3 > D/n$.
In particular, $(\la,\mu,\nu)$ does not occur in
$\bbC[G S_{N,r_1,r_2,r_3}]$
(and hence also not in $\bbC[\overline{G S_{N,r_1,r_2,r_3}}]$)
if $r_1+r_2+r_3 > D/n$,
which is precisely what we found in Section~\ref{subsec:slicerankeqdesign}, where we constructed the equations explicitly.

\section{Homogeneous minrank problem}\label{sec:homogminrank}

We consider the following problem:
given a tuple of matrices $A_1, \dots, A_k$ of the same size $m \times n$ and a number $r$, does there exist a nonzero linear combination $x_1 A_1 + \dots + x_k A_k$ with rank at most $r$?
This is a homogeneous variant of the MinRank problem, where instead of a linear combination we have an affine expression $A_0 + x_1 A_1 + \dots + x_k A_k$.
A restricted variant of this problem was first considered in~\cite{DBLP:journals/jcss/BussFS99}, where it is proven that the problem is $\NP$-hard.
The related problem of low rank matrix completion is widely studied in optimization. % \todo{Find references}

Clearly, the answer depends on the field from which we take the coefficients of the linear combination. For example, the pair of matrices
\[
  A_1 = \begin{bmatrix} 1 & 0 \\ 0 & 1 \end{bmatrix}, \,\, A_2 = \begin{bmatrix} 0 & 1 \\ -1 & 0 \end{bmatrix}
\]
has no nontrivial linear combinations of rank $1$ over $\bbR$, but over $\bbC$ we have $\rk (A_1 + i A_2) = 1$.
We will mostly work over algebraically closed fields such as $\bbC$, but many results are also true over other fields.

Let $F$ be a field.
Instead of talking about matrices $A_i, \dots, A_k \in F^{m \times n}$,
we can also phrase the homogeneous minrank problem in terms of a linear subspace $\left< A_1, \dots, A_k \right>$, a matrix of linear forms $A \colon F^k \to F^{m \times n}$ where $A(x) = \sum_{i = 1}^k x_i A_i$ or a tensor $T \in F^k \otimes F^m \otimes F^n$ such that $T = \sum e_i \otimes A_i$.
We will use the tensor language.

Recall the definition of minrank.
\begin{definition}
  Let $U, V, W$ be finite-dimensional vector spaces over some field $F$. The minrank of a tensor $T \in U \otimes V \otimes W$ is the minimal number $r$ such that there exists a nonzero $x \in U^*$ with $\rk (Tx) = r$.
\end{definition}

Let $S$ be a finite or countable subset of $F$.
\begin{problem}{$\HMinRk_{S, F}$}
  Given a tensor $T \in F^{k \times m \times n}$ with all components in $S$ and a number $r$,
  decide if the minrank of $T$ is at most $r$.
\end{problem}

In section \ref{sec:hard} we will prove that this problem is $\NP$-hard. Moreover, it is hard even if we look for rank $1$ slices.

\begin{problem}{$\HMinRkU_{S, F}$}
  Given a tensor $T \in F^{k \times m \times n}$ with all components in $S$,
  decide if the minrank of $T$ is at most $1$.
\end{problem}

\section{Geometric description of Minrank varieties} \label{sec:geo}

Over algebraically closed fields, the answer to the homogeneous minrank problem is determined by membership in a certain affine variety.

\begin{theorem} \label{thm:closed}
  Let $U$, $V$, $W$ be vector spaces over an algebraically closed field $F$.
  The set of all tensors $T \in U \otimes V \otimes W$ with minrank at most $r$ is Zariski closed.
\end{theorem}
\begin{proof}
Define an affine variety
\[
  \cX_{U \otimes V \otimes W, r} = \{ (T, x) \in (U \otimes V \otimes W) \times U^* \mid \rk (Tx) \leq r \}.
\]
Since the condition $\rk(Tx) \leq r$ is scale-invariant with respect to both $T$ and $x$, we can define the corresponding projective variety 
\[
\begin{aligned}
  \cPX_{U \otimes V \otimes W, r} = \{ ([T], [x]) \in \bbP(U \otimes V \otimes W) \times \bbP U^* \mid \rk (Tx) \leq r \} \subset \bbP(U \otimes V \otimes W) \times \bbP U^*
\end{aligned}
\]
Let $\pi \colon \bbP(U \otimes V \otimes W) \times \bbP U^* \to \bbP(U \otimes V \otimes W)$ be the projection onto the first component of the product.
Consider the image of $\cPX_{U \otimes V \otimes W, r}$ under $\pi$:
\[
    \pi \cPX_{U \otimes V \otimes W, r} = \{[T] \in \bbP(U \otimes V \otimes W) \mid  \exists x \neq 0 \colon \rk(Tx) \leq r\}
\]
As an image of a projective variety, it is a closed subvariety of $\bbP(U \otimes V \otimes W)$ (see e.~g.~\cite[Thm.~1.10]{Shafarevich1994}).
The affine cone over this subvariety is therefore also closed. This affine cone is exactly the set of tensors of minrank at most $r$.
\end{proof}

\begin{definition}
  We call the projective variety
  \[
    \cPM_{U \otimes V \otimes W, r} = \{[T] \in \bbP(U \otimes V \otimes W) \mid  \exists x \neq 0 \colon \rk(Tx) \leq r\}
  \]
  the \emph{projective minrank variety}, and the corresponding affine cone
  \[
    \cM_{U \otimes V \otimes W, r} = \{ T \in U \otimes V \otimes W \mid \exists x \neq 0 \colon \rk(Tx) \leq r \}
  \]
  the \emph{affine minrank variety}, or just the \emph{minrank variety}.
  We omit the index $U \otimes V \otimes W$ if it is clear from context.
\end{definition}

Some simple properties of minrank varieties follow directly from the definition:

\begin{lemma}
  \label{lemma-inheritance-varieties}
  Let $V'$ and $W'$ be subspaces of $V$ and $W$ respectively. Then 
  \[
    \cM_{U \otimes V' \otimes W', r} = \cM_{U \otimes V \otimes W, r} \cap (U \otimes V' \otimes W').
  \]
\end{lemma}
\begin{proof}
  Trivial. A tensor lies in $\cM_{U \otimes V' \otimes W', r}$ iff it is an element of the space $U \otimes V' \otimes W'$ and has minrank at most $r$, i.~e., lies in $\cM_{U \otimes V \otimes W, r}$.
\end{proof}

\begin{lemma}
  Let $\dim U = k$, $\dim V = n$ and $\dim W > s = n(k - 1) + r$. Then
  \[
    \cM_{U \otimes V \otimes W, r} = \bigcup_{\substack{W' \subset W \\ \dim W' = s}} \cM_{U \otimes V \otimes W', r}.
  \]
%\christian{Notation is inconsistent. We should always have $\dim V = m$ and $\dim W = n$.}
\end{lemma}
\begin{proof}
  Let $T$ be a tensor in $\cM_{U \otimes V \otimes W, r}$ and $x_1$ be a nonzero vector in $U^*$ such that $\rk(Tx_1) \leq r$.
  Choose $x_2, \dots, x_k$ such that $\{x_i\}$ is a basis of $U^*$ and set $A_i = Tx_i \in V \otimes W$.
  Since $\rk A_1 \leq r$, there exists a subspace $W_1 \subset W$ of dimension at most $r$ such that $A_1 \in V \otimes W_1$.
  Analogously, for $i > 1$ we have $A_i \in V \otimes W_i$ for some subspace $W_i \subset W$ of dimension at most $n$.
  The sum $W'$ of all $W_i$ is a subspace of dimension at most $s$. We extend it to dimension $s$ in arbitrary way if needed.
  The tensor $T$ lies in $U \otimes V \otimes W'$ and, therefore, in $\cM_{U \otimes V \otimes W', r}$.
\end{proof}

\begin{lemma}
  The variety $\cM_{U \otimes V \otimes W, r}$ is invariant under the standard action of $\GL(U) \times \GL(V) \times \GL(W)$ on $U \otimes V \otimes W$.
\end{lemma}
\begin{proof}
  Straightforward. If $\rk(Tx) \leq r$, then $(F \otimes G \otimes H) T \cdot (F x) = (G \otimes H) (Tx)$ also has rank at most $r$ (here $F x$ denotes the dual action of $GL(U)$ on $U^*$).
\end{proof}

\subsection{Minrank varieties and orbit closures}

The minrank varieties are related to orbit closures of some tensors.
Let $L = (F^n)^{\oplus (k - 1)} \oplus F^r$ be a vector space of dimension $s = n(k - 1) + r$ decomposed into $k$ summands of dimension $n$ each, except the first one, which is of dimension $r$.
%\christian{Do we tacitly assume $n \geq m$? Should we replace $n$ by $\max(n,m)$?}
%\vladimir{No, it works for any $m < s$}
Let $L_i$ be the $i$-th summand and denote the standard basis of $L_i$ by $e_{ij}$, $1 \leq j \leq \dim L_i$.
Let $U = F^k$ be a $k$-dimensional space with a standard basis $e_i$. Define the tensor $T_{k, n, r} \in U \otimes L \otimes L$ as
\[
  T_{k, n, r} = e_1 \otimes (\sum_{j = 1}^r e_{1j} \otimes e_{1j}) + \sum_{i = 2}^{k} e_i \otimes (\sum_{j = 1}^n e_{ij} \otimes e_{ij}),
\]
that is, $i$-th layer of $T$ is a block matrix with the only nonzero block being a diagonal matrix in $L_i \otimes L_i$.

The group $\GL_k \times \GL_s \times \GL_s$ acts in a usual way on $U \otimes L \otimes L$.
The minrank variety $\cM_{r}$ can be defined using the orbit closure of $T_{k, n, r}$:
\begin{theorem}
  Let $V$ be an $n$-dimensional subspace of $L$.
  Then \[\cM_{U \otimes V \otimes L, r} = \overline{(\GL_k \times \GL_s \times \GL_s) T_{k, n, r}} \cap (U \otimes V \otimes L).\]
\end{theorem}
\begin{proof}
  We have $T_{k, n, r} \in \cM_{U \otimes L \otimes L, r}$. Since the minrank variety is invariant, the entire orbit $(\GL_k \times \GL_s \times \GL_s) T_{k, n, r}$ lies in it. Since the minrank variety is Zariski closed, it also contains the orbit closure.
  By Lemma~\ref{lemma-inheritance-varieties} we have $\overline{(\GL_k \times \GL_s \times \GL_s) T_{k, n, r}} \cap (U \otimes V \otimes L) \subset \cM_{U \otimes V \otimes L, r}$.
  
  Conversely, let $T \in \cM_{U \otimes V \otimes L, r}$. We can write $T$ as $\sum_{i = 1}^k u_i \otimes A_i$ where $\{u_i\}$ is some basis of $U$ and $A_1$ is a slice with $\rk(A_1) \leq r$.

  Since $\rk(A_1) \leq r$, it can be presented as $(P_1 \otimes Q_1)(\sum_{j = 1}^r e_{1j} \otimes e_{1j})$ where $P_1 \colon L_1 \to V$ and $Q_1 \colon L_1 \to L$ are some linear maps.
  Analogously, for $i > 1$ we have $\rk(A_i) \leq \dim V = n$ and $A_i = (P_i \otimes Q_i)(\sum_{j = 1}^n e_{ij} \otimes e_{ij})$ for some $P_i \colon L_i \to V$ and $Q_i \colon L_i \to L$.
  Let $P \colon L \to V$ and $Q \colon L \to L$ be the linear maps which are equal to $P_i$ and $Q_i$ respectively when restricted to $L_i$.
  Let $R \colon U \to U$ be the map sending each $e_i$ to the corresponding $u_i$.
  Then $T = (R \otimes P \otimes Q) T_{k, n, r}$.
  The closure of $\GL(L)$ consists of all linear endomorphisms of $L$ and thus contains $P$ and $Q$.
  Therefore, $T$ lies in the closure $\overline{(\GL_k \times \GL_s \times \GL_s) T_{k, n, r}}$.
\end{proof}

\begin{corollary} \label{cor:orbit}
  Let $\dim U = k$ and $\dim V = n$. Suppose $V$ and $W$ are subspaces of a vector space $L$ of dimension $s = (k - 1)n + r$.
  Then \[\cM_{U \otimes V \otimes W, r} = \overline{(\GL(U) \times \GL(L) \times \GL(L)) T_{k, n, r}} \cap (U \otimes V \otimes W).\]
\end{corollary}

\begin{theorem}\label{thm:stab}
  If $r < n$, then the stabilizer of $T_{k, n, r}$ in $\GL_k \times \GL_s \times \GL_s$ is isomorphic to $(\GL_r \times \GL_1) \times (\GL_n \times \GL_1)^{k - 1} \rtimes \mathfrak{S}_{k - 1}$. The element 
  \[
    (Z_1, z_1, \dots, Z_k, z_k) \in (\GL_r \times \GL_1) \times (\GL_n \times \GL_1)^{k - 1}
  \]
  is included into $\GL_k \times \GL_s \times \GL_s$ via 
  \[
    (\diag(z_1, \dots, z_k), \diag(Z_1, \dots, Z_k), \diag((z_1 Z_1)^{-\trans}, \dots, (z_k Z_k)^{-\trans}))
  \]
  and the $\mathfrak{S}_{k - 1}$ factor permutes the last $k - 1$ coordinates of $F^k$ and the last $k - 1$ summands of $W$ simultaneously.
\end{theorem}
\begin{proof}
  Let $(A, B, C) \in \stab T_{k, n, r}$, so that $(A \otimes B \otimes C) T_{k, n, r} = T_{k, n, r}$.

  Let $T_i = \sum_{j = 1}^{\dim W_i} e_{1j} \otimes e_{1j}$ be the slices of $T_{k, n, r}$, so $T = \sum_{i = 1}^k e_i \otimes T_i$ and
  \[(A \otimes B \otimes C) T = \sum_{i = 1}^k (Ae_i) \otimes (B \otimes C)T_i = \sum_{i = 1}^k e_i \otimes (B \otimes C)(\sum_{j = 1}^k {a_{ij}} T_j).\]
  Note that the rank of $\sum_{j = 1}^k {a_{ij}} T_j$ and, consequently, of the $i$-th slice of $(A \otimes B \otimes C) T_{k, n, r}$,
  is equal to $sr + qn$, where $s = 0$ if $a_{i1} = 0$ and $s = 1$ otherwise, and $q$ is the number of nonzero entries among $a_{i2}, \dots, a_{ik}$.
  Therefore, $A$ contains only one nonzero entry in each row, and in the first row the nonzero entry is in the first column,
  otherwise the ranks of slices of $T_{k, n, r}$ and $(A \otimes B \otimes C) T_{k, n, r}$ do not match.
  Thus, $A$ is a product of a diagonal matrix and a permutation matrix corresponding to some permutation $\sigma$ of the last $k - 1$ coordinates of $F^k$.
  
  Let $P_{\sigma}$ be an element of $\GL_k \times \GL_s \times GL_s$ which permutes last $k - 1$ coordinates of $F^k$ and last $k - 1$ summands of $W$ according to the permutation $\sigma$. It is easy to see that $P_{\sigma} \in \stab T_{k, n, r}$.
  Thus, $(A, B, C) P_{\sigma}^{-1} = (\hat{A}, \hat{B}, \hat{C})$ is also in $\stab T_{k, n, r}$.
  From the previous discussion, $\hat{A}$ is a diagonal matrix $\diag(z_1, \dots, z_k)$.
  Let $\hat{A'} \in \GL_s$ be the linear map which scales elements of $W_i$ by $z_i$ for each $i$.
  $(\hat{A}^{-1}, \id, \hat{A'})$ also preserves $T_{k, n, r}$.
  Therefore, $(\hat{A}, \hat{B}, \hat{C}) \cdot (\hat{A}^{-1}, \id, \hat{A'}) = (\id, \hat{B}, \check{C})$ is in $\stab T_{k, n, r}$ 

  Now, since the first component of $(\id, \hat{B}, \check{C})$ is the identity, it preserves $T_{k, n, r}$ if and only if $\hat{B} \otimes \check{C}$ preserves each slice $T_i$.
  If it preserves each slice, it also preserves its sum $\sum_{i = 1}^k T_i = \sum_{i = 1}^k \sum_{j = 1}^{\dim W_i} e_{ij} \otimes e_{ij}$, the full rank diagonal matrix.
  Therefore, by the Lemma~\ref{lem:stabidentity}, $\check{C} = \hat{B}^{-\trans}$.
  Consider $\hat{B}$ as a block matrix $(B_{ij})$ according to the decomposition of $W$ into $W_i$.
  If $\hat{B}$ has a nonzero off-diagonal block $B_{ij}$, then $(\hat{B} \otimes \hat{B}^{-\trans}) T_j$ has nonzero elements in the $i$-th block of rows, and thus is not equal to $T_j$.
  Therefore, $\hat{B}$ is a block diagonal matrix $\diag(Z_1, \dots, Z_k)$.
  Using the previous lemma, we see that any such $\hat{B}$ gives rise to $(\id, \hat{B}, \hat{B}^{-\trans}) \in \stab T_{k, n, r}$.

  We decomposed an arbitrary element $A \otimes B \otimes C \in \stab T_{k, n, r}$ into a product of three special elements $(\id, \diag(Z_1, \dots, Z_k), \diag(Z_1, \dots, Z_k)^{-\trans})$, $(\diag(z_1, \dots, z_k), \id, \diag(z_1 \id, \dots, z_k \id)^{-1})$ and $P_{\sigma}$ for some permutation $\sigma \in \mathfrak{S}_{k - 1}$.
  These three types of elements correspond to three subgroups of $\stab T_{k, n, r}$.
  The subgroups intersect only by identity; elements of the first two types commute, and the conjugation with $P_{\sigma}$ permutes $z_i$ and $Z_i$ according to $\sigma$, so the product of the first subgroups is direct, and the last product is semidirect.
\end{proof}

\begin{theorem} \label{thm:sym}
  Suppose $T$ is a tensor in $F^k \otimes W \otimes W$. If $\stab T = \stab T_{k, n, r}$, then $T$ lies in the orbit $(\GL_k \times \GL_s \times \GL_s) T_{k, n, r}$. If $\stab T \supset \stab T_{k, n, r}$, then $T \in \overline{(\GL_k \times \GL_s \times \GL_s) T_{k, n, r}}$
\end{theorem}
\begin{proof}
  Suppose $T$ is stabilised by $\stab T_{k, n, r}$.
  Let $T_1, \dots, T_k$ be the slices of $T$. Decompose them into blocks $T_i = (T_{ijk})$ according to the decomposition of $W$ into $W_i$.
  
  Let $A_i(\lambda) \colon F^k \to F^k$ be the map which scales the $i$-th coordinate by $\lambda$, leaving other in place, and $B_i(\lambda) \colon W \to W$ be the map which scales $W_i$ by $\lambda$ and acts like identity on the other summands.
  Applying to $T$ the transformation $(A_i(\lambda^{-2}), B_i(\lambda), B_i(\lambda)) \in \stab T_{k, n, r}$,
  we see that all blocks of $T_i$ except $T_{iii}$ are zero, as they are multiplied by a coefficient $\lambda^{-2}$ or $\lambda^{-1}$ in this transformation.

  Applying $(\id, \diag(Z_1, \dots, Z_k), \diag(Z_1, \dots, Z_k)^{-\trans})$ with arbitrary $Z_i$ to $T$, we obtain that each $T_{iii}$ has the form $a_i \sum_{j = 1}^{\dim W_i} e_{ij} \otimes e_{ij}$.
  Applying permutations of the last $m-1$ blocks $T_{iii}$, we see that $a_2 = \cdots = a_k$.

  Therefore
  \[
    T = a_1 \sum_{j = 1}^{n} e_1 \otimes e_{1j} \otimes e_{1j} + a_2 \sum_{i = 2}^{k} \sum_{j = 1}^{n} e_i \otimes e_{ij} \otimes e_{ij}.
  \]
  If both $a_1$ and $a_2$ are nonzero, then $T = (\diag(a_1, a_2, \dots, a_2) \otimes \id \otimes \id) T_{k, n, r}$ lies in the orbit of $T_{k, n, r}$.
  In this case $\stab T = \stab T_{k, n, r}$.
  The closure of the set of tensors of this form with $a_1 \neq 0$ and $a_2 \neq 0$ includes the cases when $a_1$ or $a_2$ are zero.
  In these border cases, $T$ has more symmetries than $T_{m, n, r}$, for example, multiplication of the zero blocks by an arbitrary matrix.
\end{proof}

\section{Ideals of minrank varieties} \label{sec:ideal}

In this section we will consider tensors over $\bbC$ (or over algebraically closed field of characteristic 0).
For an affine variety $\cA$, we denote its ideal by $I(\cA)$ and its coordinate ring by $\bbC[\cA]$.
It is convenient to work with a tensor space $U^* \otimes V^* \otimes W^*$. The algebra of polynomials on this space is $\bbC[U \otimes V \otimes W]$.
Since minrank varieties $\cM_{U^* \otimes V^* \otimes W^*, r}$ are scale-invariant, their ideals and coordinate rings inherit grading from $\bbC[U \otimes V \otimes W]$.
Since they are invariant under the action of $\GL(U) \times \GL(V) \times \GL(W)$, this group also acts on $I(\cM_r)$ and $\bbC[\cM_r]$.
We study representation-theoretic properties of specific equations in $I(\cM_r)$ and multiplicities of irreducible representations in $I(\cM_r)$ and $\bbC[\cM_r]$.
Irreducible representations in $\bbC[U \otimes V \otimes W]_d$ are indexed by triples $(\lambda,\mu,\nu)$ of partitions of $D$,
and the multiplicity of the type $(\lambda,\mu,\nu)$ in $\bbC[U \otimes V \otimes W]_D$ is called the \emph{Kronecker coefficient}.
Finding a combinatorial interpretation for the Kronecker coefficient is Problem~10 in \cite{sta:00}.

Here we describe some polynomials in the ideals of minrank varieties and study their representation-theoretic properties.
The methods used in Section~\ref{sec:eq-rect} are similar to the ones used in Section~\ref{subsec:slicerankeqdesign}; and the methods used in Section~\ref{sec:eq-mult} are similar to the ones used in  Section~\ref{subsec:slicerankeqmult}.

Throughout the section we assume $\dim U = k$, $\dim V = m$, $\dim W = n$.

\subsection{Basic equations}
\label{sec:eq-basic}
For a vector space $V$ let $S^p V$ denote its $p$th symmetric power, which corresponds to the vector space of homogeneous degree $p$ polynomials in $\dim V$ variables.
Moreover, let $\wedge^p V$ denote the $p$th exterior power of $V$.

Let $T \in U^* \otimes V^* \otimes W^*$ be a tensor.
The condition $\rk (Tx) \leq r$ is equivalent to the vanishing of all $(r + 1) \times (r + 1)$ minors of $Tx$, or, equivalently, of the matrix $(Tx)^{\wedge (r + 1)} \in \Lambda^{r + 1}V^* \otimes \Lambda^{r + 1}W^*$, entries of which are multiples of the minors in question.
All the minors are homogeneous polynomials of degree $r + 1$ with respect to $x$, so the polynomial map sending $x$ to $(Tx)^{\wedge (r + 1)}$ extends to a linear map $M_{T, r} \colon S^{r + 1} U \to \Lambda^{r + 1}V^* \otimes \Lambda^{r + 1}W^*$ such that $M_{T, r}(x^{\otimes (r + 1)}) = (Tx)^{\wedge (r + 1)}$.

\begin{theorem}
  \label{eq:simple-kernel}
  Let $s = \dim S^{r + 1}U = \binom{k + r}{r + 1}$. Size $s$ minors of $M_{T, r}$ lie in the ideal $I(\cM_r)$.
\end{theorem}
\begin{proof}
  If $T \in \cM_r$, then there exists a nonzero rank 1 symmetric tensor $x^{\otimes (r + 1)}$ on which $M_{T, r}$ vanishes.
  In particular, it means that the rank of $M_{T, r}$ is less than the dimension of the source space $S^{r + 1}U$, so the $s \times s$ minors of $M_{T, r}$ vanish.
\end{proof}

The map $M_{T, r}$ is represented by a tensor in $S^{r + 1} U^* \otimes (\Lambda^{r + 1}V^* \otimes \Lambda^{r + 1}W^*)$.
Therefore, the $\GL(U) \times \GL(V) \times \GL(V)$-representation generated by size $s$ minors of $M_{T, r}$ is the image of $\Lambda^s S^{r + 1} U \otimes \Lambda^s (\Lambda^{r + 1} V \otimes \Lambda^{r + 1} W) \subset S^s S^r (U \otimes V \otimes W)$ under the symmetrization map from $S^s S^{r + 1} (U \otimes V \otimes W)$ to $S^{s(r + 1)} (U \otimes V \otimes W)$.

In the special case $n = m + k - 1$, $r = m - 1$ the dimension $\dim (\Lambda^{r + 1}V^* \otimes \Lambda^{r + 1}W^*)$ coincides with $s$ and the resulting polynomial in $I(\cM_r)$ is the determinant of the square matrix $M_{T,r}$.
This polynomial is the \emph{hyperdeterminant of boundary format}, a very special case of hyperdeterminant polynomials considered in~\cite{Gelfand1994}.
It is a $\SL(U) \times \SL(V) \times \SL(W)$-invariant of degree $sm$.

\subsection{Koszul flattenings}
\label{sec:eq-koszul}

Another family of equations can be constructed using so called Koszul flattenings, a special case of Young flattenings introduced in~\cite{Landsberg2013} in relation to secant varieties.

For any integer $p$ the antisymmetrization map $U^* \otimes \Lambda^p U^* \to \Lambda^{p + 1} U^*$ gives rise to a linear map $I_p \colon U^* \to \Lambda^p U \otimes \Lambda^{p + 1} U^*$ sending each $u \in U$ to a tensor representing the map $x \mapsto u \wedge x$.
Applying this map to the first multiplicand of the tensor $T \in U^* \otimes V^* \otimes W^*$, we get a tensor $I_p T \in \Lambda^p U \otimes \Lambda^{p + 1} U^* \otimes V^* \otimes W^*$. Rearranging tensor multiplicands, we get the Koszul flattening $F_{T, p} \colon \Lambda^p U^* \otimes V \to \Lambda^{p + 1} U^* \otimes W^*$. If $T = \sum_i x_i \otimes A_i$, then $F_{T, p}$ sends $y \otimes v$ to $\sum_i (x_i \wedge y) \otimes A_i v$.

\begin{theorem} \label{thm:koszul}
  Let $0 < p < k$. If $T \in \cM_{U^* \otimes V^* \otimes W^*, r}$, then
  $$\rk F_{T, p} \leq \binom{k - 1}{p}(r + \min(m, \frac{k - p - 1}{p + 1} n) + \min(n, \frac{p}{k - p}m)).$$
\end{theorem}
\begin{proof}
  If $T \in \cM_r$, then it can be written as $x \otimes A + T'$ where $x \in U^*$, $\rk(A) \leq r$ and $T' \in H \otimes V^* \otimes W^*$ for some hyperplane $H \subset U^*$ which does not contain $x$.
  The space $\Lambda^p U^*$ can be decomposed as $\Lambda^p U^* = \Lambda^p H \oplus (u \wedge \Lambda^{p - 1} H)$.
  Similarly, $\Lambda^{p + 1} U^* = \Lambda^{p + 1} H \oplus (u \wedge \Lambda^{p} H)$.
  Using these decomposition, the map $F_{T, p}$ is given by the $2 \times 2$ block matrix
  \[
    \raisebox{-2.5em}{$
    \begin{bmatrix}
      & \smash{\raisebox{2em}{$\Lambda^p H \otimes V$}} & \smash{\raisebox{2em}{$(u \wedge \Lambda^{p - 1} H) \otimes V$}} \\[-1.2em]
      \makebox[0pt][r]{$\Lambda^{p + 1} H \otimes W^*$\hspace{2em}} & F_{T', p} & 0 \\
      \makebox[0pt][r]{$(u \wedge \Lambda^{p} H) \otimes W^*$\hspace{2em}} & (I_p u) \otimes A & F_{T', p-1} \\
    \end{bmatrix}
    $}
  \]
  The rank of this matrix is at most the sum of the ranks of the three blocks.
  The block $I_p u \otimes A$ has rank $\dim \Lambda^p H \cdot \rk A = \binom{k - 1}{p}r$.
  The ranks of other two blocks are bounded by their sizes: $\binom{k - 1}{p} m \times \binom{k - 1}{p + 1} n$ for the top left block and $\binom{k - 1}{p - 1} m \times \binom{k - 1}{p} n$ for the bottom right.
  Factoring out the $\binom{k - 1}{p}$, we get the expression from the theorem statement.
\end{proof}
\begin{corollary}
  If $n = \frac{p + 1}{k - p}m$, and $r < \frac{m}{k - p}$, then $\rk F_{T, p} < \dim (\Lambda^p U^* \otimes V)$.
\end{corollary}
\begin{proof}
  In this case we have $\dim \Lambda^p U^* \otimes V = \binom{k}{p} m = \binom{k - 1}{p} \frac{k}{k - p} m$ and
  \[\rk F_{T, p} \leq \binom{k - 1}{p}(r + \frac{k - p - 1}{p + 1} n + \frac{p}{k - p} m) = \binom{k - 1}{p} (r + \frac{k - 1}{k - p} m) < \binom{k - 1}{p} \frac{k}{k - p} m.\qedhere\]
\end{proof}

In particular, this construction works in the case $k = 2p + 1$, $n = m$.
Landsberg~\cite{LANDSBERG20153677} showed that for a generic $(2p + 1) \times m \times m$ tensor $T$ the flattening $F_{T, p}$ has maximal possible rank $\binom{k}{p} m$, so the corollary implies that the minors of $F_{T, p}$ of corresponding size give nontrivial equations for $\cM_r$.

\subsection{Equations from rectangular designs}
\label{sec:eq-rect}

For $\alpha, \beta \in \mathbb N$, fix $\max(\alpha,\beta)$ many vectors $\mathcal U := \{u_1,\ldots,u_{\max(\alpha,\beta)}\}$ in $\bbC^{\alpha}$.
An $\alpha \times \beta$ Latin Rectangle for $\mathcal U$ is an $\alpha \times \beta$ matrix $(A_{i,j})_{i,j}$, where in each row and in each column we have each entry from $\mathcal V$ at most once.
Note that if $\alpha\geq\beta$, then each column contains each vector exactly once.
For each column $A_{.,j}$ we define its determinant $\det(A_{.,j})$ as the determinant of the $\alpha \times \alpha$ matrix whose columns are given by
the list of vectors $(A_{1,j},\ldots,A_{\alpha,j})$.
The column-determinant of a Latin rectangle $A$ is defined as
\[
\coldet(A) := \prod_{j=1}^\beta \det(A_{.,j}).
\]

For every $\alpha \in \mathbb N$ and every even $\beta \in \mathbb N$ such that $\alpha \leq \beta$, the Latin Rectangle Conjecture $\LatRect(\alpha,\beta)$ can be stated as follows.
\begin{conjecture}
%Let $\beta$ be even and $\alpha\leq \beta$ be arbitrary.
%Choose a set $\mathcal U$ of $\max(\alpha,\beta)$ many vectors in $\bbC^{\alpha}$ generically.
Choose a set $\mathcal U$ of $\beta$ many vectors in $\bbC^{\alpha}$ generically.
Then $\sum_L \coldet(L) \neq 0$, where the sum is over all Latin Rectangles for $\mathcal U$.
\end{conjecture}
$\LatRect(1,\beta)$ is trivially true for all $\beta$.
The fact that $\LatRect(2,\beta)$ is true for all even $\beta$ follows from the classical proof of Hermite's reciprocity theorem in representation theory \cite{Her:1854}.
The fact that $\LatRect(\alpha,\beta)$ is true for all even $\beta$ and $\alpha \leq 5$ follows from recent work on Foulkes' conjecture, see \cite{MN:05,McK:08,CIM:15}.
%Note that if $\alpha \geq \beta$, then all columns have the same determinant up to sign.
$\LatRect(\beta,\beta)$ can be easily seen to be equivalent to the Alon-Tarsi conjecture \cite{AT:92}, which is equivalent to the Huang-Rota conjecture \cite{HR:94}, and is known to be true for $\beta\in\{p+1,p-1\mid p \text{ an odd prime number}\}$, in particular for all even $2 \leq \beta \leq 24$ \cite{Dri:98,Gly:10}.
$\LatRect(\beta,\beta)$ implies $\LatRect(\alpha,\beta)$ for $\alpha \leq \beta$, as was shown by S.\ Kumar
as part of his work on geometric complexity theory \cite{Kum:15}.

% If we allow $\alpha > \beta$, then we have to choose $\alpha$ many vectors in $\bbC^{\alpha}$ generically.
% In this version no counterexample is known either.
% Indeed, 
% $\LatRect(\alpha,2)$ is true for $\alpha\geq 2$, because the numbers of even and odd derangements on $\alpha$ symbols differ \christian{cite}.

\begin{theorem}\label{thm:equationsfromrect}
For the sake of notational simplicity let $m \leq n$.
%If $m \geq k$, $n>rk$, and if $\LatRect(n,k)$, then there exists an irreducible representation of nontrivial equations for $\cM_r$ in degree $kn$ of type $((k \times n),(k \times n),(n\times k))$.
If $m,n > kr$ and if $\LatRect(k,m)$ holds, then there exists an irreducible representation of nontrivial equations for $\cM_r$ in degree $km$ of type $((k \times m),(m \times k),(m\times k))$.
Note that for $m=n$ this means that the equation is an $\SL(U) \times \SL(V) \times \SL(W)$-invariant polynomial.
\end{theorem}

The Kronecker coefficients for the type $((k \times m),(m \times k),(m\times k))$ are still not well understood.
The Kronecker coefficients for the slightly more general type $(\la,(m \times k),(m\times k))$ appear as and upper bound to the multiplicities of $\la$ in the coordinate ring of the orbit of the determinant, see \cite{BLMW:11}. Recent progress on rectangular Kronecker coefficients has been made in \cite{Man:11} and \cite{IP:17}.

The rest of this subsection is devoted to prove Theorem~\ref{thm:equationsfromrect}.
Like before, we start by establishing a construction principle.
\subsubsection*{Construction of highest weight vectors}

We use the same setup as in the paragraph ``Construction 
of highest weight vectors'' in Section~\ref{subsec:slicerankeqdesign}.
We are mostly interested in one specific permutation:
For $i,j\in\mathbb N$ let $\tau \in \aS_{ij}$ denote the \emph{transpose} permutation, i.e., $a + j (b-1) \mapsto b + j (a-1)$
for $1 \leq a \leq i$, $1 \leq b \leq j$.

For the proof of Theorem~\ref{thm:equationsfromrect}
we define $h := h_{k\times m} \otimes \tau(h_{m \times k}) \otimes \tau(h_{m \times k})$.
In both cases let $f$ denote the projection of $\varrho(h)\in(U\otimes V\otimes W)^{\otimes D}$ onto the $\aS_D$-invariant subspace.

\subsubsection*{Evaluation via products of determinants}
Let $T \in U^* \otimes V^* \otimes W^*$ and let $D:= km$.
Considering eq.~\eqref{slicerank:eq:evaluationcontraction}, we aim to understand the tensor contraction
\begin{equation}\label{eq:tensorcontractionfI}
\langle f^{(1)}, T^{\otimes D}\rangle.
\end{equation}
Since $T^{\otimes D}$ is symmetric under $\aS_D$, it follows
\[
\langle f^{(1)}, T^{\otimes D}\rangle = \langle \varrho(h^{(1)}), T^{\otimes D}\rangle.
\]
Our goal is to prove its vanishing for tensors from $\cM_r$, but its nonzeroness for at least one tensor.
In general, let $t = \sum_{j=q}^r a_j \otimes b_j \otimes c_j$. We expand
\[
T^{\otimes D} = \sum_{J:[D]\to[r]} a_{J(1)} \otimes b_{J(1)} \otimes c_{J(1)} \otimes \cdots \otimes a_{J(D)} \otimes b_{J(D)} \otimes c_{J(D)}.
\]
Note that
\begin{equation}\label{eq:prodofdets}
\langle h_{k \times m}, v_1 \otimes \cdots \otimes v_{mk}\rangle = \det(v_1,\ldots,v_k) \det(v_{k+1},\ldots,v_{2k}) \cdots \det(v_{(m-1)k+1},\ldots,v_{mk}),
\end{equation}
if each $v_i \in \bbC^k$.
But in our analysis the vectors $v_i$ will not always come from a $k$-dimensional vector space.
If there is $j>k$ such that each $v_i \in \bbC^j$,
then eq.~\eqref{eq:prodofdets} still holds if for vectors $w_1,\ldots,w_k \in \bbC^j$ we define $\det(w_1,\ldots,w_k)$
to be the determinant of the top $k \times k$ matrix of the $j\times k$ matrix given by $w_1,\ldots,w_k$.

Recalling that $h = h_{k \times m} \otimes \tau(h_{m \times k}) \otimes \tau(h_{m \times k})$, we see that
%\begin{equation}\label{eq:tripleprodofdets}
%\begin{minipage}{12cm}
%\vspace{-0.5cm}
\begin{eqnarray} 
\lefteqn{\langle\varrho(h),x_1 \otimes y_1 \otimes z_1 \otimes \cdots \otimes x_{mk} \otimes y_{mk} \otimes z_{mk}\rangle} && \notag\\
&=& \det(x_1,\ldots,x_k)\det(x_{k+1},\ldots,x_{2k})\cdots\det(x_{(m-1)k+1},\ldots,x_{mk})\label{eq:tripleprodofdets}\\
& & {} \cdot  \det(y_1,y_{k+1},\ldots,y_{(m-1)k+1})\det(y_2,y_{k+2},\ldots,y_{(m-1)k+2})\cdots\det(y_k,y_{2k},\ldots,y_{mk}) \notag \\
& & {}\cdot  \det(z_1,z_{k+1},\ldots,z_{(m-1)k+1})\det(z_2,z_{k+2},\ldots,z_{(m-1)k+2})\cdots\det(z_k,z_{2k},\ldots,z_{mk}). \notag
\end{eqnarray}
%\end{minipage}
%\end{equation}
The indices in \eqref{eq:tripleprodofdets} correspond to the rows and columns of the matrix
\[
\begin{pmatrix}
1 & k+1 & \hdots & (m-1)k+1 \\
2 & k+2 & \hdots & (m-1)k+2 \\
\vdots & \vdots & \ddots & \vdots \\
k & 2k & \hdots & mk
\end{pmatrix}.
\]
For notational convenience,
let $C_\beta:=\{k(\beta-1)+1,\ldots,k(\beta-1)+k\}$ denote the set of entries in column~$\beta$, $1 \leq \beta \leq m$,
and let 
$R_\alpha:=\{\alpha,k+\alpha,\ldots,k(m-1)+\alpha\}$ denote the set of entries in row~$\alpha$, $1 \leq \alpha \leq k$.

\subsubsection*{The equations vanish on $\cM_r$}
Let $T \in \cM_r$ and write
\[
T = \sum_{\ell=1}^{k-1} \Big(a_\ell \otimes \sum_{i\in[m]\atop j \in [n]} b_i^{(\ell)} \otimes c_j^{(\ell)}\Big)
+ a_k \otimes \sum_{i=1}^r b_i^{(k)} \otimes c_i^{(k)}.
\]
We now show that for this $T$, \eqref{eq:tensorcontractionfI} vanishes.
We expand $T$ in the straightforward (and not very efficient) way into a sum of $(k-1)mn+r$ many rank 1 tensors.
We expand $T^{\otimes km}$ into summands of the form
\[
x_1 \otimes y_1 \otimes z_1 \otimes \cdots \otimes x_{mk} \otimes y_{mk} \otimes z_{mk}
\]
and analyze \eqref{eq:tripleprodofdets} for each of the summands separately.
First, we observe that if \eqref{eq:tripleprodofdets} is nonzero, then in each of the sets $C_\beta$, $1 \leq \beta \leq m$,
there exists exactly one $\beta$ such that $x_\beta = a_k$.
Moreover, if \eqref{eq:tripleprodofdets} is nonzero, then the determinants for the $y$-variables (and independently also those for the $z$-variables) imply that in each of the sets $R_\alpha$
there are at most $r$ many $\alpha$ such that $x_\alpha = a_k$.
Therefore, since $m>kr$, the pigeonhole principle implies that \eqref{eq:tripleprodofdets} is zero.

The whole contraction \eqref{eq:tensorcontractionfI} vanishes, because \eqref{eq:tripleprodofdets} is zero for each summand in the expansion independently.

\subsubsection*{Nontriviality of the equations}
In this section we show that if $\LatRect(k,m)$, then there exists a tensor $t$ for which \eqref{eq:tensorcontractionfI} is nonzero,
which proves that our equations are not just the zero function.
Let $u_i \in U^*$, $1 \leq i \leq m$, be chosen generically.
Let $v_i \in V^*$, $1 \leq i \leq m$, form a basis of $V$ and let 
$w_i \in W^*$, $1 \leq i \leq n$, form a basis of $W$.
For the sake of simplicity, we assume that for $1 \leq i \leq m$ we have $v_i = w_i = e_i$ is the $i$-th standard basis vector.
We define
\[
T :=  \sum_{i=1}^m u_i \otimes v_i \otimes w_i
\]
We write \eqref{eq:tensorcontractionfI} as a sum of $m^{km}$ many summands of the form \eqref{eq:tripleprodofdets} by expanding $T^{\otimes km}$ into a sum of rank 1 tensors
\[
x_1 \otimes y_1 \otimes z_1 \otimes \cdots \otimes x_{mk} \otimes y_{mk} \otimes z_{mk}
\]
as we did in the last section.
By inspection of \eqref{eq:tripleprodofdets}
we observe that if there exist $i$ and $i'$ in $R_\alpha$ with $y_i=y_{i'}$, then \eqref{eq:tripleprodofdets} vanishes.
Moreover, if there exist $j$ and $j'$ in $C_\beta$ with $x_j=x_{j'}$, then \eqref{eq:tripleprodofdets} also vanishes.
Thus for each nonzero summand in \eqref{eq:tensorcontractionfI}, the matrix
\[
L := \begin{pmatrix}
u_1 & u_{k+1} & \hdots & u_{(m-1)k+1} \\
u_2 & u_{k+2} & \hdots & u_{(m-1)k+2} \\
\vdots & \vdots & \ddots & \vdots \\
u_k & u_{2k} & \hdots & u_{mk}
\end{pmatrix}
\]
forms a Latin Rectangle for $\mathcal U = \{u_1,\ldots,u_m\}$.
Moreover, each determinant of $y$-variables has value $\pm 1$,
and each determinant of $z$-variables also has value $\pm 1$,
and the signs of the $i$-th determinant of $y$-values and the $i$-th determinant of $z$-values coincide.
Since a product of an even number of $-1$s equals 1,
the value of each nonzero summand in \eqref{eq:tensorcontractionfI} equals $\coldet(L)$.
Thus the contraction \eqref{eq:tensorcontractionfI} equals $\sum_L \coldet(L)$,
where the sum if over all Latin Rectangles for $\mathcal U$.
This proves Theorem~\ref{thm:equationsfromrect}.

\subsection{Equations from multiplicities}
\label{sec:eq-mult}

The homogeneous part of the coordinate ring of $U \otimes V \otimes W$ in degree $d$
decomposes into two $\GL(U) \times \GL(V) \times \GL(W)$ representations:
$\bbC[U\otimes V\otimes W]_d = I(\cM_r)_d \oplus \bbC[\cM_r]_d$,
where $I(\cM_r)$ is the vanishing ideal of $\cM_r$, i.e., the subset of all polynomials on $U\otimes V\otimes W$ that vanish identically on $\cM_r$,
and $\bbC[\cM_r]_d := \bbC[U\otimes V \otimes W]/I(\cM_r)$ is the coordinate ring of $\cM_r$, whose elements can be interpreted as all restrictions of polynomials on
$U\otimes V\otimes W$ to $\cM_r$.

Determining the multiplicities of irreducible representations in $\bbC[U\otimes V\otimes W]$ can be done using classical character theory:
the multiplicities are the Kronecker coefficients.
To find equations, we prove a lower bound on multiplicities in $I(\cM_r)_d$ by proving an upper bound on multiplicities in $\bbC[\cM_r]_d$.
This is done by considering all regular functions on the orbit $(\GL_k \times \GL_s \times \GL_s) T_{k,m,r}$,
which we denote by $\bbC[(\GL_k \times \GL_s \times \GL_s) T_{k,m,r}]$:
These multiplicities are bounded from below by the multiplicities in $\bbC[\overline{(\GL_k \times \GL_s \times \GL_s) T_{k,m,r}}]_d$,
but they can be computed using branching rules in representation theory, \emph{without actually performing any calculations on tensors}.
We explain this method in this section.

Let $\{\la\}_k$ to denote the irreducible $\GL_k$-representation to the partition $\la$. We occasionally omit the subscript if the group is clear.
We write $\la \vdash_k$ to denote that $\la$ is a partition of some number into at most $k$ parts.

Let $G := \GL_k \times \GL_s \times \GL_s$.
The algebraic Peter-Weyl theorem can be used to describe the multiplicities in the coordinate ring of the orbit of $T_{k,m,r}$:
\[
\bbC[G T_{k,m,r}] = \bbC[G/H] = \bbC[G]^H = \bigoplus_{\la,\mu,\nu} \{\la,\mu,\nu\} \otimes \{\la,\mu,\nu\}^H.
\]
where $H \subseteq G$ is the stabilizer of $T_{k,n,r}$.
In particular
\[
\mult_{(\la,\mu,\nu)} \bbC[G T_{k,m,r}] = \dim \{\la,\mu,\nu\}^H.
\]
The rest of this section is devoted to determine $\dim \{\la,\mu,\nu\}^H$.
\[
\{\la,\mu,\nu\} = \{\la\}_k \otimes \{\mu\}_s \otimes \{\nu\}_s.
\]
Splitting $\{\mu\}_s$ into $\GL_r \times \GL_m^{\times k-1}$-irreducibles via the multi-Littlewood-Richardson rule yields:
\[
\{\mu\}_s = \bigoplus_{\substack{\mu^1 \vdash_r \\ \mu^2,\ldots,\mu^k \vdash_m}} c_{\mu^1,\ldots,\mu^k}^{\mu} \{\mu^1\}_r \otimes \{\mu^2\}_m \otimes \cdots \otimes \{\mu^k\}_m.
\]
Using the analogous equality for $\nu$ we obtain $\{\la,\mu,\nu\}=$
\[
\bigoplus_{\substack{\mu^1 \vdash_r \\ \mu^2,\ldots,\mu^k \vdash_m \\ \nu^1 \vdash_r \\ \nu^2,\ldots,\nu^k \vdash_m}}
c_{\mu^1,\ldots,\mu^k}^{\mu}
c_{\nu^1,\ldots,\nu^k}^{\nu}
\{\la\}_k \otimes
\{\mu^1\}_r \otimes \{\mu^2\}_m \otimes \cdots \otimes \{\mu^k\}_m
\otimes \{\nu^1\}_r \otimes \{\nu^2\}_m \otimes \cdots \otimes \{\nu^k\}_m
\]
For a partition $\xi$ we write $\xi \trianglelefteq \la$ when $\xi$ arises from $\la$ by removing boxes, at most one in each column.
Splitting $\{\la\}_k$ into irreducible $\GL_1 \times \GL_{k-1}$-representations via Pieri's rule yields
\[
\{\la\}_k = \bigoplus_{\substack{\xi \trianglelefteq \la\\\xi\vdash_{k-1}}} \{(|\la|-|\xi|)\}_1 \otimes \{\xi\}_{k-1},
\]
where $(|\la|-|\xi|)$ is the partitition to the one-row Young diagram with $|\la|-|\xi|$ many boxes.
In total, $\{\la,\mu,\nu\}=$
\[
\bigoplus_{\substack{\mu^1 \vdash_r \\ \mu^2,\ldots,\mu^k \vdash_m \\ \nu^1 \vdash_r \\ \nu^2,\ldots,\nu^k \vdash_m\\\xi \trianglelefteq \la\\\xi\vdash_{k-1}}}
c_{\mu^1,\ldots,\mu^k}^{\mu}
c_{\nu^1,\ldots,\nu^k}^{\nu}
\{(|\la|-|\xi|)\}_1 \otimes \{\xi\}_{k-1} \otimes
\{\mu^1\}_r \otimes \{\mu^2\}_m \otimes \cdots \otimes \{\mu^k\}_m
\otimes \{\nu^1\}_r \otimes \{\nu^2\}_m \otimes \cdots \otimes \{\nu^k\}_m
\]
Let $\{(a_1),\ldots,(a_k)\}$ denote the 1-dimensional irreducible $\GL_1^k$-representation to the 1-row partitions $(a_i)$.
Taking $\GL_r\times\GL_m^{k-1}$-invariants
in $\{\la,\mu,\nu\}$
and using that
\[
\dim(\{\mu^i\}\otimes\{\nu^i\}^*)^{\GL_m}=
\begin{cases}
1 & \text{ if } \mu^i=\nu^i \\
0 & \text{ otherwise }
\end{cases}
\]
yields $\{\la,\mu,\nu\}^{\GL_r\times\GL_m^{k-1}}=$
\[
\bigoplus_{\substack{\mu^1 \vdash_r \\ \mu^2,\ldots,\mu^k \vdash_m\\\xi \trianglelefteq \la\\\xi\vdash_{k-1}}}
c_{\mu^1,\ldots,\mu^k}^{\mu}
c_{\mu^1,\ldots,\mu^k}^{\nu}
\{(|\la|-|\xi|)\}_1 \otimes \{\xi\}_{k-1}
\otimes
\{(|\mu^1|),\ldots,(|\mu^k|)\}
\]
Taking $\GL_1 \times \GL_1^{k-1}$-invariants yields $\{\la,\mu,\nu\}^{\GL_1 \times \GL_1^{k-1} \times \GL_r\times\GL_m^{k-1}}=$
\[
\bigoplus_{\substack{\mu^1 \vdash_r |\la|-|\xi|\\ \mu^2,\ldots,\mu^k \vdash_m\\\xi \trianglelefteq \la\\\xi\vdash_{k-1}}}
c_{\mu^1,\ldots,\mu^k}^{\mu}
c_{\mu^1,\ldots,\mu^k}^{\nu}
\{\xi\}^{|\mu^2|,\ldots,|\mu^k|},
\]
where $\{\xi\}^{b_2,\ldots,b_k}$ is the weight space to $(b_1,\ldots,b_{k-1})$ in $\{\xi\}$.
To obtain $\dim\{\la,\mu,\nu\}^H$ we have to determine the dimension of the space of $\aS_{k-1}$-invariants in 
$\{\la,\mu,\nu\}^{\GL_1 \times \GL_1^{k-1} \times \GL_r\times\GL_m^{k-1}}$.
Observe that $\aS_{k-1}$ permutes the weight spaces, so we write $\{\la,\mu,\nu\}^H=$
\[
\bigoplus_{\substack{\mu^1 \vdash_r |\la|-|\xi|\\ \mu^2,\ldots,\mu^k \vdash_m\\\xi \trianglelefteq \la\\\xi\vdash_{k-1} \\ |\mu^2|\leq|\mu^3|\leq\cdots\leq|\mu^k|}}
c_{\mu^1,\ldots,\mu^k}^{\mu}
c_{\mu^1,\ldots,\mu^k}^{\nu}
\big(
\bigoplus_{\gamma \in \aS_{k-1}\cdot(|\mu^2|,|\mu^3|,\cdots,|\mu^k|)}
\{\xi\}^{\gamma_1,\ldots,\gamma_{k-1}},
\big)
\]
Fortunately, the dimension of $\aS_{k-1}$-invariants in the term in parentheses has been studied before in the context of geometric complexity theory and tensor rank \cite{BI:11}:
\[
\big(
\bigoplus_{\gamma \in \aS_{k-1}\cdot(|\mu^2|,|\mu^3|,\cdots,|\mu^k|)}
\{\xi\}^{\gamma_1,\ldots,\gamma_{k-1}}
\big)^{\aS_{k-1}} = \dim (\{\xi\}^{|\mu^2|,|\mu^3|,\cdots,|\mu^k|})^{\stab_{\aS_{k-1}} (|\mu^2|,|\mu^3|,\cdots,|\mu^k|)}.
\]
Let $J := (|\mu^2|,|\mu^3|,\cdots,|\mu^k|)$ and set $\aS_J := \aS_{|\mu^2|} \times \cdots \times \aS_{|\mu^k|}$.
Gay's theorem says that $\{\xi\}^{|\mu^2|,|\mu^3|,\cdots,|\mu^k|} = [\xi]^{\aS_J}$.
Let $S :=\stab_{\aS_{k-1}} (|\mu^2|,|\mu^3|,\cdots,|\mu^k|)$.
We want to determine $\dim [\xi]^{\aS_J \rtimes S}$. We calculate
\begin{equation}\label{eq:schurweyl}
\dim [\xi]^{\aS_J \rtimes S} = \dim \HWV_{\xi} \{\xi\} \otimes [\xi]^{\aS_J \rtimes S} \stackrel{\text{Schur-Weyl duality}}{=} \dim \HWV_\xi \big({\textstyle\bigotimes}^{|\xi|} V\big)^{\aS_J \rtimes S}.
\end{equation}
\begin{equation}\label{eq:prekappa}
\big({\textstyle\bigotimes}^{\xi}V\big)^{\aS_J \rtimes S} = (\Sym^{|\mu^2|}V \otimes \cdots \otimes \Sym^{|\mu^{k-1}|})^{S}
\end{equation}
Let $\kappa_i$ denote the number of times that $i$ occurs in the $J$.
Then \eqref{eq:prekappa} can be grouped as follows:
\begin{eqnarray*}
\text{\eqref{eq:prekappa}} &=& \big( {\textstyle\bigotimes}^{\kappa_1}\Sym^1 V \otimes \cdots \otimes {\textstyle\bigotimes}^{\kappa_{k-1}}\Sym^{k-1} V \big)^S \\
&=& \underbrace{\Sym^{\kappa_1}\Sym^1 V}_{= \bigoplus_{\delta^1}a_{\delta^1}(\kappa_1,1)\{\delta^1\}} \otimes \cdots \otimes \underbrace{\Sym^{\kappa_{k-1}}\Sym^{k-1} V}_{{= \bigoplus_{\delta^{k-1}}a_{\delta^{k-1}}(\kappa_{k-1},k-1)\{\delta^{k-1}\}}}
\end{eqnarray*}
Using the multi-Littlewood-Richardson rule we obtain
\[
\text{\eqref{eq:schurweyl}} = \sum_{\delta^1,\ldots,\delta^{k-1} \vdash_{\ell(\xi)} i\kappa_i} c^\xi_{\delta^1,\ldots,\delta^{k-1}} \prod_{i=1}^{|\xi|} a_{\delta^i}(\kappa_i,i).
\]
Altogether, $\dim\{\la,\mu,\nu\}^H=$
\[
\sum_{\substack{\mu^1 \vdash_r |\la|-|\xi|\\ \mu^2,\ldots,\mu^k \vdash_m\\\xi \trianglelefteq \la\\\xi\vdash_{k-1} \\ |\mu^2|\leq|\mu^3|\leq\cdots\leq|\mu^k|\leq |\xi|}}
\sum_{\delta^1,\ldots,\delta^{k-1} \vdash_{\ell(\xi)} \atop |\delta^i| = i\kappa_i}
c_{\mu^1,\ldots,\mu^k}^{\mu}
c_{\mu^1,\ldots,\mu^k}^{\nu}
c^\xi_{\delta^1,\ldots,\delta^{k-1}} \prod_{i=1}^{|\xi|} a_{\delta^i}(\kappa_i,i),
\]
where the $\kappa_i$ denotes the number of times $i$ occurs in $(|\mu^2|,\ldots,|\mu^k|)$.

Implementing this formula, we see that it indeed yields equations!
For example,
\[
\mult_{\la,\mu,\nu}\bbC[G T_{3,3,1}]_6 = 0 < 1 = k(\la,\mu,\nu)
\]
where $k(\la,\mu,\nu)$ denotes the Kronecker coefficient and $(\la,\mu,\nu)$ is one of the following cases:
\begin{itemize}
\item $((3,3),(2,2,2),(3,3))$
\item $((3,3),(3,3),(2,2,2))$
\item $((3,3),(2,2,2),(4,1,1))$
\item $((3,3),(4,1,1),(2,2,2))$
\end{itemize}
Numerous other partition triples can be readily generated.
Restricting the first partition to two rows and the second and third to three rows, we checked with the software \textsc{Macaulay2} combined
with methods from \cite{BI:13}
that
this method only misses one triple: $((3,3),(3,2,1),(3,2,1))$, where the multiplicities on the left hand side and the right hand side are both 2.

% Observe that $\aS_{k-1}$ permutes the weight spaces, so an invariant can only be obtained if $|\mu^2|=\cdots=|\mu^k|$.
% In this case the dimension equals the plethysm coefficient $a_{\xi}(k-1,|\mu^2|)$,
% i.e., the multiplicity of $\{\xi\}_{k-1}$ in $\Sym^{k-1}(\Sym^{|\mu^2|}(\bbC^{k-1}))$.
% We interpret $a_{\xi}(0,i) = 1$ iff $\xi$ is the empty partition, and $a_{\xi}(0,i) = 0$ otherwise.
% Altogether, $\dim\{\la,\mu,\nu\}^H=$
% \[
% \sum_{\substack{\xi \trianglelefteq \la\\
% k-1 \text{ divides } |\xi| \\
% \xi\vdash_{k-1}\\
% \mu^1 \vdash_r |\la| - |\xi|
% \\ \mu^2,\ldots,\mu^k \vdash_m \frac{|\xi|}{k-1}}}
% c_{\mu^1,\ldots,\mu^k}^{\mu}
% c_{\mu^1,\ldots,\mu^k}^{\nu}
% a_{\xi}(k-1,\tfrac{|\xi|}{k-1}).
% \]

\section{Complexity-theoretic properties}
Here we show the $\NP$-hardness of the slice rank and the minrank problems.

\subsection{Hardness of Slice Rank} \label{sec:hardslices}

In this section, we show that the problem of testing if a given 3-tensor has slice rank at most $r$ is $\NP$-hard. We do this by showing that a variant of hypergraph vertex cover testing is $\NP$-hard. Tao and Sawin \cite{ts16} showed the equivalence of the slice rank problem to this variant of hypergraph vertex cover testing.

We fix a field $\mathbb{F}$. Given a $3$-uniform, $3$-partite hypergraph $H$ with $3$ partitions $U,V $ and $W$ with $|U| = n_1$, $|V| = n_2$, and $ |W| = n_3$, $n_i \in \mathbb{N}$, $i \in [3]$, with edge set being $E \subseteq U \times V \times W$, we can define a $3$-tensor $T_H (\mathbf{x_1}, \mathbf{x_2}, \mathbf{x_3})$ corresponding to $H$, where $\mathbf{x_i}$ is a tuple of $[n_i]$ variables in the following way.
\[ T_H (\mathbf{x_1}, \mathbf{x_2}, \mathbf{x_3}) = \sum_{(u_{i_1},v_{i_2}, w_{i_3}) \in E}x_{1,i_1} \cdot x_{1,i_2} \cdot  x_{3,i_3} \]

We label the nodes in $U, V$ and $W$ from the set of integers. For two hyperedges $e_1 := (u_{a_1}, v_{b_1}, w_{c_1})$ and $e_2 :=(u_{a_2}, v_{b_2}, w_{c_2})$, we say that  $  e_1 \leq e_2$ iff $(a_1 \leq a_2) \land (b_1 \leq b_2) \land (c_1 \leq c_2)$. If neither $e_1 \leq e_2$ nor $e_2 \leq e_1$ holds, we say that $e_1$ and $e_2$ are incomparable. In $E$, if every pair of hyperedges is incomparable to each other, we say that $E$ is an antichain. 

Tao and Sawin (see \cite[Proposition 4]{ts16}) showed the following.

\begin{lemma} If the hyperedge set $E$ is an antichain, then the slice rank of $T_H$ is the same as the size of the minimum vertex cover of the hypergraph $H$.
\end{lemma}

Thus, in order to show that computing slice rank of 3-tensors is NP-hard, we show that the hypergraph  minimum vertex cover problem for a 3-partite, 3-uniform graph, where the edge set is an antichain, is NP-hard.

Our reduction is inspired by \cite{gottlob10} where they show the NP-hardness of the hypergraph vertex cover problem for 3-uniform 3-partite graphs. Their reduction involved reducing 3-SAT to this problem. Here we need to show the hardness under the extra condition that the hyperedge set of the graph is an antichain.
This makes the reduction far more involved, and we also change the hard problem that we reduce to our problem.\\
The NP-hard problem that we use for our reduction is a bounded occurrence mixed SAT problem (bom-SAT), where we have 3-clauses and 2-clauses, such that every variable appears exactly thrice, once in a 3-clause, while the other two occurrences are in 2-clauses (note that the number of variables, $n = 3t$, for some $t$, where $t$ is the number of 3-clauses).

\begin{remark}
It is easy to see that the above mentioned bom-SAT is $\NP$-hard. For this, start with any 3-SAT instance. Now assume that a variable $Z$ appears $m$ times. Introduce $m$ copies $Z_1,...,Z_m$ of $X$. Replace every occurrence of $Z$ by one $Z_i$. We do this for all the variables. Now every variables appears only once. However, we have to ensure consistency, that is, $Z_1,...Z_m$ should have the same value. So we add the 2-clauses: $(Z_1 \lor \neg Z_2) \land (Z_2 \lor \neg Z_3) \land \cdots \land (Z_m \lor \neg Z_1)$.
These 2-clauses can only be satisfied if we set all the $Z_i$'s to $0$ or all the  $Z_i$'s to $1$. The resulting formula is a bom-SAT instance as described above.
\end{remark}

In the reduction, given a bom-SAT formula $\phi$ in $n$ variables $X_1, \ldots X_n$ with $t$ 3-clauses and $m$ 2-clauses, the construction of a 3-uniform 3-partite hypergraph $G^{\phi}$ with 3 vertex partitions $U, V$ and $W$ proceeds as follows.
First of all we sort all the clauses such that all the 3-clauses precede all the 2-clauses. Next we rename all the variables such that the variables in the $r$-th 3-clause ($r \in t$) are $Y_{3(r-1) +1} , Y_{3(r-1) +2}$ and $Y_{3(r-1)+3}$ corresponding to the first, second and the third position of the clause respectively. We also say that  $Y_{3(r-1) +1} , Y_{3(r-1) +2}$ and $Y_{3(r-1)+3}$ belong to the same triple of variables.\\
Now, we have a gadget $G^{\phi}_k$ corresponding to each variable $Y_k$, $k \in[n]$. $G^{\phi}_k$ consists of nodes $(i,j)^k$ and  $\overline{(i,j)}^k, i, j \in \{1,2,3\}$. Here $(i,j)^k$ refers to the node corresponding to the $i$-th occurrence of the variable $Y_k$, and it occurs at the $j$-th position in the clause in which it appears. $\overline{(i,j)}^k$ refers to the negation of $Y_k$ in its $i$-th occurrence at the $j$-th position in the clause. We will drop the superscript $k$, when it is clear from the context. Clearly, there are $18$ such \textit{literal-nodes} in a gadget $G^{\phi}_k$, which are ordered along a circle (see the outer circle in Figure \ref{fig:partition}). Since $Y_k$ appears exactly thrice in $\phi$, exactly 3 out of these 18 nodes will correspond to some occurrence of $Y_k$ in $\phi$. $G^{\phi}_k$ also consists of $18$ other nodes, which we call \textit{free-nodes} (as they do not correspond to any literal), that are useful in the construction (see the inner circle in Figure \ref{fig:partition}). 
We have hyperedges connecting two literal-nodes and a free-node. There are total $18$ hyperedges in $G^{\phi}_k$ each consisting of three vertices that form a triangle in Figure~\ref{fig:partition}.
Note that every literal-node appears in exactly $2$ hyperedges, while a free-node appears in exactly one of them. We partition the set of nodes in 3 parts, as illustrated in the figure. Among the literal-nodes, the nodes corresponding to the first-occurrences ($j=1$) go to the set $U$, the ones corresponding to the second-occurrences ($j=2$) go to the set $V$, while the ones corresponding to third occurrences ($j=3$) go to the set $W$. We distribute the free-nodes equally among the three sets, while maintaining the property of being 3-partite (see Figure \ref{fig:partition}).

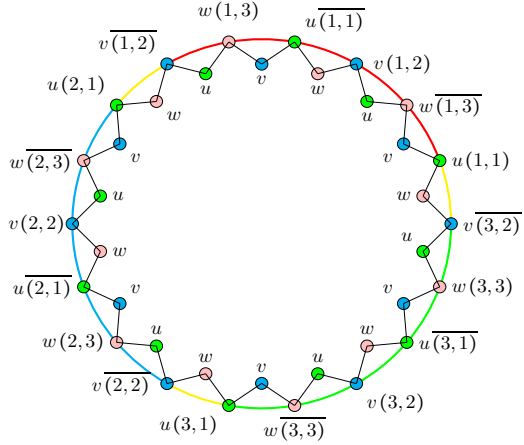
\begin{figure} 
\centering
\resizebox{7cm}{6cm}{
\begin{tikzpicture}
%	\draw[-] (0,0) circle (4.0cm);
%%	\draw[-] (0,0) circle (5.0cm);
	\draw [very thick] [yellow] (0:3.7) arc [radius=3.7 cm, start angle=0, end angle=20];
	  \draw[very thick] [red] (20:3.7) arc [radius=3.7 cm, start angle=20, end angle=120];
	  \draw[very thick] [yellow] (120:3.7) arc [radius=3.7 cm, start angle=120, end angle=140];
  \draw [very thick] [cyan] (140:3.7) arc [radius=3.7 cm, start angle=140, end angle=240];
  \draw [very thick] [yellow] (240:3.7) arc [radius=3.7 cm, start angle=240, end angle=260];
  \draw [very thick] [green] (260:3.7) arc [radius=3.7 cm, start angle=260, end angle=360];
	
%	\draw[-] (0,0) circle (5.0cm);

  \foreach \phi in {1,...,18}{
    \coordinate (A\phi) at (360/18*\phi:3.7cm);
    \coordinate (B\phi) at (360/18*\phi +30:3.2cm);
%    \draw[myshorten,myedge] (A\phi) arc ( 360/18*\phi: 360/18*(\phi+1): 1cm);
%    \draw[myshorten,myedge] (B\phi) arc ( 360/18*\phi: 360/18*(\phi+1): 2.5cm);
  }
  \foreach \phi in {1,4,7,10,13,16}{
      \draw[fill = green] (A\phi)circle(1. 2 mm);
		}
 \foreach \phi in {2,5,8,11,14,17}{
      \draw[fill = pink] (A\phi)circle(1. 2 mm);
		}
 \foreach \phi in {3,6,9,12,15,18}{
      \draw[fill = cyan] (A\phi)circle(1. 2 mm);
		}

  \foreach \phi in {1,4,7,10,13,16}{
      \draw[fill = green] (B\phi)circle(1. 2 mm);
		}
 \foreach \phi in {2,5,8,11,14,17}{
      \draw[fill = pink] (B\phi)circle(1. 2 mm);
		}
 \foreach \phi in {3,6,9,12,15,18}{
      \draw[fill = cyan] (B\phi)circle(1. 2 mm);	
   }

  \foreach \phi [count = \s from 2] in {1,...,16}
  {
  \draw[-] (B\phi) -- (A\s);
    } 
    \foreach \phi [count = \s from 3] in {1,...,16}
  {
  \draw[-] (B\phi) -- (A\s);
    }   
    \draw[-] (B17) -- (A18);
    \draw[-] (B17) -- (A1);
    \draw[-] (B18) -- (A1);
    \draw[-] (B18) -- (A2);

	\node at (A1)[right= 1mm]{$u \hspace{0.4mm} (1,1)$};
	\node at (A2)[right= 1mm]{$w\hspace{0.4mm}\overline{(1,3)}$};
	\node at (A3)[right= 2mm]{$v\hspace{0.4mm}(1,2)$};
	\node at (A4)[above right= 1mm]{$u\hspace{0.4mm}\overline{(1,1)}$};
	\node at (A5)[above = 3mm]{$w\hspace{0.4mm} (1,3)$};
	\node at (A6)[above left= 1mm]{$v\hspace{0.4mm} \overline{(1,2)}$};
	\node at (A7)[above left= 1mm]{$u\hspace{0.4mm}(2,1)$};
	\node at (A8)[left= 1mm]{$w\hspace{0.4mm}\overline{(2,3)}$};
	\node at (A9)[left= 0mm]{$v\hspace{0.2mm}(2,2)$};
	\node at (A10)[left= 1mm]{$u\hspace{0.4mm}\overline{(2,1)}$};
	\node at (A11)[left= 1mm]{$w\hspace{0.4mm}(2,3)$};
	\node at (A12)[left= 2mm]{$v\hspace{0.4mm}\overline{(2,2)}$};
	\node at (A13)[below left= 1mm]{$u\hspace{0.4mm}(3,1)$};
	\node at (A14)[below= 1mm]{$w\hspace{0.4mm}\overline{(3,3)}$};
	\node at (A15)[below right= 1mm]{$v\hspace{0.4mm}(3,2)$};
	\node at (A16)[right= 1mm]{$u\hspace{0.4mm}\overline{(3,1)}$};
	\node at (A17)[right= 1mm]{$w\hspace{0.4mm}(3,3)$};
	\node at (A18)[right= 1mm]{$v\hspace{0.4mm}\overline{(3,2)}$};

	\node at (B1)[below= 1mm]{$u$};
	\node at (B2)[below= 1mm]{$w$};
	\node at (B3)[below= 1mm]{$v$};
	\node at (B4)[below= 1mm]{$u$};
	\node at (B5)[below right= 1mm]{$w$};
	\node at (B6)[below right= 1mm]{$v$};
	\node at (B7)[right= 1mm]{$u$};
	\node at (B8)[right= 1mm]{$w$};
	\node at (B9)[above right= 1mm]{$v$};
	\node at (B10)[above = 1mm]{$u$};
	\node at (B11)[above= 1mm]{$w$};
	\node at (B12)[above= 1mm]{$v$};
	\node at (B13)[above= 1mm]{$u$};
	\node at (B14)[above= 1mm]{$w$};
	\node at (B15)[above left = 1mm]{$v$};
	\node at (B16)[above left= 1mm]{$u$};
	\node at (B17)[left= 1mm]{$w$};
	\node at (B18)[below left= 1mm]{$v$};
\end{tikzpicture}}
\caption{ A variable gadget $G^{\phi}_{k}$ corresponding to the variable $Y_k$ in $\phi$. Nodes sharing the red, cyan and green arcs correspond to the first, second, and third occurrence of $Y_k$ in a clause respectively. Exactly 3 out of 18 literal-nodes are used in clause hyperedges. Nodes with an overline indicate that the negation of $Y_k$ appeared in the corresponding clause. Nodes in the inner circle correspond to the free-nodes. \label{fig:partition}}
\label{fig:gadget}
\end{figure}

Additionally, we have clause hyperedges, which for a 3-clause, connect the nodes corresponding to the three literals present in it. For every 2-clause, we first introduce another free-node to the graph, added to set $W$ (as there are no literals at the third position in a 2-clause). Now, there is an hyperedge for every 2-clause as well, connecting the two nodes corresponding to its literals and a free-node. We refer to the hyperedges in a variable gadget either as \textit{variable hyperedges} or \textit{local hyperedges}. We refer to the hyperedges corresponding to the clauses as \textit{clause hyperedges} or \textit{global hyperedges}.
 We illustrate the set up with an example. See Figure \ref{fig:fullgadget}.\\
 
 The following two lemmas finishes the reduction.

\begin{lemma}
The size of the minimum vertex cover of the hypergraph $G^{\phi}$ is at most $9n$ if and only the \emph{bom-SAT} instance $\phi$ is satisfiable.
\end{lemma}
The proof of this lemma follows very closely the proof of hardness of hypergraph minimum vertex cover problem (see \cite[Lemma 5.3]{gottlob10}), which was itself inspired by the proof of NP-hardness of 3-dimensional matching given in Garey and Johnson \cite{garey79}. We give a sketch here.

\begin{proof}
Let $\phi$ be satisfiable with $\nu$ being a satisfying assignment on the variables $Y_1, \ldots, Y_n$.
Now, we construct the vertex cover set $S$ for $G^{\phi}$ of size $9n$ as follows. If $\nu(Y_k) = 0$, we add all the $9$ overlined nodes from $G^{\phi}_{k}$ to $S$, otherwise we add the other $9$ nodes to $S$. Note that $S$ covers all the local hyperedges. Since $\nu$ is a satisfying assignment, all the clause hyperedges are also covered by $S$ as well.

Conversely, assume there is a minimum vertex cover $S$ of $G^{\phi}$ of size at most $9n$.
Now, since all the free-nodes appear in only one hyperedge each, we can assume that $S$ does not contain any free-node, since we can always replace them by a literal-node of the same hyperedge. Now, for $i \in \{1, \ldots, n \}$ if $S_i$ is the subset of $S$ such that $S_i$ only contains the vertices corresponding to the variable gadget $G^{\phi}_i$, it can be easily seen that $|S_i| \geq 9$ for all the variable hyperedges to be covered. This implies that $|S_i| = 9$ since we assumed that $|S| = |\cup_{i=1}^{n}S_i| \leq 9n$.
Thus $S_i$ forms a vertex cover corresponding to the local gadget $G^{\phi}_i$ and hence covers the hyperedges in $G^{\phi}_i$. However, there are only two vertex covers of $G^{\phi}_i$ of size $9$, namely the one set containing all the overlined nodes, i.e., they correspond to $\neg Y_k$, and the other set where none of the nodes are overlined, i.e., they correspond to $Y_k$. In the first case, we assign the value $0$ to $Y_k$, and we assign $1$ in the second case. Thus we construct the assignment $\nu$ for $Y_1, \ldots, Y_n$.  Now, since $S$ is a vertex cover and hence span all the hyperedges including the clause hyperedges, $\nu$ satisfies all the clauses of $\phi$.
\end{proof}

The following lemma ensures that the edge set $E$ of the above constructed graph $G^{\phi}$ is indeed an antichain under some labelling.

\begin{lemma}\label{lem:slicelabel}
For every formula $\phi$, there exists a way of labelling of the nodes in hypergraph $G^{\phi}$ such that the hyperedge set of $G^{\phi}$ is an antichain.
\end{lemma}

\begin{proof}
We first give the labelling used.
We have literal-nodes and free-nodes.
The literal-nodes either correspond to the first occurrence, the second occurrence or the third occurrence of a variable. In every gadget, we have $6$ nodes corresponding to each occurrence, $2$ from each partition $U, V$ and $W$. The free-nodes although do not correspond to any occurrences, we say that they correspond to first occurrence if the two literal-nodes that they connect both correspond to the first occurrence. In every gadget, there are $5$ such nodes, $2$ each belonging to $U$ and $V$, while one belonging to $W$.  If a free-node does not correspond to the first occurrence, we say that it corresponds to the second or third occurrence (we do not make distinction within them as it is not needed). 

\begin{figure} 
\centering
\resizebox{7cm}{6cm}{
\begin{tikzpicture}
%	\draw[-] (0,0) circle (4.0cm);
%%	\draw[-] (0,0) circle (5.0cm);
	\draw [very thick][green] (0:3.7) arc [radius=3.7 cm, start angle=0, end angle=20];
	  \draw [very thick][red] (20:3.7) arc [radius=3.7 cm, start angle=20, end angle=120];
  \draw [very thick][green] (120:3.7) arc [radius=3.7 cm, start angle=120, end angle=360];
	
%	\draw[-] (0,0) circle (5.0cm);

  \foreach \phi in {1,...,18}{
    \coordinate (A\phi) at (360/18*\phi:3.7cm);
    \coordinate (B\phi) at (360/18*\phi +30:3.2cm);
%    \draw[myshorten,myedge] (A\phi) arc ( 360/18*\phi: 360/18*(\phi+1): 1cm);
%    \draw[myshorten,myedge] (B\phi) arc ( 360/18*\phi: 360/18*(\phi+1): 2.5cm);
  }
  \foreach \phi in {1,4,7,10,13,16}{
      \draw[fill = green] (A\phi)circle(1. 2 mm);
		}
 \foreach \phi in {2,5,8,11,14,17}{
      \draw[fill = pink] (A\phi)circle(1. 2 mm);
		}
 \foreach \phi in {3,6,9,12,15,18}{
      \draw[fill = cyan] (A\phi)circle(1. 2 mm);
		}

  \foreach \phi in {1,4,7,10,13,16}{
      \draw[fill = green] (B\phi)circle(1. 2 mm);
		}
 \foreach \phi in {2,5,8,11,14,17}{
      \draw[fill = pink] (B\phi)circle(1. 2 mm);
		}
 \foreach \phi in {3,6,9,12,15,18}{
      \draw[fill = cyan] (B\phi)circle(1. 2 mm);	
   }

  \foreach \phi [count = \s from 2] in {1,...,16}
  {
  \draw[-] (B\phi) -- (A\s);
    } 
    \foreach \phi [count = \s from 3] in {1,...,16}
  {
  \draw[-] (B\phi) -- (A\s);
    }   
    \draw[-] (B17) -- (A18);
    \draw[-] (B17) -- (A1);
    \draw[-] (B18) -- (A1);
    \draw[-] (B18) -- (A2);

	\node at (A1)[right= 1mm]{$u_{1} \hspace{0.4mm} (1,1)$};
	\node at (A2)[right= 1mm]{$w_{42}\hspace{0.4mm}\overline{(1,3)}$};
	\node at (A3)[right= 2mm]{$v_{1}\hspace{0.4mm}(1,2)$};
	\node at (A4)[above right= 1mm]{$u_{2} \hspace{0.4mm}\overline{(1,1)}$};
	\node at (A5)[above = 3mm]{$w_{41}\hspace{0.4mm} (1,3)$};
	\node at (A6)[above left= 1mm]{$v_{2}\hspace{0.4mm} \overline{(1,2)}$};
	\node at (A7)[above left= 1mm]{$u_{13}\hspace{0.4mm}(2,1)$};
	\node at (A11)[left= 1mm]{$w_{1}\hspace{0.4mm}(2,3)$};
	\node at (A9)[left= 0mm]{$v_{13}\hspace{0.2mm}(2,2)$};
	\node at (A10)[left= 1mm]{$u_{14}\hspace{0.4mm}\overline{(2,1)}$};
	\node at (A8)[left= 1mm]{$w_{2}\hspace{0.4mm}\overline{(2,3)}$};
	\node at (A12)[left= 2mm]{$v_{14}\hspace{0.4mm}\overline{(2,2)}$};
	\node at (A13)[below left= 1mm]{$u_{15}\hspace{0.4mm}(3,1)$};
	\node at (A17)[right= 1mm]{$w_{3}\hspace{0.4mm}(3,3)$};
	\node at (A15)[below right= 1mm]{$v_{15}\hspace{0.4mm}(3,2)$};
	\node at (A16)[right= 1mm]{$u_{16}\hspace{0.4mm}\overline{(3,1)}$};
	\node at (A14)[below= 1mm]{$w_{4}\hspace{0.4mm}\overline{(3,3)}$};
	\node at (A18)[right= 1mm]{$v_{16}\hspace{0.4mm}\overline{(3,2)}$};

	\node at (B1)[below= 1mm]{$u_{\mhyphen1}$};
	\node at (B2)[below= 1mm]{$w_{40}$};
	\node at (B3)[below= 1mm]{$v_{\mhyphen2}$};
	\node at (B4)[below= 1mm]{$u_{\mhyphen2}$};
	\node at (B5)[below right= 1mm]{$w_{\mhyphen1}$};
	\node at (B6)[below right= 1mm]{$v_{\mhyphen13}$};
	\node at (B7)[right= 1mm]{$u_{\mhyphen13}$};
	\node at (B8)[right= 1mm]{$w_{\mhyphen2}$};
	\node at (B9)[above right= 1mm]{$v_{\mhyphen14}$};
	\node at (B10)[above = 1mm]{$u_{\mhyphen14}$};
	\node at (B11)[above= 1mm]{$w_{\mhyphen3}$};
	\node at (B12)[above= 1mm]{$v_{\mhyphen15}$};
	\node at (B13)[above= 1mm]{$u_{\mhyphen15}$};
	\node at (B14)[above= 1mm]{$w_{\mhyphen4}$};
	\node at (B15)[above left = 1mm]{$v_{\mhyphen16}$};
	\node at (B16)[above left= 1mm]{$u_{\mhyphen16}$};
	\node at (B17)[left= 1mm]{$w_{\mhyphen5}$};
	\node at (B18)[below left= 1mm]{$v_{\mhyphen1}$};

\end{tikzpicture}}
\caption{The labelling of variable gadgets $G^{\phi}_1$ for $n=6$. The hyperedges with a red arc correspond to the first occurrence of variables. Notice the difference in labelling of $W$ nodes. Literal-nodes are all labelled positive. Free-nodes are all labelled negative except the $W$ node connecting the two first occurrence literal-nodes. \label{fig:labels}}
\end{figure}
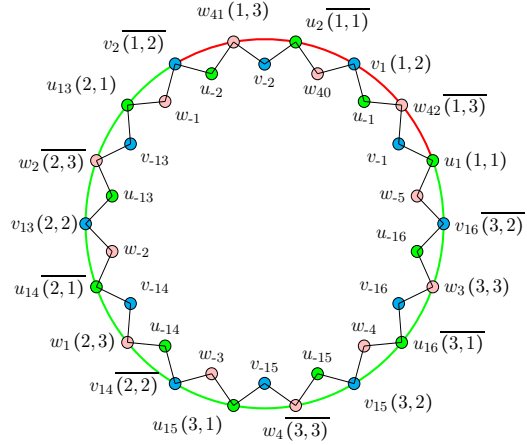

We first give the labelling corresponding to the nodes corresponding to the second and the third occurrences of variables:
\begin{itemize}
\item The position 1 literal-nodes $(i,1)^k$ and $\overline{(i,1)}^k$ in $G^{\phi}_k$ are labelled\\ $u_{2n+ 2(i-2) + 4(k-1)+1}$ and $u_{2n + 2(i-2) + 4(k-1) +2}$, respectively, $\forall k$, for $i =2,3$.

\item Similarly, the position 2 literal-nodes $(i,2)^k$ and $\overline{(i,2)}^k$ are labelled\\ $v_{2n + 2(i-2) + 4(k-1)+1}$ and $v_{2n + 2(i-2) + 2(k-1) +2}$, respectively, $\forall k$, for $i =2,3$.

\item Likewise, the position 3 literal-nodes $(i,3)^k$ and $\overline{(i,3)}^k$ are labelled\\ $w_{2n + 2(i-2) + 4(k-1)+1}$ and $w_{2n + 2(i-2) + 4(k-1) +2}$ respectively, $\forall k$, for $i =2,3$.

\item The 4 free $U$ nodes in $G^{\phi}_k$ corresponding to the second or third occurrence are labelled $u_{-2n -4(k-1)-\ell}$, $\ell \in [4]$ (see Figure \ref{fig:labels} to see which ones exactly).

\item Similarly, the 4 such free $V$ nodes in $G^{\phi}_k$ are labelled $v_{-2n-4(k-1)-\ell}$, $\ell \in [4]$.

\item Finally, the 5 such free $W$ nodes in $G^{\phi}_k$ are labelled $w_{-5(k-1)-\ell}$, $\ell \in [5]$.

\item All the 2-clauses also correspond to the second and third occurrence of variables. Each such 2-clause will have a corresponding hyperedge. Here we have a freedom to choose the position for the free node. We invariably choose it to be at the third position. Thus the first two nodes of the hyperedges will take the relevant literals as per the clause, while the $W$ nodes will be free ones. For the $s-$th 2-clause (under an arbitrary order), $s \in [m]$  label the $W$ nodes as $w_{-5n-s}$. 

\item We take all the hyperedges that include all the above labelled free $W$ nodes. This will include all the 2-clause hyperedges along with 5 hyperedges per variable gadget. Now the tuple of $U$ and $V$ coordinates $(u_a, v_b)$ of these hyperedges will have a partial order among themselves. We shuffle their $W$ coordinates so that the order of the $W$ coordinates becomes the reverse of the order of the tuple $(u_a,v_b)$. We can do this without disturbing other hyperedges because these $W$ nodes are all free and are used in only one hyperedge each.

\end{itemize}

Now it remains to label the literal-nodes corresponding to the first occurrences and the free nodes pertaining to them. They are labelled differently so as to ensure that the antichain property indeed holds when the hyperedges connecting these would be compared with the 3-clause hyperedges. One key difference is that the labels of $W$ nodes for $G^{\phi}_k$ in this case also depend on whether $k \equiv 1, 2$ or $0 \mod 3$.

\begin{itemize}

\item The position 1 literal-nodes $(1,1)^k$ and $\overline{(1,1)}^k$ in $G^{\phi}_k$ are labelled $u_{2(k-1)+1}$ and  $u_{2(k-1)+2}$, respectively, $\forall k$.

\item The position 2 literal-nodes $(1,2)^k$ and $\overline{(1,2)}^k$ are labelled $v_{2(k-1)+1}$ and $v_{2(k-1)+2}$, respectively, $ \forall k$.

\item The position 3 literal-nodes $(1,3)^k$ and $\overline{(1,3)}^k$ get the labels 
%$w_{7n -3(k-1)}$ and $w_{7n -3(k-1)-1}$ respectively, for $k\equiv 1,2 \mod 3$, whereas
% $w_{7n -3(k-1)+1}$ and $w_{7n -3(k-1)}$ respectively for $k \equiv 0 \mod 3$.
% Thus within the same triple of variables, say $q-$th triple $Y_{3(q-1) +1},Y_{3(q-1) +2},Y_{3(q-1) +3} $, these labels become 
$w_{7n -9(q-1)}$ and $w_{7n -9(q-1)-1}$, respectively, for $k = 3(q-1) +1$, whereas
 $w_{7n - 9(q-1) -3}$ and $w_{7n -9(q-1)-4}$, respectively, for $k = 3(q-1) +2$, and
 $w_{7n -9(q-1)-5}$ and $w_{7n -9(q-1)-6}$, respectively, for $k = 3(q-1) +3$

\item The 2 free $U$ nodes corresponding to the first occurrence of the variable get the labels $u_{-2(k-1) -1}$ and $u_{- 2(k-1)-2}$, respectively. Similarly such free $V$ nodes get the labels $v_{-2(k-1)-1}$ and $v_{-2(k-1)-2}$ respectively, whereas the such free $W$ nodes (1 per gadget) get  the labels 
%$w_{7n - 3(k-1)-2}$ for $k\equiv 1,0 \mod 3$ and $w_{7n - 3(k-1)-4}$ for $k\equiv 2 \mod 3$. Thus, within the same triple of variables, the free $W$ nodes have the labels  
$w_{7n - 9(q-1)-2}$ for $k = 3(q-1)+1$ and $w_{7n - 9(q-1)-7}$ for $k=3(q-1)+2$, and  $w_{7n - 9(q-1)-8}$ for $k = 3(q-1)+3$.

\end{itemize}

Figure \ref{fig:fullgadget} illustrates the labelling for $k=1, 2, 3$ when $n=6$.

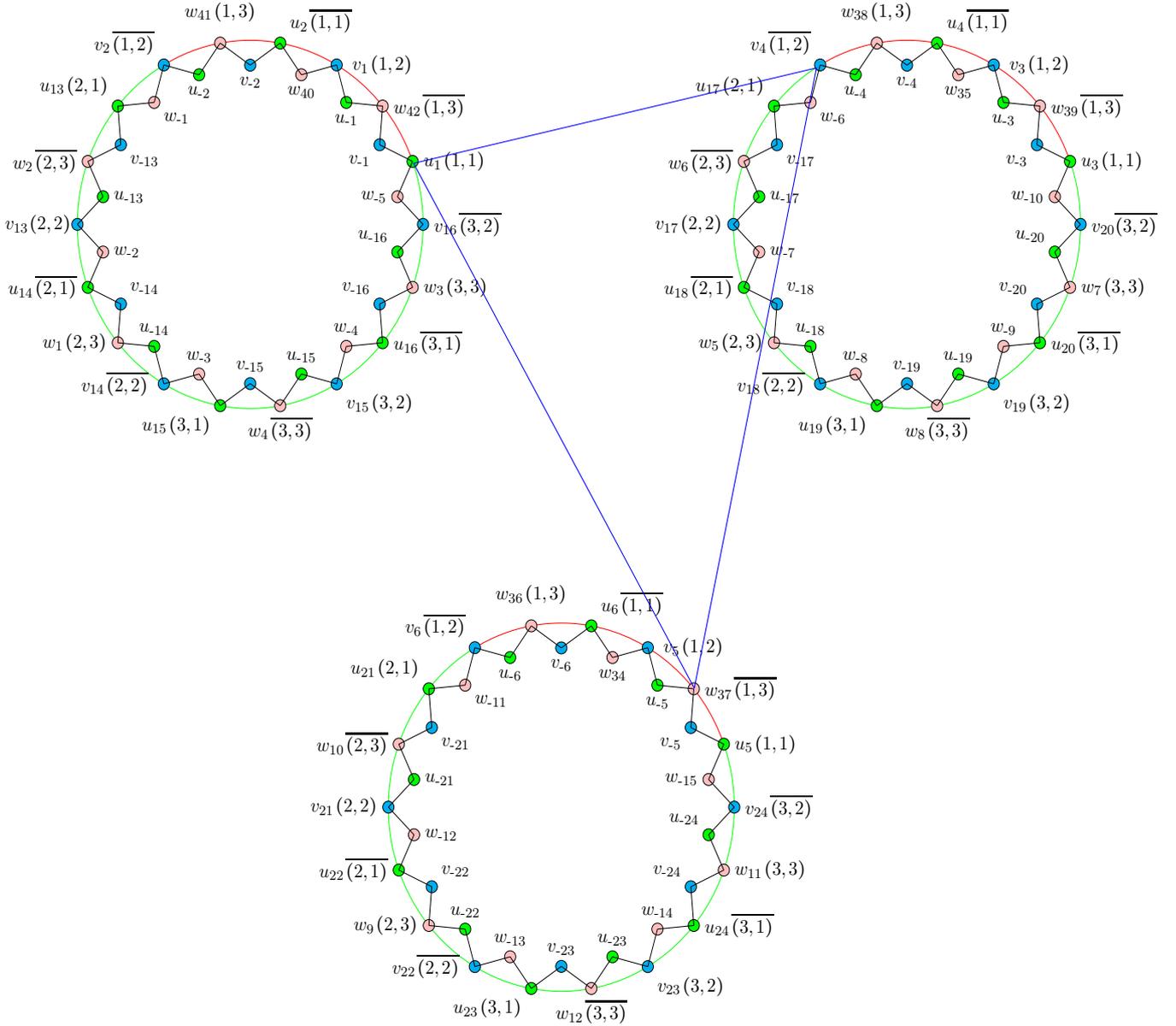
\begin{figure} 
\begin{minipage}{0.5\textwidth}
\hspace{-1 cm}
\resizebox{8cm}{7cm}{
\begin{tikzpicture}[remember picture]
%	\draw[-] (0,0) circle (4.0cm);
%%	\draw[-] (0,0) circle (5.0cm);
	\draw [green] (0:3.7) arc [radius=3.7 cm, start angle=0, end angle=20];
	  \draw [red] (20:3.7) arc [radius=3.7 cm, start angle=20, end angle=120];
  \draw [green] (120:3.7) arc [radius=3.7 cm, start angle=120, end angle=360];
	
%	\draw[-] (0,0) circle (5.0cm);

  \foreach \phi in {1,...,18}{
    \coordinate (A\phi) at (360/18*\phi:3.7cm);
    \coordinate (B\phi) at (360/18*\phi +30:3.2cm);
%    \draw[myshorten,myedge] (A\phi) arc ( 360/18*\phi: 360/18*(\phi+1): 1cm);
%    \draw[myshorten,myedge] (B\phi) arc ( 360/18*\phi: 360/18*(\phi+1): 2.5cm);
  }
  \foreach \phi in {1,4,7,10,13,16}{
      \draw[fill = green] (A\phi)circle(1. 2 mm);
		}
 \foreach \phi in {2,5,8,11,14,17}{
      \draw[fill = pink] (A\phi)circle(1. 2 mm);
		}
 \foreach \phi in {3,6,9,12,15,18}{
      \draw[fill = cyan] (A\phi)circle(1. 2 mm);
		}

  \foreach \phi in {1,4,7,10,13,16}{
      \draw[fill = green] (B\phi)circle(1. 2 mm);
		}
 \foreach \phi in {2,5,8,11,14,17}{
      \draw[fill = pink] (B\phi)circle(1. 2 mm);
		}
 \foreach \phi in {3,6,9,12,15,18}{
      \draw[fill = cyan] (B\phi)circle(1. 2 mm);	
   }

  \foreach \phi [count = \s from 2] in {1,...,16}
  {
  \draw[-] (B\phi) -- (A\s);
    } 
    \foreach \phi [count = \s from 3] in {1,...,16}
  {
  \draw[-] (B\phi) -- (A\s);
    }   
    \draw[-] (B17) -- (A18);
    \draw[-] (B17) -- (A1);
    \draw[-] (B18) -- (A1);
    \draw[-] (B18) -- (A2);

	\node (n1) at (A1)[right= 1mm]{$u_{1} \hspace{0.4mm} (1,1)$};
	\node at (A2)[right= 1mm]{$w_{42}\hspace{0.4mm}\overline{(1,3)}$};
	\node at (A3)[right= 2mm]{$v_{1}\hspace{0.4mm}(1,2)$};
	\node at (A4)[above right= 1mm]{$u_{2} \hspace{0.4mm}\overline{(1,1)}$};
	\node at (A5)[above = 3mm]{$w_{41}\hspace{0.4mm} (1,3)$};
	\node at (A6)[above left= 1mm]{$v_{2}\hspace{0.4mm} \overline{(1,2)}$};
	\node at (A7)[above left= 1mm]{$u_{13}\hspace{0.4mm}(2,1)$};
	\node at (A11)[left= 1mm]{$w_{1}\hspace{0.4mm}(2,3)$};
	\node at (A9)[left= 0mm]{$v_{13}\hspace{0.2mm}(2,2)$};
	\node at (A10)[left= 1mm]{$u_{14}\hspace{0.4mm}\overline{(2,1)}$};
	\node at (A8)[left= 1mm]{$w_{2}\hspace{0.4mm}\overline{(2,3)}$};
	\node at (A12)[left= 2mm]{$v_{14}\hspace{0.4mm}\overline{(2,2)}$};
	\node at (A13)[below left= 1mm]{$u_{15}\hspace{0.4mm}(3,1)$};
	\node at (A17)[right= 1mm]{$w_{3}\hspace{0.4mm}(3,3)$};
	\node at (A15)[below right= 1mm]{$v_{15}\hspace{0.4mm}(3,2)$};
	\node at (A16)[right= 1mm]{$u_{16}\hspace{0.4mm}\overline{(3,1)}$};
	\node at (A14)[below= 1mm]{$w_{4}\hspace{0.4mm}\overline{(3,3)}$};
	\node at (A18)[right= 1mm]{$v_{16}\hspace{0.4mm}\overline{(3,2)}$};

	\node at (B1)[below= 1mm]{$u_{\mhyphen1}$};
	\node at (B2)[below= 1mm]{$w_{40}$};
	\node at (B3)[below= 1mm]{$v_{\mhyphen2}$};
	\node at (B4)[below= 1mm]{$u_{\mhyphen2}$};
	\node at (B5)[below right= 1mm]{$w_{\mhyphen1}$};
	\node at (B6)[below right= 1mm]{$v_{\mhyphen13}$};
	\node at (B7)[right= 1mm]{$u_{\mhyphen13}$};
	\node at (B8)[right= 1mm]{$w_{\mhyphen2}$};
	\node at (B9)[above right= 1mm]{$v_{\mhyphen14}$};
	\node at (B10)[above = 1mm]{$u_{\mhyphen14}$};
	\node at (B11)[above= 1mm]{$w_{\mhyphen3}$};
	\node at (B12)[above= 1mm]{$v_{\mhyphen15}$};
	\node at (B13)[above= 1mm]{$u_{\mhyphen15}$};
	\node at (B14)[above= 1mm]{$w_{\mhyphen4}$};
	\node at (B15)[above left = 1mm]{$v_{\mhyphen16}$};
	\node at (B16)[above left= 1mm]{$u_{\mhyphen16}$};
	\node at (B17)[left= 1mm]{$w_{\mhyphen5}$};
	\node at (B18)[below left= 1mm]{$v_{\mhyphen1}$};

\end{tikzpicture}
}
\end{minipage}
\hspace{1cm}
\begin{minipage}{.5\textwidth}
\resizebox{8cm}{7cm}{
\begin{tikzpicture}[remember picture]
%	\draw[-] (0,0) circle (4.0cm);
	\draw [green] (0:3.7) arc [radius=3.7 cm, start angle=0, end angle=20];
	  \draw [red] (20:3.7) arc [radius=3.7 cm, start angle=20, end angle=120];
  \draw [green] (120:3.7) arc [radius=3.7 cm, start angle=120, end angle=360];
	
%	\draw[-] (0,0) circle (5.0cm);

  \foreach \phi in {1,...,18}{
    \coordinate (C\phi) at (360/18*\phi:3.7cm);
    \coordinate (D\phi) at (360/18*\phi +30:3.2cm);
%    \draw[myshorten,myedge] (A\phi) arc ( 360/18*\phi: 360/18*(\phi+1): 1cm);
%    \draw[myshorten,myedge] (B\phi) arc ( 360/18*\phi: 360/18*(\phi+1): 2.5cm);
  }
  \foreach \phi in {1,4,7,10,13,16}{
      \draw[fill = green] (C\phi)circle(1. 2 mm);
		}
 \foreach \phi in {2,5,8,11,14,17}{
      \draw[fill = pink] (C\phi)circle(1. 2 mm);
		}
 \foreach \phi in {3,6,9,12,15,18}{
      \draw[fill = cyan] (C\phi)circle(1. 2 mm);
		}

  \foreach \phi in {1,4,7,10,13,16}{
      \draw[fill = green] (D\phi)circle(1. 2 mm);
		}
 \foreach \phi in {2,5,8,11,14,17}{
      \draw[fill = pink] (D\phi)circle(1. 2 mm);
		}
 \foreach \phi in {3,6,9,12,15,18}{
      \draw[fill = cyan] (D\phi)circle(1. 2 mm);	
   }

  \foreach \phi [count = \s from 2] in {1,...,16}
  {
  \draw[-] (D\phi) -- (C\s);
    } 
    \foreach \phi [count = \s from 3] in {1,...,16}
  {
  \draw[-] (D\phi) -- (C\s);
    }   
    \draw[-] (D17) -- (C18);
    \draw[-] (D17) -- (C1);
    \draw[-] (D18) -- (C1);
    \draw[-] (D18) -- (C2);	

	\node at (C1)[right= 1mm]{$u_{3} \hspace{0.4mm} (1,1)$};
	\node at (C2)[right= 1mm]{$w_{39}\hspace{0.4mm}\overline{(1,3)}$};
	\node at (C3)[right= 2mm]{$v_{3}\hspace{0.4mm}(1,2)$};
	\node at (C4)[above right= 1mm]{$u_{4} \hspace{0.4mm}\overline{(1,1)}$};
	\node at (C5)[above = 3mm]{$w_{38}\hspace{0.4mm} (1,3)$};
	\node (n2) at (C6)[above left= 1mm]{$v_{4}\hspace{0.4mm} \overline{(1,2)}$}(C6);
	\node at (C7)[above left= 1mm]{$u_{17}\hspace{0.4mm}(2,1)$};
	\node at (C11)[left= 1mm]{$w_{5}\hspace{0.4mm}(2,3)$};
	\node at (C9)[left= 1mm]{$v_{17}\hspace{0.4mm}(2,2)$};
	\node at (C10)[left= 1mm]{$u_{18}\hspace{0.4mm}\overline{(2,1)}$};
	\node at (C8)[left= 1mm]{$w_{6}\hspace{0.4mm}\overline{(2,3)}$};
	\node at (C12)[left= 2mm]{$v_{18}\hspace{0.4mm}\overline{(2,2)}$};
	\node at (C13)[below left= 1mm]{$u_{19}\hspace{0.4mm}(3,1)$};
	\node at (C17)[right= 1mm]{$w_{7}\hspace{0.4mm}(3,3)$};
	\node at (C15)[below right= 1mm]{$v_{19}\hspace{0.4mm}(3,2)$};
	\node at (C16)[right= 1mm]{$u_{20}\hspace{0.4mm}\overline{(3,1)}$};
	\node at (C14)[below= 1mm]{$w_{8}\hspace{0.4mm}\overline{(3,3)}$};
	\node at (C18)[right= 1mm]{$v_{20}\hspace{0mm}\overline{(3,2)}$};

	\node at (D1)[below= 1mm]{$u_{\mhyphen3}$};
	\node at (D2)[below= 1mm]{$w_{35}$};
	\node at (D3)[below= 1mm]{$v_{\mhyphen4}$};
	\node at (D4)[below= 1mm]{$u_{\mhyphen4}$};
	\node at (D5)[below right= 1mm]{$w_{\mhyphen6}$};
	\node at (D6)[below right= 1mm]{$v_{\mhyphen17}$};
	\node at (D7)[right= 1mm]{$u_{\mhyphen17}$};
	\node at (D8)[right= 1mm]{$w_{\mhyphen7}$};
	\node at (D9)[above right= 1mm]{$v_{\mhyphen18}$};
	\node at (D10)[above = 1mm]{$u_{\mhyphen18}$};
	\node at (D11)[above= 1mm]{$w_{\mhyphen8}$};
	\node at (D12)[above= 1mm]{$v_{\mhyphen19}$};
	\node at (D13)[above= 1mm]{$u_{\mhyphen19}$};
	\node at (D14)[above= 1mm]{$w_{\mhyphen9}$};
	\node at (D15)[above left = 1mm]{$v_{\mhyphen20}$};
	\node at (D16)[above left= 1mm]{$u_{\mhyphen20}$};
	\node at (D17)[left= 1mm]{$w_{\mhyphen10}$};
	\node at (D18)[below left= 1mm]{$v_{\mhyphen3}$};
       
\end{tikzpicture}}
\end{minipage}

\vspace{2cm}
\hspace{3.7cm}
\resizebox{8cm}{7cm}{ 
\begin{tikzpicture}[remember picture]

%[
%    myedge/.style={thick,draw=black, postaction={decorate},
%      decoration={markings,mark=at position .6 with {\arrow[black]{triangle 45}}}},
%    myshorten/.style={shorten <= 2pt, shorten >= 2pt}]
%	\draw[-] (0,0) circle (4.0cm);
%%	\draw[-] (0,0) circle (5.0cm);
	\draw [green] (0:3.7) arc [radius=3.7 cm, start angle=0, end angle=20];
	  \draw [red] (20:3.7) arc [radius=3.7 cm, start angle=20, end angle=120];
  \draw [green] (120:3.7) arc [radius=3.7 cm, start angle=120, end angle=360];
	
%	\draw[-] (0,0) circle (5.0cm);

  \foreach \phi in {1,...,18}{
    \coordinate (E\phi) at (360/18*\phi:3.7cm);
    \coordinate (F\phi) at (360/18*\phi +30:3.2cm);
%    \draw[myshorten,myedge] (A\phi) arc ( 360/18*\phi: 360/18*(\phi+1): 1cm);
%    \draw[myshorten,myedge] (B\phi) arc ( 360/18*\phi: 360/18*(\phi+1): 2.5cm);
  }
  \foreach \phi in {1,4,7,10,13,16}{
      \draw[fill = green] (E\phi)circle(1. 2 mm);
		}
 \foreach \phi in {2,5,8,11,14,17}{
      \draw[fill = pink] (E\phi)circle(1. 2 mm);
		}
 \foreach \phi in {3,6,9,12,15,18}{
      \draw[fill = cyan] (E\phi)circle(1. 2 mm);
		}

  \foreach \phi in {1,4,7,10,13,16}{
      \draw[fill = green] (F\phi)circle(1. 2 mm);
		}
 \foreach \phi in {2,5,8,11,14,17}{
      \draw[fill = pink] (F\phi)circle(1. 2 mm);
		}
 \foreach \phi in {3,6,9,12,15,18}{
      \draw[fill = cyan] (F\phi)circle(1. 2 mm);	
   }

  \foreach \phi [count = \s from 2] in {1,...,16}
  {
  \draw[-] (F\phi) -- (E\s);
    } 
    \foreach \phi [count = \s from 3] in {1,...,16}
  {
  \draw[-] (F\phi) -- (E\s);
    }   
    \draw[-] (F17) -- (E18);
    \draw[-] (F17) -- (E1);
    \draw[-] (F18) -- (E1);
    \draw[-] (F18) -- (E2);	
     
	\node at (E1)[right= 1mm]{$u_{5} \hspace{0.4mm} (1,1)$};
	\node at (E2)[right= 1mm]{$w_{37}\hspace{0.4mm}\overline{(1,3)}$};
	\node at (E3)[right= 2mm]{$v_{5}\hspace{0.4mm}(1,2)$};
	\node at (E4)[above right= 1mm]{$u_{6} \hspace{0.4mm}\overline{(1,1)}$};
	\node (n3) at (E5)[above = 3mm]{$w_{36}\hspace{0.4mm} (1,3)$}(E5);
	\node at (E6)[above left= 1mm]{$v_{6}\hspace{0.4mm} \overline{(1,2)}$};
	\node at (E7)[above left= 1mm]{$u_{21}\hspace{0.4mm}(2,1)$};
	\node at (E11)[left= 1mm]{$w_{9}\hspace{0.4mm}(2,3)$};
	\node at (E9)[left= 1mm]{$v_{21}\hspace{0.4mm}(2,2)$};
	\node at (E10)[left= 1mm]{$u_{22}\hspace{0.4mm}\overline{(2,1)}$};
	\node at (E8)[left= 1mm]{$w_{10}\hspace{0.4mm}\overline{(2,3)}$};
	\node at (E12)[left= 2mm]{$v_{22}\hspace{0.4mm}\overline{(2,2)}$};
	\node at (E13)[below left= 1mm]{$u_{23}\hspace{0.4mm}(3,1)$};
	\node at (E17)[right= 1mm]{$w_{11}\hspace{0.4mm}(3,3)$};
	\node at (E15)[below right= 1mm]{$v_{23}\hspace{0.4mm}(3,2)$};
	\node at (E16)[right= 1mm]{$u_{24}\hspace{0.4mm}\overline{(3,1)}$};
	\node at (E14)[below= 1mm]{$w_{12}\hspace{0.4mm}\overline{(3,3)}$};
	\node at (E18)[right= 1mm]{$v_{24}\hspace{0.4mm}\overline{(3,2)}$};

	\node at (F1)[below= 1mm]{$u_{\mhyphen5}$};
	\node at (F2)[below= 1mm]{$w_{34}$};
	\node at (F3)[below= 1mm]{$v_{\mhyphen6}$};
	\node at (F4)[below= 1mm]{$u_{\mhyphen6}$};
	\node at (F5)[below right= 1mm]{$w_{\mhyphen11}$};
	\node at (F6)[below right= 1mm]{$v_{\mhyphen21}$};
	\node at (F7)[right= 1mm]{$u_{\mhyphen21}$};
	\node at (F8)[right= 1mm]{$w_{\mhyphen12}$};
	\node at (F9)[above right= 1mm]{$v_{\mhyphen22}$};
	\node at (F10)[above = 1mm]{$u_{\mhyphen22}$};
	\node at (F11)[above= 1mm]{$w_{\mhyphen13}$};
	\node at (F12)[above= 1mm]{$v_{\mhyphen23}$};
	\node at (F13)[above= 1mm]{$u_{\mhyphen23}$};
	\node at (F14)[above= 1mm]{$w_{\mhyphen14}$};
	\node at (F15)[above left = 1mm]{$v_{\mhyphen24}$};
	\node at (F16)[above left= 1mm]{$u_{\mhyphen24}$};
	\node at (F17)[left= 1mm]{$w_{\mhyphen15}$};
	\node at (F18)[below left= 1mm]{$v_{\mhyphen5}$};     
     
\end{tikzpicture}}

\begin{tikzpicture}[remember picture,overlay]
\draw[-, color = blue] (5.6,13.92) to (11.8,15.4);
\draw[-, color = blue] (11.8,15.4) to (9.92,5.78);
\draw[-, color = blue] (9.92,5.78) to (5.6,13.92);
\end{tikzpicture}
\caption{The variable gadgets $G^{\phi}_k, k=1,2,3$ for $n=6$. The hyperedges with a red arc correspond to the first occurrence of variables. Notice the difference in labelling of $W$ nodes. The clause edge corresponds to the clause $Y_1 \lor \overline{Y_2} \lor \overline{Y_3}$.\label{fig:fullgadget} }
\end{figure}

We now show that with the above ordering, the set of hyperedges $E$ of the hypergraph $G^{\phi}$ indeed is an antichain.

To simplify the argument, we divide the set of hyperedges in two parts $E =\mathcal{A} \cupdot\mathcal{B}$:
\begin{itemize}
\item Set $\mathcal{A}$: This set consists of local hyperedges in which both the literal-nodes correspond to the first occurrence of variables. We also include the 3-clause hyperedges.
\item Set $\mathcal{B}$: The set consisting of the remaining hyperedges, i.e., the ones in which at least one of the literal-nodes correspond to the second or the third occurrences of variables. We also include the 2-clause hyperedges.
\end{itemize}

We first argue that the subset $\mathcal{B}$ is an antichain.\\
We note that in $\mathcal{B}$, the literal-nodes are all labelled positive $(2n + 2(i-2) + 4(k-1)+j)$, $i \in \{2,3 \}$, $k\in [n]$, $j \in [4]$, while the free-nodes are all labelled negative $(-2n - 4(k-1)-\ell)$, $k \in [n]$, $\ell \in [4]$, for $U$ and $V$ nodes, whereas $(-5(k-1) - \ell)$, $k \in [n]$, $\ell \in [5]$ for $W$ nodes, and it is easy to verify that as the labels of the literal-node increase, the labels along the free-node decrease. \\
Now we take two arbitrary elements of the the set $\mathcal{B}$. Recall that every hyperedge in $\mathcal{B}$ contains exactly one free-node. Now the free-node will either be in the same partition or in different ones. \\
If they are in different ones, we are done because we have a pair of coordinates such that, in one of them, one hyperedge is labelled positive while the other is labelled negative, while the opposite happens in the other coordinate.
If the free nodes are in the same coordinate, we are done again because as the literal coordinate increases, the free coordinate decreases. \\
Note that, since we have already shuffled the nodes with free $W$ nodes taking the 2-clause hyperedges into account, the 2-clause hyperedges are also taken care off.

Now, we argue that given an arbitrary hyperedge of the set $\mathcal{A}$, and an arbitrary hyperedge of the set $\mathcal{B}$, they are incomparable too.\\
For this, we notice that, the labels of the $W$ nodes of all the hyperedges in $\mathcal{A}$ are higher than the labels of all the $W$ nodes of the hyperedges in $\mathcal{B}$. For this, we simply note that range of the $W$ labels of the second and the third occurrence (set $\mathcal{B}$) is $\{-5n, \ldots ,4n\} \setminus \{ 0 \}$, whereas the $W$ labels of the first occurrence ($\mathcal{A}$) has the range from $\{4n+1, \ldots ,7n\}$. Secondly, notice that the labels of the $U$ and $V$ literal-nodes at the second and third occurrences, i.e., from the edges of set $\mathcal{B}$ (range $\{2n+1, \ldots, 6n \}$) are all higher than that of the first occurrence i.e. from the edges of the set $\mathcal{A}$ (range $\{1, \ldots, 2n \}$).
 
We are done since for every pair of hyperedges $(h_a,h_b)$, where $h_a \in \mathcal{A}$ and $h_b \in \mathcal{B}$, we have that the $W$coordinate of $h_a$ will be higher than that of $h_b$, whereas the among the other two coordinates, whichever is positive (i.e. corresponds to a literal-node) in $h_b$ will be higher than the correpsonding coordinate in $h_a$. 

Finally we are left to show that $\mathcal{A}$ is also an antichain. 

We remind the reader that we have named the variables such that every 3-clause comprises of variables from only one triple of variables i.e. every 3-clause involves $Y_{3(q-1) +1},Y_{3(q-1) +2},Y_{3(q-1) +3} $ at first, second and third position respectively, for some $q>0$. 
Now first of all we notice that for a pair of hyperedges which come from a different triple of variables, we are done, because $W$ coordinates of a higher triple are all lower than the $W$ coordinates of a lower triple, since the labels are $(7n - 9(q-1) -\ell), \ell \in \{0, \ldots,8 \} $ for $q-$th triple of variables $Y_{3(q-1) +1},Y_{3(q-1) +2},Y_{3(q-1) +3} $, whereas the positive coordinate among $U$ or $V$ will be higher for the higher triple (lables are $4(k-1)+\ell, \ell \in [2]$).
When they are in the same triple of variables, it helps to remark that there are three kinds of hyperedges in $\mathcal{A}$, i.e. $\mathcal{A} = \mathcal{A}_1 \cupdot \mathcal{A}_2 \cupdot \mathcal{A}_c$:
\begin{itemize}
\item $\mathcal{A}_1$: the ones where the free-nodes belong to $U$ or $V$. These hyperedges have exactly one negative coordinate, which will either be in the $U$ coordinate or the $V$ coordinate.
\item $\mathcal{A}_2$: the ones where the free nodes belong to $W$. All the coordinates are positive.
\item $\mathcal{A}_c$: the set of 3-clause hyperedges: All the coordinates are again positive, as all the nodes are literal-nodes.
\end{itemize} 

Now, we need to compare the hyperedges of $\mathcal{A}_1, \mathcal{A}_1$ and $\mathcal{A}_c$ with each other and within themselves when they all belong to the same triple of variables, say $q-$th triple, $Y_{3(q-1) +1},Y_{3(q-1) +2},Y_{3(q-1) +3} $ for some $q \in [t]$. We remind the reader that the labelling of the $W$ nodes that appear in $\mathcal{A}$ varies depending on whether the corresponding index $k = 3(q-1) +1$, $3(q-1) +1,$ or $3(q-1) +3$.

There are six possible cases:

\begin{enumerate}[i.]
\item $\mathcal{A}_1$: same proof that was given for the elements of $\mathcal{B}$, where also we had exactly one negative coordinate.

\item $\mathcal{A}_2$: for the higher variable, the $W$ coordinate is lower (labels are $7n - 9(q-1) -2$ for $ k = 3(q-1)+1$, $7n - 9(q-1) -7$ for $k = 3(q-1)+2$ and $7n - 9(q-1) -8$ for $k = 3(q-1)+3$), while the other two coordinates are higher, since both $U$ and $V$ labels are $2(k-1)+1,2$.

\item $\mathcal{A}_c$: two different clauses clearly belong to different triple of variables: already taken care of above.

\item $\mathcal{A}_1-\mathcal{A}_2$ $(h_{a_1} \in \mathcal{A}_1, h_{a_2} \in \mathcal{A}_2)$: Here we have two cases: namely, either $h_{a_1}$ belonging to a higher variable, or $h_{a_1}$ belonging to the same or lower variable as compared to $h_{a_2}$. In the first case, one of the $U$ or $V$ coordinate of $h_{a_1}$ (whichever is positive) will be higher, while the other coordinate being negative will be lower than that of $h_{a_2}$ (whose all coordinates are positive). In the second case, we note that the $W$ coordinate of $h_{a_2}$ will be lower, since for the same variable, it has the lowest $W$ coordinate (being $7n - 9(q-1) -2$ versus $7n - 9(q-1)$, $7n - 9(q-1)-1$ for $k = 3(q-1) +1 $, $7n - 9(q-1) -7$ versus $7n - 9(q-1)-3$, $7n - 9(q-1)-4$ for $k = 3(q-1) +2$ and $7n - 3(k-1) -8$ versus $7n - 9(q-1)-5$, $7n - 9(q-1)-6$ for $k = 3(q-1) +3$), and as we go up the variables, $W$ coordinate decreases, while at least one of the other two coordinate will be higher, i.e., in the coordinate in which $h_{a_1}$ is negative and $h_{a_2}$ is positive.

\item $\mathcal{A}_1-\mathcal{A}_c$ $(h_{a_1} \in \mathcal{A}_1, h_{a_c} \in \mathcal{A}_c)$: When $h_{a_1}$ belongs to $G^{\phi}_{3(q-1)+1}$ or $G^{\phi}_{3(q-1)+2}$, its $W$ coordinate will be higher than that of $h_{a_c}$, since for the clause hyperedge $h_{a_c}$, the $W$ node is picked from $G^{\phi}_{3(q-1)+3}$.  However, one of the other two coordinates in $h_{a_1}$ is negative. So, it will be lower than that of $h_{a_c}$. So, we are done. When $h_{a_1}$ belongs to $G^{\phi}_{3(q-1)+3}$, both $h_{a_1}$ and $h_{a_c}$ might share the $W$ coordinate. However, in such $h_{a_1}$, the positive node among the $U$ and $V$ coordinate will be higher than that of $h_{a_c}$, since $h_{a_1}$ comes from the highest variable among the triple, and both $U$ and $V$ coordinate increase with higher variables, being labelled $2(k-1)+1,2$, whereas the negative coordinate will of course be lower than that of $h_{a_c}$ which has no negative coordinate.

\item $\mathcal{A}_2-\mathcal{A}_c$ $(h_{a_2} \in \mathcal{A}_2, h_{a_c} \in \mathcal{A}_c)$: Here when $h_{a_2} \in G^{\phi}_{3(q-1)+1}$, its $V$ coordinate will be less since $Y_{3(q-1)+1}$ is the lowest variable, whereas the $V$ coordinate of the clause hyperedge $h_{a_c}$ is picked from $G^{3(q-1)+2}$. However, the $W$ coordinate will be higher for $h_{a_2}$ as it is labelled $7n-9(q-1)-2$, whereas the clause gets the $W$ coordinate corresponding to the $G^{\phi}_{3(q-1)+3}$ and hence the label $7n-9(q-1)-5$ or $7n-9(q-1)-6$ .
Whereas when $h_{a_2} \in G^{\phi}_{3(q-1)+2}$ or $G^{\phi}_{3(q-1)+3}$, the $W$ coordinate will be lower for $h_{a_2}$ (labelled $7n-9(q-1)-7$ or $7n-9(q-1)-8$ respectively) than $h_{a_c}$ (labelled $7n-9(q-1)-5$ or $7n-9(q-1)-6$), whereas the $U$ coordinate of $h_{a_2}$ will be higher, since the clause hyperedge $h_{a_c}$ gets the $U$ coordinate corresponding to variable $Y_{3(q-1)+1}$ which is the lowest variable within the triple and hence has the lowest $U$ coordinate ($U$ labels being $2(k-1)+1,2$).
\end{enumerate}
\end{proof}

\subsection{NP-Hardness of minrank} \label{sec:hard}

In this section we prove $\NP$-hardness of $\HMinRk$ by reducing it to the following problem:

\begin{problem}{$\HQSAT_{S, F}$}
  Given a set of quadratic forms with coefficients from $S$, represented by lists of coefficients, determine if it has a common zero over $F$.
\end{problem}

To implement the reduction, we need to perform linear algebra computations with elements of the field.
\begin{definition}
  An \emph{effective field} is a finite or countable field $F$ with a binary encoding of elements of $F$ such that the following operations can be performed in time polynomial in the length of the encoding of arguments:
  \begin{itemize}
    \item multiplication and addition of two elements over $F$,
    \item multiplication of an arbitrary number of matrices over $F$ (follows from the first item),
    \item equality comparison of two elements of $F$,
    \item division of two elements of $F$ (if the denominator is zero, the algorithm should fail).
  \end{itemize}
Furthermore, we want that polynomial identity testing is in $\BPP$, that is, there is a $\BPP$-machine that
given an algebraic circuit computing a polynomial over $F$, decides in whether this polynomial is identically zero.
\end{definition}

In our paper, we usually deal with polynomials over uncountable fields like $\bbC$. In the algebraic complexity setting,
this is no problem. However, when we want to compute with Turing machines, we have to restrict ourselves to appropriate
subfields. This is modelled by effective fields.
In particular, $\bbQ$ is effective and the natural effective subfield of $\bbR$ and $\bbQ + i \bbQ$ 
is natural choice for $\bbC$. Finite fields are effective,
when we drop the last condition about identity testing, which we only need in the second part of this section.

Efficient multiplication of several matrices implies that products and linear combinations of elements can also be computed in polynomial time.
It also allows for various polynomial-time linear algebra procedures. In particular, we are interested in the following:

\begin{theorem}
  For an effective field $K$ there is a polynomial time algorithm which, given a matrix $A$ over $K$, computes a basis of $\ker A$.
\end{theorem}
\begin{proof}
  Determinants of matrices over an effective field are computable in polynomial time, because determinant can be represented as an iterated matrix multiplication of polynomial size (see e.~g.~\cite{IKENMEYER20172911}).
  This allows computing the inverse of a nonsingular matrix.
  Also, we can find one of the maximal nonzero minors of a given nonzero matrix, by starting from any nonzero entry and trying to enlarge the minor by checking all rows and columns at each step.
  We can then compute the basis of the kernel by basic linear algebra.
\end{proof}

Hillar and Lim~\cite[Thm.~2.6]{DBLP:journals/jacm/HillarL13} proved that $\HQSAT$ is $\NP$-hard over the fields $\bbR$ and $\bbC$.
Their proof also works for any field of characteristic different from $3$ containing cubic roots of unity.
The $\NP$-hardness for arbitrary fields was proven by Grenet, Koiran and Portier in~\cite{DBLP:journals/jc/GrenetKP13}.
We give another proof for arbitrary fields based on the idea of Hillar and Lim.
Compared to~\cite{DBLP:journals/jc/GrenetKP13}, we describe a general construction for all fields
instead of treating characteristic 2 as a special case, and only use coefficients from $\{-1, 0, 1\}$.

\begin{theorem}
  $\HQSAT_{\{0,1,-1\}, F}$ is $\NP$-hard for any field $F$.
\end{theorem}
\begin{proof}
  We reduce from graph $3$-colorability.

  Given a graph $G = (V, E)$, we will construct a system of quadratic homogeneous equation, solutions of which correspond to colorings of the graph.
  The set of variables consists of two variables $x_v$ and $y_v$ for each vertex $v \in V$ and one additional variable $z$.
  Consider a system of homogeneous quadratic equations which contains for each vertex $v$ the three equations
  \begin{gather*}
      x_v y_v = 0 \\
      x_v^2 - x_v z = 0 \\
      y_v^2 - y_v z = 0 
  \end{gather*}
  and for each edge $(v, w) \in E$ the equation
  \[
    x_v^2 + y_v^2 + x_w^2 + y_w^2 - x_v y_w - x_w y_w - z^2 = 0
  \]
  If $z = 0$, then from vertex equations we deduce $x_v = y_v = 0$ for all $v \in V$. Therefore, a nontrivial solution must have nonzero $z$.
  We can scale it so that $z = 1$.
  When $z = 1$, the vertex equations give $(x_v, y_v) \in \{(0, 0), (0, 1), (1, 0)\}$.
  Restricted to these values, the left-hand side of the edge equation has the following values:
  \begin{center}
    \begin{tabular}{|c||c|c|c|}
      \hline
      \backslashbox{$v$}{$w$} & $(0,0)$ & $(0,1)$ & $(1,0)$ \\
      \hline
      \hline
      $(0,0)$ & $-1$ & $0$ & $0$ \\
      \hline
      $(0,1)$ & $0$ & $1$ & $0$ \\ 
      \hline
      $(1,0)$ & $0$ & $0$ & $1$ \\ 
      \hline
    \end{tabular}
  \end{center}
  That is, the edge equation forces the tuples $(x_v, y_v)$ and $(x_w, y_w)$ to be different.
  Thus, nontrivial solutions with $z = 1$ are in one-to-one correspondence with colorings of the graph $G$ into three colors, given by the three possible solutions of the vertex equations.
\end{proof}

\begin{theorem}
  Let $F$ be a field and $K$ be an effective subfield of $F$. Then $\HMinRkU_{K, F}$ is polynomial-time equivalent to $\HQSAT_{K, F}$.
\end{theorem}
\begin{proof}
  To reduce from $\HMinRkU$ to $\HQSAT$, note that the condition $\rk(Tx) \leq 1$ can be expressed by homogeneous quadratic equations on $x$,
  namely, vanishing of $2 \times 2$ minors of the matrix of linear forms $Tx$.

  Now we describe the reduction from $\HQSAT$ to $\HMinRkU$. Let $k$ be a number of given quadratic forms and $n$ be the number of variables.
  Each quadratic form $q(x) = \sum_{1 \leq i \leq j \leq n} a_{ij} x_i x_j$ on $F^n$ corresponds to a linear form $Q(X) = \sum_{1\leq i \leq j \leq n} a_{ij} x_{ij}$ on the space $\Sym^2 F^n \subset F^n \otimes F^n$ of symmetric matrices,
  and a vector $x$ is a zero of $q$ if and only if $x \otimes x$ is a zero of $Q$.
  Therefore, a set of $k$ linear forms on $F^n$ corresponds to a linear map $L \colon \Sym^2 F^n \to F^k$ given by a matrix consisting from the coefficients of quadratic forms,
  and $x$ is a common zero if and only if $x \otimes x$ is contained in $\ker L$.
  Since all the coefficients lie in $K$, the map $L$ is an extension of a linear map $\Sym^2 S^n \to S^k$, and its kernel has a basis consisting of vectors in $\Sym^2 S^n$, which can be computed in polynomial time.
  Let $A_1, \dots, A_m$ be such basis and $T = \sum_{i = 1}^m e_i \otimes A_i \in S^m \otimes S^n \otimes S^n$.
  Nontrivial common zeros $x \in F^n$ of the original set of quadratic forms corresponds to rank $1$ symmetric matrices $x \otimes x$ which can be presented as a nontrivial linear combination $\sum_{i = 1}^m y_i A_i$ with $y_i \in F$ or, equivalently, as a contraction $Ty$ with nonzero $y \in F^m$.
  This is the resulting instance of $\HMinRkU$ problem.
\end{proof}

\begin{corollary}
  Let $F$ be a field and $K$ be an effective subfield of $F$. Then $\HMinRkU_{K, F}$ is $\NP$-hard.
\end{corollary}

The $\HMinRk$ problem is also hard in other regimes.

\begin{theorem}
  Let $F$ be a field of characteristic $0$ and $K$ be an effective subfield of $F$. Then $\HMinRk_{\bbQ, F}$ is $\NP$-hard for $n \times (2n+1) \times (2n+1)$ tensors and $r = n + 1$.
\end{theorem}
\begin{proof}
  The proof is based on a similar theorem for finite fields is sketched in~\cite[\S 3.3]{DBLP:conf/asiacrypt/Courtois01}, which uses $\NP$-completeness of the minimum distance problem for linear codes proved in~\cite{DBLP:journals/tit/Vardy97}.

  We reduce from a variant of the \textsc{Partition} problem: given a list of $2n$ integers such that each integer appears at most $n - 2$ times, determine if it can be partitioned into 2 subsets of size $n$ with equal sums. $\NP$-completeness of this variant is noted in~\cite[SP12]{DBLP:books/fm/GareyJ79}.

  From the input $\{a_i, \dots, a_{2n}\}$ of the \textsc{Partition} problem construct a $(n + 1) \times (2n + 1)$ matrix
  \[
    A = \begin{bmatrix}
    1 & 1 & \dots & 1 & 0 \\
    a_1 & a_2 & \dots & a_{2n} & 0 \\
    a_1^2 & a_2^2 & \dots & a_{2n}^2 & 0 \\
    \vdots & \vdots & \ddots & \vdots & \vdots \\
    a_1^{n - 2} & a_2^{n - 2} & \dots & a_{2n}^{n - 2} & 0 \\
    a_1^{n - 1} & a_2^{n - 1} & \dots & a_{2n}^{n - 1} & 1 \\
    a_1^n & a_2^n & \dots & a_{2n} & S/2
    \end{bmatrix}
  \]
  where $S$ is the sum of all $a_i$.
  From the properties of Vandermonde determinants we see that any $(n + 1) \times (n + 1)$ minor is nonzero if it does not contain the last column.
  If a minor does contain the last column and columns $i_1, \dots, i_n$, it vanishes if and only if $S/2 = a_{i_1} + \dots + a_{i_n}$~\cite[Lem.~1]{DBLP:journals/tit/Vardy97}.
  Thus, the matrix $A$ has rank $n + 1$. Moreover, it has $n + 1$ linearly dependent columns if and only if the original \textsc{Partition} problem has a solution.

  Let $b_1, \dots, b_n$ be a basis of $\ker A$. Since subsets of $k$ linearly dependent columns corresponds to vectors in $\ker A$ which have at most $k$ nonzero coordinates, the original problem has a solution if and only if there is a nonzero linear combination of $b_i$ with at most $n + 1$ nonzero coordinates.

  Let $B_i$ be a $(2n + 1) \times (2n + 1)$ matrix constructed from $b_i$ by placing its coordinates on the diagonal. The rank of a linear combination of $B_i$ is equal to the number of nonzero coordinates in the corresponding linear combination of vectors $b_i$. Thus, the answer to the $\HMinRk$ problem for the $n \times (2n + 1) \times (2n + 1)$ tensor $\sum_{i = 1}^{n} e_i \otimes B_i$ and $r = n + 1$ determines the answer to the original problem.
\end{proof}

From the facts that the minrank problem is $\NP$-hard and that minrank varieties can be written as orbit closures,
we immediately get the following hardness result for the orbit closure containment problem.

\begin{corollary} \label{cor:occhard}
Given two tensors $t$ and $t'$, deciding whether the orbit closure of $t$ is contained in the orbit closure of $t'$
(under the usual $\GL_n \times \GL_n \times \GL_n$ action) is $\NP$-hard.
\end{corollary}

\subsection{Slice rank and minrank varieties and algebraic natural proofs}

We have found a lot of equations for the minrank varieties and it is a natural question how hard these equations are.
In particular, in the GCT setting, we have a sequence of varieties $V_n$ and a sequence of points $x_n$ and want
to prove that $x_n$ is not in $V_n$. This is done by giving equations $f_n$ such that $f_n$
vanishes on $V_n$, but $f_n(x_n) \not= 0$. The meta-question is how ``difficult'' is it to prove that $f_n$
has the desired properties. That is, why is progress on algebraic circuit lower bounds so hard?
For instance, if $f_n$ has high circuit complexity, then it is very unlikely
that we will be able to prove  $f_n(x_n) \not= 0$ by evaluating this circuit. In turns out that when testing 
membership in $V_n$ is a hard problem, then this high circuit complexity is in some sense unavoidable. 
To deal with this questions,
we generalize the methods from \cite{DBLP:conf/stoc/BlaserIJL18} and make them applicable to varieties for which
the membership problem is hard.

We call a sequence $(V_n)$ a p-family of varieties if $V_n$ is a subset of $F^{p(n)}$ 
for some polynomially bounded function $p$.

\begin{definition}
A family of varieties $(V_n)$ is \emph{polynomially definable}, if for each $n$, there are polynomials $f_1,\dots,f_m$ such that
$V_n$ is the common zero set of these polynomials and $L(f_i)$ is polynomially bounded in $n$ for all $1 \le i \le m$.
\end{definition}

Here $L(f_i)$ denotes the algebraic circuit complexity of $f_i$, that is, the size of a smallest circuit computing $f$.
Note that we do not require that $m$ is polynomially bounded in $n$. 
%\textbf{Markus:} Maybe definable is a bad name, since
%already used in definition of VNP.

% \framebox{\textcolor{red}{\textbf{What is a p-family of varieties?}}}

\begin{definition} \label{def:generated}
Let $F$ be a field and $K$ be an effective subfield. A p-family of varieties $(V_n)$ with $V_n \subseteq F^{p(n)}$ is \emph{uniformly generated}
if for all $n$, there are polynomials $g_1,\dots,g_{p(n)}$ over $K$ such that 
\begin{enumerate}
\item the image of $(g_1,\dots,g_{p(n)})$ is dense in $V_n$,
\item each $g_i$ has polynomial circuit complexity, and
\item there is a polynomial time bounded Turing machine $M$ that given $n$ in unary, 
outputs for each $g_i$ an arithmetic circuit. 
\end{enumerate}
\end{definition}

The $(V_n)$-membership problem is the following decision problem: Given $n$ and an encoding of a point $x \in S^{p(n)}$,
decide whether $x \in V_n$.

\begin{theorem} \label{thm:hardproof}
Let $F$ be a field and $K$ be an effective subfield.
Let $V = (V_n)$ be a p-family of varieties such that $V$ is polynomially definable over $K$ and uniformly generated
and the $V$-membership problem is $\NP$-hard. Then $\coNP \subseteq \exists \BPP$.
\end{theorem}

\begin{proof}
We give an $\exists \BPP$-algorithm for the $V$-non-membership problem, that is given a point 
$x = (x_1,\dots,x_{p(n)}) \in S^{p(n)}$,
decide whether $x \notin V_n$. Since $V$-membership is $\NP$-hard, $V$-non-membership is $\coNP$-hard
and the result follows. 
The idea is to guess an equation $f$ of the variety $V_n$ such that $f(x) \not= 0$. 
Since $V$ is polynomially definable, there is a set of defining equations of $V_n$ that all have
polynomial circuit complexity. Of course, we need to check that $f$ vanishes indeed of $V_n$.
The algorithm works as follows:
\begin{enumerate}
\item Guess a circuit $C$ of size polynomial in $n$ computing a polynomial $f(X_1,\dots,X_{p(n)})$.
\item Generate the circuits $D_1,\dots,D_{p(n)}$ computing polynomials $g_1,\dots,g_{p(n)}$
as in Definition~\ref{def:generated}.
\item Use polynomial identity testing to check whether $C(g_1,\dots,g_{p(n)})$ is identically zero.
If not, reject.
\item Otherwise, use polynomial identity testing to check whether $C(x_1,\dots,x_{p(n)})$ 
is identically zero. If yes, reject. Otherwise accept.
\end{enumerate}

Since polynomial identity testing over $K$ can be done in $\BPP$, this is clearly an
$\exists\BPP$-algorithm. 

Assume that $x$ is not in the variety. Then there is an equation
of polynomial circuit complexity $f$ that vanishes on $V_n$ such that $f(x) \not= 0$ by the definiability of $V$.
Assume we guessed a circuit $C$ for $f$ in the first step. Since the image of $(g_1,\dots,g_{p(n)})$
lies in $V_n$, $C(g_1,\dots,g_{p(n)})$ will not be identically zero. We pass the test in step 3 with probility $1 - \epsilon$.
Since $f(x) \not= 0$, we accept with probility $1 - \epsilon$ in step 4. Therefore, the overall 
acceptance probability is bounded by $1 - 2\epsilon$.

Now assume that $x \in V_n$. If the guessed circuit computes an equation $f$ of $V_n$,
then we will reject with probability $1 - \epsilon$ in step $4$. If $f$ is not an equation
of $V_n$, then we reject in step $1 - \epsilon$ in step $3$. In both cases the acceptance
probability is bounded by $\epsilon$. This shows the correctness of the algorithm.
\end{proof}

\begin{lemma} \label{lem:orbit}
Let $(V_n) \subseteq F^{p(n)}$ be a p-family of varieties. Let $(G_n)$ be a sequence of groups and $(u_n)$ be a sequence of vectors such that 
$V_n$ is the $G_n$-orbit closure of $u_n$. If for a generic element $g \in G_n$, the coordinate functions $(\gamma_1,\dots,\gamma_{p(n)})$
of $g u_n$ can be described by polynomial size circuits $(C_1,\dots,C_{p(n)})$ and the mapping $1^n \mapsto (C_1,\dots,C_{p(n)})$
is polynomial time computable, then $(V_n)$ is uniformly generated.
\end{lemma} 

\begin{proof} Since $V_n$ is an orbit closure, the orbit lies dense in $V_n$ by definition. The other two 
items in Definition~\ref{def:generated} follow from the prerequisites of the lemma.
\end{proof}
 
\begin{remark}
The same statement is true, if every $V_n$ is not an orbit closure but an intersection of an orbit closure with a vector space.
The proof is almost identical.
\end{remark}

\begin{corollary}
Let $S$ be an effective subfield of $F$. 
For infinitely many $n$, there is an $m$, a tensor $t \in S^{m \times n \times n}$
and a value $r$ such that there is no 
algebraic $\poly(n)$-natural proof for the fact that the minrank of $t$ is greater than $r$ 
unless $\coNP \subseteq \exists \BPP$.
\end{corollary}

\begin{proof}
The proof is by contradiction. If there is a $\poly(n)$-natural proof for every tensor $t$ for almost all $n$,
then the corresponding sequence of minrank varieties is $p$-definable. Since each minrank variety
can be written as an orbit closures, where the groups are triples of general linear groups,
by Lemma~\ref{lem:orbit}, the minrank varieties are also uniformly generated. 
Therefore, by Theorem~\ref{thm:hardproof}, $\coNP \subseteq \exists \BPP$.
\end{proof}

\begin{remark}
The result above can also be extended to the slice rank varieties. Since each
of them can be written as a polynomial union of orbit closures,
instead of testing whether the circuit $C$ in the proof of Theorem~\ref{thm:hardproof}
vanishes on one dense subset, we test whether it vanishes on polynomially many 
dense subsets. 
\end{remark}

%\textbf{Write something about the result by Heintz?}

%\section{Concluding remarks}

% While we found various nontrivial equations for minrank rank varieties, we essentially proved 

%\textbf{It would be interesting to find explicit families of tensors with high minrank for which representation-theoretic methods prove lower bounds not obtainable by simpler techniques.}

%\textbf{Open question: find an explicit family of high minrank tensors $T$ and family of representations $\lambda$ such that highest weight vectors of $\lambda$ separate $T$ from the low minrank tensors, but there is no algebraic proof of polynomial size.}

\nocite{DBLP:conf/icalp/GrochowMQ16}

\bibliographystyle{plainurl}
\bibliography{minrank,br-mc,slice_rank}

\end{document}